\newif\ifsubm\submfalse

\documentclass[a4paper,english,11pt]{scrartcl}
\usepackage{authblk}
\pdfoutput=1
\usepackage{tikz}
\usepackage[utf8]{inputenc}
\usepackage[T1]{fontenc}

\usepackage{mathtools,adjustbox}				% bessere Variante von amsmath (das von diesem Paket geladen wird). Zahlreiche Mathesachen
\usepackage{amssymb}				% zahlreiche mathematische Symbole
\usepackage{mathrsfs, stmaryrd}		% Zusätzliche Mathesymbole

\usepackage{amsthm}			% Theorem- und Beweisumgebungen

\usetikzlibrary{chains,shapes.multipart}
\usetikzlibrary{shapes,calc,patterns}
\usetikzlibrary{automata,positioning}

\usetikzlibrary{arrows, decorations.markings}
\definecolor{myred}{RGB}{220,43,25}
\definecolor{mygreen}{RGB}{0,146,64}
\definecolor{myblue}{RGB}{0,143,224}

\usetikzlibrary{shapes}

\usepackage{pgfplots}
\pgfplotsset{compat=newest} 

\usepgfplotslibrary{fillbetween}

\tikzstyle{vertex} = [shape=circle,draw=black]
\tikzstyle{namedVertex} = [shape=circle,draw=black]
\tikzstyle{edge} = [draw,->,thick]
\tikzstyle{labeledNodeS}=[circle, color=black!75!white, draw, inner sep = 0.1em, minimum size = 1.5em, scale=1.25]
\tikzstyle{normalEdge}=[very thick, >=stealth]

\usepackage[linesnumbered,ruled,nokwfunc,longend]{algorithm2e}
\DontPrintSemicolon

\usepackage{listings}
\usepackage{color}
\definecolor{purplekeywords}{rgb}{0.5,0,0.33}
\definecolor{greencomments}{rgb}{0,0.5,0}
\definecolor{bluestrings}{rgb}{0,0,1}
\definecolor{types}{rgb}{0.17,0.57,0.68}
\usepackage[bookmarks=false,colorlinks=true, linkcolor=blue, urlcolor=blue, citecolor=blue, breaklinks=true]{hyperref}
\usepackage{cleveref}			% Referenzen mit Name

\crefname{cons}{constraint}{constraints}
\Crefname{cons}{Constraint}{Constraints}
\creflabelformat{cons}{(#2#1#3)}
\crefname{claim}{claim}{claims}
\Crefname{claim}{Claim}{Claims}

\theoremstyle{definition}
\newtheorem{defn}{Definition}[section]

\theoremstyle{plain}

\newtheorem{cor}[defn]{Corollary}
\newtheorem{lemma}[defn]{Lemma}
\newtheorem{theorem}[defn]{Theorem}

\newtheorem{corollary}[defn]{Corollary}
\newtheorem{claim}{Claim}
\newenvironment{proofClaim}[1][]{\ifthenelse{\equal{#1}{}}{\begin{proof}}{\begin{proof}[#1]}}{\end{proof}}

\newtheorem{innerapplemma}{Lemma}
\newenvironment{applemma}[1]
{\renewcommand\theinnerapplemma{#1}\innerapplemma}
{\endinnerapplemma}

\theoremstyle{remark}
\newtheorem{remark}[defn]{Remark}
\newtheorem{obs}[defn]{Observation}

\usepackage{array}					% Verbesserte Implementierung der tabular und array-Umgebungen

\usepackage{braket}					% \Set{ ... | ... } und \set{ ... | ... }
\usepackage{bm}						% fette Symbole in Math-Umgebung (\bm)
\usepackage{nicefrac} 				% TODO: Schönere Brüche?

\allowdisplaybreaks					% erlaubt Seitenumbrüche in align*-Umgebung (aber nicht align, gather, ...)

\newcommand{\IN}{\mathbb{N}}	% Natürliche Zahlen
\newcommand{\INo}{\IN_0}		%	mit Null
\newcommand{\INs}{\IN^\ast}		%	ohne Null
\newcommand{\IR}{\mathbb{R}}
\newcommand{\R}{\mathbb{R}}	% Reelle Zahlen
\newcommand{\BigO}{\mathcal{O}}

\newcommand{\ceil}[1]{\left\lceil#1\right\rceil}
\newcommand{\floor}[1]{\left\lfloor#1\right\rfloor}
\newcommand{\abs}[1]{\left|#1\right|}

\newcommand{\noqed}{\let\qed\relax}

\usepackage{comment}			% Umgebungen für Textabschnitte die nur in bestimmten Versionen erscheinen
\usepackage{textcomp}			% Textsymbole

\usepackage{rotating}

\usepackage{setspace}

\usepackage[font={small,it}]{caption} % Kleinere Captions

\usepackage{transparent}
\usepackage{color}
\usepackage{graphicx}

\usepackage[all]{xy}
\usepackage{lmodern}
\usepackage{tabto}

\usepackage[babel]{csquotes}

\usepackage{framed,enumitem}

\setlist[itemize]{leftmargin=1.4em}
\setlist{nosep}

\newcommand{\tauMinDiff}{\tau_\Delta}
\newcommand{\vol}[3]{\mathrm{vol}_{#1}(#2,#3)}
\newcommand{\inflowOver}[2]{\int_{#1}^{#2}f_e^-(\theta)d\theta}
\newcommand{\outflowOver}[2]{\int_{#1}^{#2}f_e^+(\theta)d\theta}
\newcommand{\distP}[1]{\tilde{d}_{#1}}
\newcommand{\tPmax}{\tau(P_{\max})}
\newcommand{\tPmin}{\tau(P_{\min})}

\newcommand{\edgeLoad}[1][]{%
	\ifthenelse{\equal{#1}{}}
	{F^{\Delta}}
	{F^{\Delta}_{#1}}
}

\newcommand{\pClosEL}[1]{F_{\leq #1}}

\newcommand{\termTime}[1][]{%
	\ifthenelse{\equal{#1}{}}
	{\Theta}
	{\Theta_{\mathrm{#1}}}
}

\newcommand{\PoA}{\ensuremath{\mathrm{PoA}}}

\newcommand{\bigO}{\mathcal{O}}
\newcommand{\bigOm}{\Omega}

\newcommand{\network}{\mathcal{N}}

\usepackage{dsfont} 
\newcommand{\CharF}[1][]{%
	\ifthenelse{\equal{#1}{}}
	{\mathds{1}}
	{\mathds{1}_{#1}}
}

\newcommand{\edgesLeaving}[1]{\delta^+_{#1}}
\newcommand{\edgesEntering}[1]{\delta^-_{#1}}

\usepackage[left=3cm, right=3cm, top=3cm, bottom=4cm]{geometry}

\usepackage[numbers]{natbib}
\usepackage{etoolbox}
\makeatletter
\pretocmd{\NAT@citexnum}{\@ifnum{\NAT@ctype>\z@}{\let\NAT@hyper@\relax}{}}{}{}
\makeatother

\newcommand{\sumOfTraveltimes}[1][]{%
	\ifthenelse{\equal{#1}{}}
	{\Xi}
	{\Xi_{\mathrm{#1}}}
}

\newcommand{\MPoA}{\ensuremath{{\operatorname{M-PoA}}}}
\newcommand{\TPoA}{\ensuremath{{\operatorname{T-PoA}}}}

\title{The Price of Anarchy \\ for Instantaneous Dynamic Equilibria}

\author{Lukas Graf}
\author{Tobias Harks}
\affil{\small Augsburg University, Institute of Mathematics, 86135 Augsburg\\
\href{mailto:lukas.graf@math.uni-augsburg.de}{\{\texttt{lukas.graf,tobias.harks\}@math.uni-augsburg.de}}}

\usepackage{subcaption}
\usepackage{longtable}

\begin{document}
\maketitle
\begin{abstract}
	We consider flows over time within the deterministic queueing model of Vickrey
	and study the solution concept of instantaneous dynamic equilibrium (IDE) in which flow particles select at every decision point
	a currently shortest path. 
	The length of such a path is measured
	by the physical travel time plus the time spent in queues. 
	Although IDE have been studied since 
	the eighties, the worst-case efficiency of the solution concept with respect to the travel time is not well understood.
	We study the  price of anarchy for this model under the makespan and total travel time objective
	and show the first known upper bound for both measures of order $\bigO(U\cdot \tau)$ for single-sink instances, where $U$ denotes the total inflow volume and $\tau $ the sum
	of edge travel times. We complement this upper bound with
	a family of quite complex instances proving a lower bound of order $\bigOm(U\cdot \log \tau)$.
\end{abstract}

%!TEX root = ../article-PoA.tex

\section{Introduction}
Dynamic flows have gained substantial interest
over the last decades in modeling dynamic network systems
such as urban traffic  systems. As one of the earliest works in this area,  Ford and Fulkerson~\cite{Ford62} proposed flows over time as a generalization of static flows incorporating a time component. Shortly after, Vickrey~\cite{Vickrey69} introduced 
a game-theoretic variant using a fluid queueing model. This model -- known as the Vickrey bottleneck model -- is arguably one of the most important traffic models to date  (see Li, Huang and Yang~\cite{LI2020} for a research overview of the past 50 years).
In the network-variant of the fluid queueing model, there is a directed graph $G=(V,E)$, where edges
$e\in E$ are associated with a queue with positive rate capacity $\nu_e\in\IR_+$ and a physical transit time $\tau_e\in \IR_+$.
If the total inflow into an edge  $e=vw\in E$ exceeds the rate capacity
$\nu_e$, a queue builds up and arriving flow particles need to 
wait in the queue before they are forwarded along the edge (cf. \Cref{fig:QueueingModel}). The total travel time
along $e$  is thus composed of the waiting time spent in the queue plus the physical transit time $\tau_e$.

\begin{figure}[bh!]\centering
	\begin{adjustbox}{max width=.3\textwidth}
		\begin{tikzpicture}
	\newcommand{\flowColor}{red!30}

	\node[namedVertex,opacity=0] (v) at (0,0) {$v$};
	\node[namedVertex,opacity=0] (w) at (3,0) {$w$};

	\draw[line width=.2cm,\flowColor] ($(v.150)+(.01,.2)$) -- +(-1,.5);
	\draw[line width=.2cm,\flowColor] ($(v.150)+(0,.2)$) to[out=330,in=180] ($(v.east)+(0,.2)$);
	\draw[line width=.4cm,\flowColor] ($(v.210)+(.01,-.1)$) -- +(-1,-.5);
	\draw[line width=.4cm,\flowColor] ($(v.210)+(0,-.1)$) to[out=30,in=180] ($(v.east)+(0,-.1)$);
	
	\draw[line width=.3cm,\flowColor] (v) -- (w);
	\draw[line width=.6cm,\flowColor] ($(v)+(.69,.2)$) -- +(0,.62);
	
	\draw[line width=.3cm,\flowColor] (w.west) to[out=0,in=150] (w.330);
	\draw[line width=.3cm,\flowColor] ($(w.330)+(-.01,0)$) -- +(1,-.5);
	
	\draw[edge,thick] (v) -- (w);
	
	\draw[edge,<-] ($(v.150)+(.13,.13)$) -- +(-1.13,.57);
	\draw[edge,<-] ($(v.210)+(.05,-.07)$) -- +(-1.05,-.53);
	
	\draw[edge] (w.30) -- +(1,.5);
	\draw[edge] (w.330) -- +(1,-.5);
	
	\node () at ($(v)+(.7,1.1)$){$q_e(\theta)$};
	
	\draw[<->] (2,-.15cm) -- +(0,.3cm);
	\node () at (2,.4) {$\nu_e$};
	
	\draw[<->] ($(v.east)+(0,-.4)$) -- ($(w.west)+(0,-.4)$);
	\node () at ($(v)!.5!(w) + (0,-.6)$) {$\tau_e$};
	
	\node[namedVertex,fill=white,fill opacity=.5] () at (v) {$v$};
	\node[namedVertex,fill=white,fill opacity=.5] () at (w) {$w$};
\end{tikzpicture}
	\end{adjustbox}
	\caption{An edge $e=vw$. As the inflow rate at node $v$ exceeds the edge's capacity, a queue forms at its tail.}\label{fig:QueueingModel}
\end{figure}
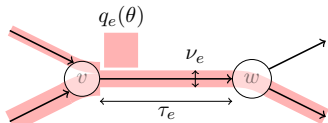

For obtaining a predictive model of traffic flows,  the physical flow model needs to be enhanced by 
a \emph{behavioral model} prescribing the actions of flow particles.  
Most works in the transportation science literature as well as recent
works in the mathematics and computer science literature 
adopt the \emph{full information model}, i.e., all flow particles have complete
information on the state of the network for all points in time
(including the  future evolution of all flow particles) and based on this information travel along a shortest path. This  leads to the concept of \emph{dynamic equilibrium} (Nash equilibrium) and has been analyzed in the transportation science literature for decades, see \citeauthor*{Friesz93}~\cite{Friesz93}, \citeauthor*{MeunierW10}~\cite{MeunierW10}, \citeauthor*{ZhuM00}~\cite{ZhuM00}
and the more recent works by \citeauthor*{Koch11}~\cite{Koch11} and \citeauthor*{CominettiCL15}~\cite{CominettiCL15}. 
The full information assumption has been justified by assuming that the game
is played
repeatedly and a dynamic equilibrium is then an attractor of a learning process. 
In light of the
wide-spread use of navigation devices, this concept may not be completely realistic anymore, because  drivers are informed in
real-time about the current traffic situation and, if beneficial,
reroute instantaneously no matter how good or bad that route
was in hindsight. This aspect is also discussed by
\citeauthor*{Marcotte04}~\cite{Marcotte04}, \citeauthor*{Hamdouch2004}~\cite{Hamdouch2004} and \citeauthor*{UnnikrishnanW09}~\cite{UnnikrishnanW09}. 

Instead of the (classical) dynamic equilibrium,
we consider in this paper \emph{instantaneous dynamic equilibria (IDE)},
where for every point in time and at every decision node, flow only enters
those edges that lie on a currently shortest path towards the respective sink.
This concept assumes far less information (only the network-wide queue length which are continuously measured) and leads to a decentralized
dynamic using only present information that is readily available via real-time information.  IDE have been proposed  already in the late 80's 
(cf.~Ran and Boyce~\cite[\S~VII-IX]{Ran96}, Boyce, Ran and LeBlanc~\cite{BoyceRL95,RanBL93}, Friesz et al.~\cite{FrieszLTW89}) and it is known that IDE do exist in the Vickrey model under quite general conditions, see Graf, Harks and Sering~\cite{GHS18}.

\subsection{Motivating Example and Price of Anarchy}
In comparison to the dynamic equilibrium of the full information model, an IDE flow behaves quite differently and several fundamental aspects of IDE are not well understood.
There are, for instance,  simple  single-sink instances in which the unique IDE flow exhibits cycling behavior, that is, some flow particles travel along cycles
before they reach the sink.  Let us give an example taken from Graf, Harks and Sering~\cite{GHS18} illustrating IDE and their
possible cycling behavior. 
\begin{figure}[t!]
	\begin{center}
		\begin{adjustbox}{max width=1\textwidth}
			\begin{tikzpicture}[scale=1.2]

	\newcommand{\exampleGraph}[2]{		
		\node[labeledNodeS, fill=white] (s1) at (0, 4) {$s_1$};
		\draw (s1) ++(4, 0) node[labeledNodeS, fill=white] (v) {$v$};
		\draw (s1) ++(0, -2) node[labeledNodeS, fill=white] (t) {$t$};
		\draw (v) ++(0, -2) node[labeledNodeS, fill=white] (s2) {$s_2$};
		
		\draw[normalEdge,->] (s1) -- node[below] {\ifthenelse{\equal{#2}{labels}}{\Large$(1,2)$}{}} (v);
		\draw[normalEdge,->] (s1) -- node[right] {\ifthenelse{\equal{#2}{labels}}{\Large$(3,1)$}{}} (t);
		\draw[normalEdge,->] (v)  -- node[left] {\ifthenelse{\equal{#2}{labels}}{\Large$(1,2)$}{}} (s2);
		\draw[normalEdge,->] (s2) -- node[below] {\ifthenelse{\equal{#2}{labels}}{\Large$(\tau_{s_2t},\nu_{s_2t})=(1,1)$}{}} (t);
		\draw[normalEdge,->] (s2) -- node[below,sloped] {\ifthenelse{\equal{#2}{labels}}{\Large$(1,1)$}{}} (s1);
		
		\node at (2,5) [rectangle,draw] (theta0) {\Large$\theta=#1$:};
	}

	\newcommand{\colComRed}{red!30}
	\newcommand{\colComBlue}{blue!60}
	
	% theta = 0 #########################################################################
	\begin{scope}
	
		\exampleGraph{0}{labels}
		
		\draw (s1) +(0,1) node[rectangle, draw, minimum width=1.5cm, minimum height=.8cm,fill=\colComRed] (s1u) {\large $u_{s_1}\equiv3\cdot\CharF[[0,1)]$}; 
		\draw [->, line width=2pt] (s1u) -- (s1);

	\end{scope}

	% theta = 1 #########################################################################
	\begin{scope}[xshift=5.2cm]

		\node[rectangle, color=\colComRed,draw, minimum width=.66cm,minimum height=.2cm, fill=\colComRed,rotate around={90:(0,0)}] at (0,3.6) {};
		\node[rectangle, color=\colComRed,draw, minimum width=4.2cm,minimum height=.4cm, fill=\colComRed] at (2,4) {};
	
		\exampleGraph{1}{}	
	
		\draw (s2) ++(0,-1) node[rectangle, draw, minimum width=1.5cm, minimum height=.8cm,fill=\colComBlue] (s2u) {\large $u_{s_2}\equiv4\cdot\CharF[[1,2)]$}; 
		\draw [->, line width=2pt] (s2u) -- (s2);
	
	\end{scope}

	% theta = 2 #########################################################################
	\begin{scope}[xshift=10.4cm]

		\node[rectangle, color=\colComRed,  draw, minimum width=.66cm,minimum height=.2cm, fill=\colComRed,rotate around={90:(0,0)}] at (0,3) {};
		\node[rectangle, color=\colComRed,  draw, minimum width=1.8cm,minimum height=.4cm, fill=\colComRed,rotate around={90:(0,0)}] at (4,3) {};
		\node[rectangle, color=\colComBlue, draw, minimum width=4.2cm,minimum height=.2cm, fill=\colComBlue] at (2,2) {};
		\node[rectangle, color=\colComBlue, draw, minimum width=.5cm,minimum height=1.5cm, fill=\colComBlue] at (3.5,1.2) {};
		\node at (2,1.2) [] () {\Large $q_{s_2t}(2)=3$};
		
		\exampleGraph{2}{}
	
	\end{scope}

	% theta = 3 #########################################################################
	\begin{scope}[xshift=15.6cm]

		\node[rectangle, color=\colComRed,draw, minimum width=.66cm,minimum height=.2cm, fill=\colComRed,rotate around={90:(0,0)}] at (0,2.4) {};
		\node[rectangle, color=\colComRed,draw, minimum width=4.7cm,minimum height=.2cm, fill=\colComRed,rotate around={-26.5:(0,0)}] at (2,3) {};
		\node[rectangle, color=\colComBlue,draw, minimum width=4.2cm,minimum height=.2cm, fill=\colComBlue] at (2,2) {};
		\node[rectangle, color=\colComBlue,draw, minimum width=.5cm,minimum height=1cm, fill=\colComBlue] at (3.5,1.4) {};
		\node[rectangle, color=\colComRed,draw, minimum width=.5cm,minimum height=.5cm, fill=\colComRed] at (3.5,.8) {};
		\node at (2,1.2) [] () {\Large $q_{s_2t}(3)=3$};
	
		\exampleGraph{3}{}
		
	\end{scope}

	% theta = 4 #########################################################################
	\begin{scope}[xshift=20.8cm]
		
		\node[rectangle, color=\colComRed,draw, minimum width=.66cm,minimum height=.2cm, fill=\colComRed,rotate around={90:(0,0)}] at (0,3.6) {};

		\node[rectangle, color=\colComBlue,draw, minimum width=4.2cm,minimum height=.2cm, fill=\colComBlue] at (2,2) {};
		\node[rectangle, color=\colComBlue,draw, minimum width=.5cm,minimum height=.5cm, fill=\colComBlue] at (3.5,1.6) {};
		\node[rectangle, color=\colComRed,draw, minimum width=.5cm,minimum height=.5cm, fill=\colComRed] at (3.5,1.2) {};
		\node at (2,1.2) [] () {\Large $q_{s_2t}(4)=2$};
		
		\exampleGraph{4}{}
		
	\end{scope}	
\end{tikzpicture}
		\end{adjustbox}
	\end{center}
	\vspace{-0.5cm}
	\caption[format=hang]{The evolution of an IDE flow over the time horizon $[0,4]$. At time $\theta =4$ all remaining particles are on one of the edges leading directly towards the sink $t$. Thus, the last particles will arrive at the sink three time steps later, i.e. at time $\theta=7$. In contrast the optimal makespan would be $\theta=6$, which can be achieved by only sending flow along the direct edges $s_1t$ and $s_2t$ (note, that the maximal inflow rate into sink $t$ is $0$ up to time $2$, then $1$ up to time $3$ and finally $2$ after that -- thus time $\theta=6$ is the earliest point at which the total flow volume of $7$ may have reached the sink). The price of anarchy with respect to the makespan for this example is therefore $7/6$.}
	\label{fig:ide}
\end{figure}
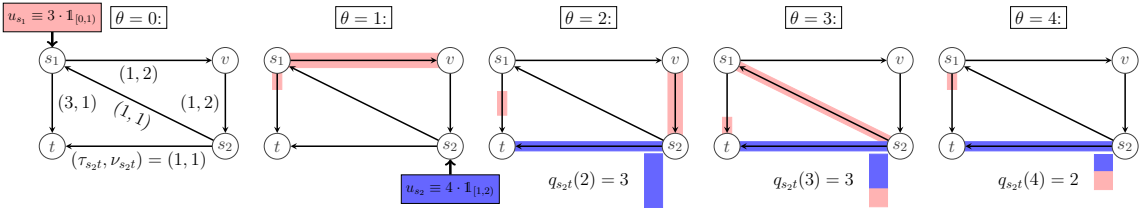
\begin{figure}[!ht]
	\begin{center}
		\begin{adjustbox}{max width=.9\textwidth}
			{\newcommand{\colComRed}{red!30}
\newcommand{\colComBlue}{blue!60}
\begin{tikzpicture}[scale=.8,solid,black,
		declare function={
			f(\x)= and(\x>=0, \x<1) * 3 * \x + and(\x>=1, \x<3) * 3 + and(\x>=3, \x<6) * (3-(\x-3));	
			g(\x)= f(\x) + and(\x>=1, \x<2) * 4 * (\x-1) + and(\x>=2, \x<6) * (4-(\x-2));			
		}
		]

		\begin{axis}[xmin=0,xmax=7.5,ymax=7.5, ymin=0, samples=500,width=12cm,height=7cm,
			axis x line*=bottom]
			\addplot[name path=f, blue, ultra thick,domain=0:8] {g(x)};
			\addplot[name path=g, red, draw=none,domain=0:8] {f(x)};
			\path[name path=xaxis] (axis cs:0,0) -- (axis cs:7,0);
			%\addplot[pattern=horizontal lines, pattern color=\colComBlue]fill between[of=f and g, soft clip={domain=0:7}];
			\addplot[fill=\colComBlue]fill between[of=f and g, soft clip={domain=0:7}];
			\addplot[fill=\colComRed]fill between[of=g and xaxis, soft clip={domain=0:7}];
		\end{axis}

		\node[blue]() at (1,3.5) {\large$F^\Delta$};

\end{tikzpicture}
\hspace{2em}
\begin{tikzpicture}[scale=.8,solid,black,
	declare function={
		f(\x)= and(\x>=0, \x<1) * 3 * \x + and(\x>=1, \x<3) * 3 + and(\x>=3, \x<4) * (3-(\x-3)) + and(\x>=4, \x< 6)*2 + and(\x>=6, \x < 7) * (2 - 2*(\x-6));	
		g(\x)= f(\x) + and(\x>=1, \x<2) * 4 * (\x-1) + and(\x>=2, \x<6) * (4-(\x-2));			
	}
	]

	\begin{axis}[xmin=0,xmax=7.5,ymax=7.5, ymin=0, samples=500,width=12cm,height=7cm,
		axis x line*=bottom]
		\addplot[name path=f, blue, ultra thick,domain=0:8] {g(x)};
		\addplot[name path=g, red, draw=none,domain=0:8] {f(x)};
		\path[name path=xaxis] (axis cs:0,0) -- (axis cs:7,0);
		%\addplot[pattern=horizontal lines, pattern color=\colComBlue]fill between[of=f and g, soft clip={domain=0:7}];
		\addplot[fill=\colComBlue]fill between[of=f and g, soft clip={domain=0:7}];
		\addplot[fill=\colComRed]fill between[of=g and xaxis, soft clip={domain=0:7}];
	\end{axis}
	
	\node[blue]() at (1,3.5) {\large$F^\Delta$};
\end{tikzpicture}
}
		\end{adjustbox}
	\end{center}
	\vspace{-0.5cm}
	\caption[format=hang]{The total flow volume $F^\Delta$ present in the network at different times for the optimal solution (left) and the IDE solution (right) described in \Cref{fig:ide}. The total shaded area corresponds to the sum of travel times in each of the flows (cf. \Cref{lemma:SumOfTraveltimesAltDef}), which is $22$ and $25$, respectively. Thus, the price of anarchy with respect to total travel times is $25/22$ in this instance.}
	\label{fig:exampleVolumeGraph}
\end{figure}
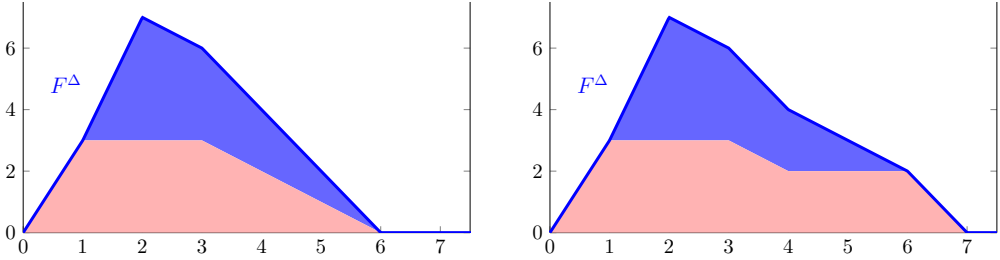

Consider the network in \Cref{fig:ide} (left). There are two source nodes $s_1$ and~$s_2$ with constant inflow rates $u_1(\theta) \equiv 3$ for times $\theta \in [0,1)$ and $u_2(\theta)\equiv 4$ for $\theta\in [1,2)$. Commodity $1$ (red) has two simple paths connecting $s_1$ with the sink $t$. Since both have equal length ($\sum_e \tau_e =3$), in an IDE both can be used by commodity $1$. In \Cref{fig:ide},  the flow takes the direct edge to $t$ with a rate of one, while edge $s_1v$ is used at a rate of two. This is actually the only split possible in an IDE, since any other split (different in more than just a subset of measure zero of $[0,1)$) would result in a queue forming on one of the two edges, which would make the respective path longer than the other one.
At time $\theta=1$, the inflow at $s_1$ stops and a new inflow of commodity 2 (blue) at $s_2$ starts. This new flow again has two possible paths to $t$, however, here the direct path ($\sum_e \tau_e =1$) is shorter than the alternative ($\sum_e \tau_e =4$). So all flow enters edge $s_2t$ and starts to form a queue. 
At time $\theta= 2$, the first flow particles of commodity $1$ arrive at $s_2$ with a rate of $2$. Since the flow of commodity $2$ has built up a queue of length $3$ on edge $s_2t$ by this time, the estimated travel times $\sum_e (\tau_e + q_e(\theta))$ are the same on both simple $s_2$-$t$ paths. Thus, the red flow is split evenly between both possible paths. This results in the queue-length on edge $s_2t$ remaining constant and therefore this split gives us an IDE flow for the interval~$[2,3)$.
At time $\theta=3$, red particles will arrive at $s_1$ again, thus, completing a full cycle (namely $s_1,v,s_2,s_1$). This example shows that IDE flows may involve a flow decomposition along cycles.  Note, that cycles can appear even in the case of only a single commodity (and therefore a single source and sink) and a constant inflow rate over a single interval, see Graf, Harks and Sering~\cite{GHS18} for details.
In contrast to this cycling behavior, the (classical) dynamic equilibrium flow will just send more of the red flow along the direct path $s_1,t$ since the future queue growth at edge $s_2t$ of the alternative path is already anticipated. For the price of anarchy with respect to the makespan, we need to specify
the latest arrival time (makespan) of a system optimal dynamic flow (minimizing the makespan).
Since the maximal inflow rate into sink $t$ is $0$ up to time $2$,  $1$ up to time $3$ and finally $2$ for all later times, the time $\theta=6$ is the earliest point in time at which the total flow volume of $7$ may have reached the sink.
Thus, the price of anarchy with respect to the makespan for this example is exactly $7/6$. Similarly, we can determine the price of anarchy with respect to the total travel times to be exactly $25/22$ here (see \Cref{fig:exampleVolumeGraph}). This example motivates the question of the price of anarchy of IDE flows.\\

\textbf{Question} (PoA): Assuming single-sink instances with constant inflow rates for a finite time interval, what is the maximum time (makespan) needed so that every flow particle reaches the sink?
Alternatively, what is the maximum total travel time for all flow particles?
As is common in the price of anarchy literature, we will compare the makespan and total travel time of an IDE solution with that of a system optimal solution minimizing the respective objective function (makespan or total travel time). 
Formally,  the price of anarchy is defined as the 
worst-case ratio  of the respective objective function value evaluated at an IDE flow and that of a respective system optimal dynamic flow.

\subsection{Our Results and Proof Techniques}
We study the price of anarchy (PoA) of IDE flows
for single-sink instances and derive the first quantitative upper bound
for the makespan and total travel time objective. Our bound is parameterized
in the numbers $U$ and  $\tau$ denoting the total flow volume injected into
the sources and  the sum of physical travel
times, respectively. We denote by  $\MPoA(U,\tau)$ and $\TPoA(U,\tau)$
the price of anarchy over the family of instances parameterized by $U$ and $\tau$
for the makespan and total travel time objective, respectively.
Our main result reads as follows:
\begin{framed}\noindent
Thm.~\ref{thm:Termination_SingleSink2}/Thm.~\ref{thm:PoAUppperBound}:
	For multi-source single-sink networks, any IDE flow over time terminates after at most $\BigO(U\tau)$ time. 
	Moreover, for both objectives we obtain $\MPoA(U,\tau),\TPoA(U,\tau) \in \bigO(U \tau)$.\footnotemark	
	\end{framed}
\footnotetext{Here, $\bigO(U \tau)$ denotes the set of all functions in $U$ and $\tau$ growing asymptotically at most as fast as $U\tau$, i.e. all functions $f: \INs \times \INs \to \IR_{\geq 0}$ such that there exist some $k, U_0,\tau_0 \in \INo$ with $f(U,\tau) \leq k\cdot U\tau$ for all $U \geq U_0, \tau \geq \tau_0$.}
We prove this bound by first deriving a general
termination bound for acyclic graphs.
Using this bound, we then show that 
there exist so-called sink-like subgraphs 
that can effectively be treated as acyclic graphs.
This way, we can argue that at all times a sufficiently large
flow volume enters the current sink-like subgraph and, by the bound for acyclic graphs, reaches the sink within the claimed time. \Cref{fig:LogicalStructureSec3} gives an overview over the logical structure of the complete proof of \Cref{thm:Termination_SingleSink2,thm:PoAUppperBound}. 
The proof techniques and the bound itself are completely
different to those for dynamic equilibria in Correa, Cristi, and Oosterwijk~\cite{CCO19}
or for static Wardrop equilibria (see \citeauthor*{Rough05}~\cite{Rough05}, \citeauthor*{CorreaSM04}~\cite{CorreaSM04} and \citeauthor*{Watling16}~\cite{Watling16}). 

\begin{figure}
	\begin{adjustbox}{max width=\textwidth}
		{
	\Crefname{lemma}{Lem.}{Lem.}
	\Crefname{corollary}{Cor.}{Cor.}
	\Crefname{theorem}{Thm.}{Thm.}

	\tikzstyle{Thm} = [rectangle, draw, ultra thick, text width=30em, text badly centered]
	\tikzstyle{Lem} = [rectangle, draw, text width=17em]
	\tikzstyle{Cor} = [rectangle, draw, text width=17em]
	\tikzstyle{Claim} = [text width=17em]
	\tikzstyle{extern} = [fill=gray!20]

	\tikzstyle{implies} = [draw, -latex', double]
	\tikzstyle{line2} = [draw]
	
	\small	
	\begin{tikzpicture}[node distance = 5em, auto]
		
		\node[Thm] (PoAThm) {\Cref{thm:PoAUppperBound}: $\MPoA(U,\tau),\TPoA(U,\tau) \in \bigO(U\tau)$};
		
		\node[Thm, below of=PoAThm, node distance=4em] (UpperBoundThm) {\Cref{thm:Termination_SingleSink2}: Upper bound on IDE makespan and total travel times};
		
		\node[below of=UpperBoundThm] (belowUpperBoundThm) {};
		
		% Left side
		\node[Cor, left of=belowUpperBoundThm, node distance=12em] (CorSinkLikeTerm) {\Cref{cor:TerminationIfGIsSinkLike}: Upper bound for termination time in sink-like graphs};
		
		\node[Lem, below of=CorSinkLikeTerm] (TermBoundAcyclicLem) {\Cref{lemma:TerminationAcyclicNetworksV2}:\textsuperscript{(V)} Upper bounds for termination time in acyclic networks};
		\node[below of=TermBoundAcyclicLem, node distance=4em] () {
			\begin{adjustbox}{width=14em}
				\large
				\begin{tikzpicture}
	\node[color=gray]() at (-1,-2) {$\tilde{d}:$};
	\draw[color=gray,ultra thick, dashed] (0,-1.5) -- (0,1.5);	
	\node[color=gray]() at (0,-2) {$5$};
	\draw[color=gray,ultra thick, dashed] (2,-1.5) -- (2,1.5);	
	\node[color=gray]() at (2,-2) {$4$};
	\draw[color=gray,ultra thick, dashed] (4,-1.5) -- (4,1.5);	
	\node[color=gray]() at (4,-2) {$3$};
	\draw[color=gray,ultra thick, dashed] (6,-1.5) -- (6,1.5);	
	\node[color=gray]() at (6,-2) {$2$};
	\draw[color=gray,ultra thick, dashed] (8,-1.5) -- (8,1.5);	
	\node[color=gray]() at (8,-2) {$1$};
	\draw[color=gray,ultra thick, dashed] (10,-1.5) -- (10,1.5);	
	\node[color=gray]() at (10,-2) {$0$};

	\node[vertex,fill=white] (1) at (0,0) {};
	\node[vertex,fill=white] (2) at (2,1) {};
	\node[vertex,fill=white] (2') at (2,-1) {};
	\node[vertex,fill=white] (3) at (4,0) {};
	\node[vertex,fill=white] (4) at (6,0) {};
	\node[vertex,fill=white] (5) at (8,1) {};
	\node[vertex,fill=white] (5') at (8,-1) {};
	\node[namedVertex,fill=white] (t) at (10,0) {$t$};

	\draw[edge,ultra thick] (1)	-- (2);
	\draw[edge,ultra thick] (2)	-- (3);
	\draw[edge] (2')	-- (3);
	\draw[edge] (2')	-- (5');
	\draw[edge,ultra thick] (3)	-- (4);
	\draw[edge,ultra thick] (4)	-- (5');
	\draw[edge] (4)	-- (t);
	\draw[edge] (5)	-- (t);
	\draw[edge,ultra thick] (5')	-- (t);
\end{tikzpicture}
			\end{adjustbox}
		};		
	
		\node[Lem, below of=TermBoundAcyclicLem, node distance=10em] (OnlyShortestPathsActiveLem) {\Cref{lemma:InSinkLikeOnlyShortestPathsAreActive}: In sink-like subgraphs only physically shortest paths are active};
		
		\node[Lem, extern, below of=OnlyShortestPathsActiveLem] (SmallVolumeDoesntDivertLem) {\cite[Lem. 4.4]{GHS18}: Minimum required volume to divert flow away from physically shortest paths.};
		
		\node[Lem, below of=SmallVolumeDoesntDivertLem] (FlowLeavesEdgesFastLem) {\Cref{lemma:FlowOnEdgesLeavesFastVar}:\textsuperscript{(V)} Upper bound on residence time of flow on edges};
		
		% Right side
		\node[Claim, right of=belowUpperBoundThm, node distance=12em] (Claim5) {\underline{\Cref{thm:Termination_SingleSink2:Claim}}: Inflow $<\nicefrac{1}{2}$ into sink for $\approx [\theta,\theta+\tau(G)]$ then $G$ is sink-like at $\theta$};
		
		\node[Lem, below of=Claim5] (extendingSinkLikeLemma) {\Cref{lemma:ExtensionOfSinkLikeSubgraphs}: $T$ sink-like for $[\theta_1,\theta_2]$ then $T\cup\set{v}$ is sink-like for $[\theta_1,\theta_2-\delta]$};
		\node[below of=extendingSinkLikeLemma, node distance=4.8em] () {
			\begin{adjustbox}{width=11em}
					\newcommand{\colInT}{green!70!blue!70}
					\newcommand{\colVtoT}{red!50!blue!80}
					\newcommand{\colVtoG}{black}
					\newcommand{\colGtoV}{blue!70}
					\newcommand{\colGtoT}{blue!30}
					\newcommand{\colTtoV}{red}
					\newcommand{\colTtoG}{black}
				\large
				\newcommand{\colTbackground}{black!20!white}

\tikzset{
	styleVtoT/.style={edge,dashed,\colVtoT},
	styleTtoV/.style={edge,dashdotted,\colTtoV},
	styleInT/.style={edge,\colInT},
	styleGtoT/.style={edge,loosely dotted,\colGtoT},
	styleGtoV/.style={edge,dotted,\colGtoV},
	styleTtoG/.style={edge,loosely dotted,\colTtoG},
	styleVtoG/.style={edge,loosely dotted,\colVtoG}
}

\begin{tikzpicture}[scale=1.0]	
	\node () at (1,4) {\Large$G$};
	\fill[rounded corners=10, color=\colTbackground] (3.5,4.5) rectangle (8.5, -.5);
	\node() at (8,4) {\Large$T$};
	\draw[rounded corners=10, ultra thick, color=\colTbackground] (3.5,4.5) -- (8.5,4.5) -- (8.5, -.5) -- (3.5,-.5) -- (1.2,2) -- cycle;
	\node() at (3,2.5) {\Large$T'$};

	\node[namedVertex, fill=white] (t) at (8, 2) {$t$};
	\node[namedVertex, fill=white] (v) at (2,2) {$v$};
	\node[vertex, fill=white] (1) at (4,4) {};
	\node[vertex, fill=white] (2) at (7,4) {};
	\node[vertex, fill=white] (3) at (5,2) {};
	\node[vertex, fill=white] (4) at (4,0) {};
	\node[vertex, fill=white] (5) at (8,0) {};
	
	\draw[styleVtoT] (v) -- (3);
	\draw[styleVtoT] (v) -- (4);
	\draw[styleVtoG] (v) -- ++(-1,-1);
	\draw[styleGtoV,<-] (v) -- ++(-1.5,0);
	\draw[styleGtoV,<-] (v) -- ++(-1,1);
	
	\draw[styleInT] (1) -- (2);
	\draw[styleInT] (1) -- (3);
	\draw[styleTtoV] (1) -- (v);
	\draw[styleTtoG] (1) -- ++(-.8,.8);
	\draw[styleGtoT,<-] (1) -- ++(-1.5,0);
	
	\draw[styleInT] (2) -- (3);
	
	\draw[styleInT] (3) -- (t);
	
	\draw[styleInT] (4) -- (3);
	\draw[styleGtoT,<-] (4) -- ++(-1.5,0);
	
	\draw[styleInT] (5) -- (4);
	\draw[styleInT] (5) -- (t);

\end{tikzpicture}
			\end{adjustbox}
		};	
		
		\node[Claim, below of=extendingSinkLikeLemma, node distance=10em] (eSLLCl1) {\underline{\Cref{lemma:ExtensionOfSinkLikeSubgraphs:Claim:ReachingVFast}}: Flow from $T$ to $v$ at $\theta_1$ reaches $v$ before $\approx \theta_2 - \nicefrac{\delta}{2}$};
		\node[Claim, below of=eSLLCl1] (eSLLCl2) {\underline{\Cref{lemma:ExtensionOfSinkLikeSubgraphs:Claim:FromVOnlyToT}}: All flow reaching $v$ before $\approx \theta_2 - \nicefrac{\delta}{2}$ continues towards $T$};
		\node[Claim, below of=eSLLCl2] (eSLLCl3) {\underline{\Cref{lemma:ExtensionOfSinkLikeSubgraphs:Claim:ReachingTFast}}: All flow between $v$ and $T$ during $\approx [\theta_1,\theta_2 - \nicefrac{\delta}{2}]$ reaches $T$ by $\theta_2$};
		
		% Paths:
		
		\path [implies] (UpperBoundThm) -- (PoAThm);
		
		\path[implies] (TermBoundAcyclicLem.east) to[bend right=8] (UpperBoundThm);
		\path[implies] (OnlyShortestPathsActiveLem.east) to[bend right=10] (UpperBoundThm);
		
		\path[implies] (CorSinkLikeTerm) -- (UpperBoundThm);
		\path[implies] (TermBoundAcyclicLem) -- (CorSinkLikeTerm);
		\path[implies] (OnlyShortestPathsActiveLem.west) to[bend left=30](CorSinkLikeTerm.west);
		\path[implies] (SmallVolumeDoesntDivertLem) -- (OnlyShortestPathsActiveLem);
		
		\path[implies] (SmallVolumeDoesntDivertLem) -- (eSLLCl2);
		
		\path[implies] (FlowLeavesEdgesFastLem.east) -- (eSLLCl1.west);
		\path[implies] (FlowLeavesEdgesFastLem) -- (eSLLCl3);
			
		\path[implies] (Claim5) -- (UpperBoundThm);
		\path[implies] (extendingSinkLikeLemma) -- (Claim5);
		
		\path[implies] ($0.3*(eSLLCl1.north)+0.7*(eSLLCl1.north west)$) -- ($0.3*(extendingSinkLikeLemma.south)+0.7*(extendingSinkLikeLemma.south west)$);
		\path[implies] (eSLLCl2.north east) to[bend right=20] (extendingSinkLikeLemma.east);
		\path[implies] (eSLLCl3.north east) to[bend right=30] (extendingSinkLikeLemma.east);

		\path[implies] ($0.3*(eSLLCl2.north)+0.7*(eSLLCl2.north west)$) -- ($0.3*(eSLLCl1.south)+0.7*(eSLLCl1.south west)$);
		\path[implies] ($0.3*(eSLLCl3.north)+0.7*(eSLLCl3.north west)$) -- ($0.3*(eSLLCl2.south)+0.7*(eSLLCl2.south west)$);
		\path[implies] (eSLLCl3.north west) to[bend left=30] (eSLLCl1.south west);
	\end{tikzpicture}
}
	\end{adjustbox}
	\caption{The logical structure of \cref{sec:UpperBounds}. Implication arrows are used to indicate which statements are used to prove which other statements. A superscript (V) is used to highlight statements for which their proofs directly use the fact that in the Vickrey bottleneck model queues operate at capacity (i.e. \cref{eq:FeasibleFlow-QueueOpAtCap}). Additionally, the proof of the lemma with grey background (from \cite{GHS18}) depends on the way particles in an IDE predict the future travel times (and, thus, indirectly also depends on the way congestion is modelled). Therefore, these are the main places where one would possibly have to adjust the proofs if one wants to transfer our results to a different flow model.}\label{fig:LogicalStructureSec3}
\end{figure}

We then turn to lower bounds on the price of anarchy for
IDE flows. 
\begin{framed}\noindent
Thm.~\ref{thm:PoALowerBound}:
	For $(U,\tau)\in \INs\times\INs$ with $U \geq 2\tau$, we have 
	$\MPoA(U,\tau),\TPoA(U,\tau) \in \bigOm(U\log \tau)$.\footnotemark
\end{framed} 
\footnotetext{Similarly, $\bigOm(U \log\tau)$ denotes the set of all functions in $U$ and $\tau$ growing asymptotically at least as fast as $U\log\tau$, i.e. all functions $f: \INs \times \INs \to \IR_{\geq 0}$ such that there exist some $k, U_0,\tau_0 \in \INo$ with $k\cdot f(U,\tau) \geq U\log\tau$ for all $U \geq U_0, \tau \geq \tau_0$.}
The lower bound is based on a quite complex instance (see \Cref{fig:SlowTermGraphIntuitive})
that works roughly as follows.
We combine two gadgets: A ``cycling gadget'' consisting of a large cycle made of edges with capacity $\approx U$ and a ``blocking gadget'' consisting
of paths with low capacity and length of about $\tau$ connecting the nodes on the cycle to the sink node (see \Cref{fig:SlowTermGadgetC} for the cycling gadget, \Cref{fig:SlowTermGadgetBK} for the blocking gadget and \Cref{fig:SlowTermGraph} for how they are combined). 
An IDE flow within this graph can then alternate between two different phases: A ``charging phase'', wherein the main amount of flow travels once around the big cycle, loosing a small amount of flow to each of the paths leading towards the sink, and a ``blocking phase'', in which the particles traveling along the paths form queues again and again in just the right way as to keep the main amount of flow traveling around on the large cycle without loosing any more flow. In order to derive a lower bound on the price of anarchy we then augment this instance in such a way that a respective system optimal flow can just bypass the two gadgets and reach the sink in constant time while any IDE flow gets diverted into the cycling gadget.

\subsection{Further Related Work}
In the dynamic traffic assignment literature, the concept of IDE
 was introduced and studied  by several papers such as  Ran and Boyce~\cite[\S~VII-IX]{Ran96}, Boyce, Ran and LeBlanc~\cite{BoyceRL95,RanBL93} and \citeauthor*{FrieszLTW89}~\cite{FrieszLTW89}. These works developed an optimal control-theoretic formulation and characterized instantaneous user equilibria by Pontryagin's optimality conditions.
It is worth mentioning that the control-theoretic formulation is actually not compatible with the  deterministic queueing model of Vickrey.
In Boyce, Ran and LeBlanc~\cite{BoyceRL95}, a differential equation
per edge governing the cumulative edge flow (state variable) is used. The right-hand side of the
differential equation depends on the function denoting the rate at which flow leaves the edge (the exit flow function) which
is assumed to be differentiable and strictly positive for any
positive inflow. Both assumptions (positivity and differentiability)
are not satisfied for the Vickrey model. For example, flow entering
an empty edge needs a strictly positive time after which it leaves 
the edge again, thus, violating the strict positiveness of the exit
flow function. More importantly, differentiability of the exit flow function
is not guaranteed for the Vickrey queueing model.
Non-differentiability (or equivalently discontinuity w.r.t. the state variable) is a well-known obstacle in the convergence
analysis of a discretization of the Vickrey model.  
\citeauthor*{Han2013a}~\cite{Han2013a} considered dynamic equilibrium and reformulated
the Vickrey model using a partial differential equation formulation. 
In~\cite{Han2013b}, they further derived numerical algorithms for computing 
dynamic equilibria by discretizing the model and showing that the limit points
correspond to dynamic equilibria of the continuous model.
For further numerical schemes computing dynamic equilibria, we refer to the recent survey
of Friesz and Han~\cite{Friesz19} and Peeta et al~\cite{Peeta01}.   
 Koch and Skutella~\cite{Koch11} gave a novel constructive characterization of dynamic equilibria in terms of a family of static flows (thin flows).
 Cominetti, Correa and Omar~\cite{CominettiCL15} refined this characterization and Koch and Sering~\cite{Sering2019} incorporated spillbacks in the fluid queuing model. 
Kaiser~\cite{Kaiser20} showed that the thin flows needed for computing dynamic equilibria can be determined in polynomial time for series-parallel networks. Cominetti, Correa and Olver~\cite{CominettiCO17} studied the long term behavior of dynamic equilibria with infinitely lasting constant inflow rate at a single source. They characterized stable state distributions by an LP formulation provided that the inflow rate is at most the capacity of a minimal $s$-$t$ cut.

Regarding the price of anarchy of dynamic equilibria, Koch and Skutella~\cite{Koch11} derived the first results on the price of anarchy for dynamic equilibria, which were recently improved by  Correa, Cristi and Oosterwijk~\cite{CCO19} devising a tight bound of $\frac{e}{e-1}$, provided that a certain monotonicity conjecture holds. {While this conjecture is wildly believed to be true, it still remains unproven to this day.
Israel and Sering~\cite{israel2020impact} investigated the price of anarchy
for the model with spillbacks.
Bhaskar, Fleischer and Anshelevich~\cite{BhaskarFA15} devised Stackelberg strategies in order to improve the efficiency of dynamic equilibria. Recently, Frascaria and Olver~\cite{Frascaria2019} considered a flexible departure choice
 model from an optimization point of view and derived insights into
 devising tolls for improving the performance of dynamic equilibria. 
 
Ismaili~\cite{ismaili2017} considered a discrete version of IDEs (see also Cao et al.~\cite{CaoCCW17}, Scarsini et al. \cite{ScarsiniST18}, Harks et al.~\cite{HarksPSTK18} and Hoefer et al.~\cite{HoeferMRT11} for further works on discrete packet models)
and investigated the price of anarchy with respect to the sum of travel times and derived lower bounds of order $\Omega(|V|+n)$
for the setting that only simple paths are allowed.
Here $n$ denotes the number of discrete players in the game. 
Very recently \citeauthor*{AtomicBound2021}~\cite{AtomicBound2021} proved for acyclic graphs an upper bound of $\tPmax + n$ on the residence time of agents in the disrete version. Here, $\tPmax$ denotes the length of a longest source-sink-path in the network. 
For general multi-commodity instances allowing also cycles,
Ismaili~\cite{ismaili2017} proved that the price of anarchy is unbounded. Similarly, Graf, Harks and Sering~\cite{GHS18} showed that for the continuous version, multi-commodity IDE flows may cycle forever and, thus, the price of anarchy is infinity. For IDE flows in single-sink networks, on the other hand, they showed that termination is always guaranteed. However, due to the non-constructive nature of the proof no explicit bound on the termination time or the price of anarchy follows.

%!TEX root = ../article-PoA.tex

\section{Model}\label{sec:Model}

Let $\network = (G,(\nu_e)_{e \in E},(\tau_e)_{e \in E},(u_v)_{v \in V\setminus\set{t}}, t)$ be a network consisting of a directed graph $G=(V,E)$, edge capacities $\nu_e \in \INs$, edge travel times $\tau_e \in \INs$,\footnote{Throughout this paper we will restrict ourselves to integer travel times and edge capacities to make the statements and proofs cleaner. However, all results can be easily applied to instances with rational travel times and capacities by simply rescaling the instance appropriately. Note, however, that all bounds will then scale accordingly.} a single sink node $t \in V$ reachable from every other node and for every node $v \in V\setminus\set{t}$ a corresponding integrable (network) inflow rate $u_v: \IR_{\geq 0} \to \IR_{\geq 0}$. The idea then is that, at all times $\theta\in \IR_{\geq 0}$ infinitesimal small agents enter the network at node $v$ at a rate according to $u_v(\theta)$ and start traveling through the graph towards the common sink $t$. Such a dynamic can be described by a \emph{flow over time}, a tuple $f = (f^+,f^-)$ where $f^+,f^-: E \times \IR_{\geq 0} \to \IR_{\geq 0}$ are integrable functions. For any edge $e \in E$ and time $\theta \in \IR_{\geq 0}$ the value $f^+_e(\theta)$ describes the \emph{(edge) inflow rate} into $e$ at time $\theta$ and $f^-_e(\theta)$ is the \emph{(edge) outflow rate} from $e$ at time $\theta$. 

For any such flow over time $f$ we define the \emph{cumulative (edge) in- and outflow rates} $F^+$ and $F^-$ as
\begin{align*}
	F^+_e(\theta) \coloneqq \int_{0}^{\theta}f^+_e(\zeta)d\zeta \quad\text{ and }\quad F^-_e(\theta) \coloneqq \int_{0}^{\theta}f^-_e(\zeta)d\zeta, 	
\end{align*}
respectively. The queue length of edge $e$ at time $\theta$ is then defined as
\begin{align*}
	q_e(\theta) \coloneqq F^+_e(\theta)-F^-_e(\theta+\tau_e). 
\end{align*}

We call such a flow $f$ a \emph{feasible flow} for a given set of network inflow rates $u_v: \IR_{\geq 0} \to \IR_{\geq 0}$ for each node $v \in V \setminus \{t\}$, if it satisfies the following \cref{eq:FeasibleFlow-FlowConservation,eq:FeasibleFlow-FlowConservationSink,eq:FeasibleFlow-NoOutflowBeforeTau,eq:FeasibleFlow-QueueOpAtCap}.
The \emph{flow conservation constraints} are modeled for all nodes $v \neq t$ as
\begin{align}\label[cons]{eq:FeasibleFlow-FlowConservation} 
\sum_{e \in \delta^+_v}f^+_e(\theta) - \sum_{e \in \delta^-_v}f^-_e(\theta) = u_v(\theta) & \quad \text{ for all } \theta \in \IR_{\geq 0},
\end{align}
where $\delta^+_v := \set{vu \in E}$ and $\delta^-_v := \set{uv \in E}$ are the sets of outgoing edges from $v$ and incoming edges into $v$, respectively. For the sink node $t$ we require
\begin{align}\label[cons]{eq:FeasibleFlow-FlowConservationSink} 
\sum_{e \in \delta^+_{t}}f^+_e(\theta) - \sum_{e \in \delta^-_{t}}f^-_e(\theta) \leq 0
\end{align}
and for all edges $e \in E$ we always assume 
\begin{align}\label[cons]{eq:FeasibleFlow-NoOutflowBeforeTau}
f_e^-(\theta) = 0 &\text{ f.a. } \theta < \tau_e.
\end{align}
Finally we assume that the queues operate at capacity which can be modeled by
\begin{align}	\label[cons]{eq:FeasibleFlow-QueueOpAtCap} 
f_e^-(\theta + \tau_e) = 
\begin{cases}
\nu_e, & \text{ if } q_e(\theta) > 0 \\
\min\set {f^+_e(\theta), \nu_e}, & \text{ if } q_e(\theta) \leq 0
\end{cases}  
& \quad \text{ for all }  e \in E, \theta \in \IR_{\geq 0}.
\end{align}

\paragraph{Objectives of Flows over Time.}

We will now introduce some additional notation in order to formally define the two objectives we want to measure IDE against: Makespan and total travel times. Since both objectives are only relevant for flows with finitely lasting inflow rates, from here on we will always assume that there exists some time $\theta_0$, such that the supports of all network inflow rates $u_v$ are contained in $[0,\theta_0]$.

Following \cite{Sering2019}, for any feasible flow $f$ and every edge $e \in E$ we define the \emph{edge load} function $\edgeLoad[e]$ that gives us for any time $\theta$ the flow volume currently on edge $e$ (either waiting in its queue or traveling along the edge): 
	\[\edgeLoad[e]: \IR_{\geq 0} \to \IR_{\geq 0}, \theta \mapsto F^+_e(\theta)-F^-_e(\theta). \]
The function $\edgeLoad(\theta) := \sum_{e \in E}\edgeLoad[e](\theta)$ then gives  the \emph{total flow volume in the network at time $\theta$}. 
Furthermore, we define a function $Z$ indicating the volume of flow that already reached the sink $t$ by time $\theta$:
	\begin{align}\label{eq:DefZ}
		Z: \IR_{\geq 0} \to \IR_{\geq 0}, \theta \mapsto \sum_{e \in \edgesEntering{t}}F^-_e(\theta) - \sum_{e \in \edgesLeaving{t}}F^+_e(\theta)
	\end{align}
and for any node $v \neq t$ the \emph{cumulative network inflow at $v$} as
	\[U_v: \IR_{\geq 0} \to \IR_{\geq 0}, \theta \mapsto \int_{0}^{\theta} u_v(\zeta)d\zeta.\]
We will make use of the following connection between these functions (cf. \cite[Section 4]{GHS18}).
	
\begin{lemma}\label{lem:GeIsTotalFlowInGraph}
	Let $f$ be a feasible flow. Then for every subset $W \subseteq V$ and any time $\theta$ we have
	\[\sum_{e \in E(W)}\edgeLoad[e](\theta) = \sum_{e \in \delta^-_W}F_e^-(\theta) + \sum_{v \in W}U_v(\theta) - \sum_{e \in \delta^+_W}F_e^+(\theta) - \begin{cases}
	Z(\theta), 	& t \in W \\
	0,				& \text{else}
	\end{cases},\]
	where we define the edge set $E(W) \coloneqq \set{vw \in E | v,w \in W}$ of all edges between nodes in $W$ as well as the sets $\delta^+_W := \set{wv \in E | w \in W, v\notin W}$ of outgoing edges from $W$ and $\delta^-_W := \set{vw \in E | v \notin W, w \in W}$ of incoming edges into $W$. In particular taking $W = V$ we get
	\[\edgeLoad(\theta) = \sum_{v \in V\setminus\set{t}}U_v(\theta) - Z(\theta). \qedhere\]
\end{lemma}
	
\begin{proof} 
	The proof is a direct computation - we only show the case $t \in U$ as the other one can be proven exactly the same way:
	\begin{align*}
	&\sum_{e \in E(U)}\edgeLoad[e](\theta) = \sum_{e \in E(U)}\left(F^+_e(\theta)-F^-_e(\theta)\right) \\
	&\quad= \sum_{v \in U}\left(\sum_{e \in \delta^+_v}F_e^+(\theta) - \sum_{e \in \delta^-_v}F_e^-(\theta)\right) - \sum_{e \in \delta^+_U}F_e^+(\theta) + \sum_{e \in \delta^-_U}F_e^-(\theta)\\
	&\quad= \sum_{v \in U\setminus\set{t}}\left(\sum_{e \in \delta^+_v}\int_0^{\theta}f_e^+(\zeta)d\zeta - \sum_{e \in \delta^-_v}\int_0^{\theta}f_e^-(\zeta)d\zeta\right) - Z(\theta) - \sum_{e \in \delta^+_U}F_e^+(\theta) + \sum_{e \in \delta^-_U}F_e^-(\theta)\\
	&\quad= \int_0^{\theta}\sum_{v \in U\setminus\set{t}}\left(\sum_{e \in \delta^+_v}f_e^+(\zeta) - \sum_{e \in \delta^-_v}f_e^-(\zeta)\right)d\zeta - Z(\theta) - \sum_{e \in \delta^+_U}F_e^+(\theta) + \sum_{e \in \delta^-_U}F_e^-(\theta)\\
	&\quad= \int_0^{\theta}\sum_{v \in V\setminus\set{t}}u_v(\zeta)d\zeta + \sum_{e \in \delta^-_U}F_e^-(\theta) - Z(\theta) - \sum_{e \in \delta^+_U}F_e^+(\theta) \qedhere
	\end{align*}
\end{proof}
The last part then implies that -- as one should expect -- after time $\theta_0$, the total flow volume in the network is non-increasing since the cumulative network inflow is constant after this point while the total flow volume having already reached the sink is non-decreasing.

\begin{cor}\label{cor:FlowVolumeReduction}
	Let $f$ be a feasible flow. Then for all $b \geq a \geq \theta_0$, we have
	$\edgeLoad(b) \leq \edgeLoad(a). $
	In particular, for $\hat{\theta} \geq \theta_0$ with $\edgeLoad(\hat{\theta}) = 0$, we have $\edgeLoad(\theta) = 0$ for all $\theta \geq \hat{\theta}$.
\end{cor}

\begin{proof}
	The \namecref{cor:FlowVolumeReduction} follows directly from \Cref{lem:GeIsTotalFlowInGraph}, since $\sum_{e \in \delta^-_t}f^-_e(\theta) - \sum_{e \in \delta^+_t}f^+_e(\theta)$ is always non-negative by \cref{eq:FeasibleFlow-FlowConservationSink} and so $Z(\theta) = \sum_{e \in \delta^-_t}F^-_e(\theta) - \sum_{e \in \delta^+_t}F^+_e(\theta)$ is non-decreasing, while -- by the definition if $\theta_0$ -- all $U_v$ are constant after time $\theta_0$. 
\end{proof}

This motivates the following definition of the makespan of a feasible flow.
\begin{defn}\label{def:makespan}
	A feasible flow over time $f$ \emph{terminates} if it satisfies 
		\[\inf\set{\theta \geq \theta_0 | \edgeLoad(\theta)=0} < \infty,\] 
	that is if there exists some finite time after time $\theta_0$ such that the total flow volume in the network at that time is zero. We then call the smallest such time $\termTime(f) \coloneqq \inf\set{\theta \geq \theta_0 | \edgeLoad(\theta)=0}$ the \emph{makespan} of $f$ and say that $f$ \emph{terminates} by time $\termTime(f)$. \Cref{cor:FlowVolumeReduction} implies $\edgeLoad(\theta) = 0$ for all $\theta > \termTime(f)$. Note, that the makespan is also defined for non-terminating flows -- as the infimum over an empty set it is then infinity.
\end{defn}

For terminating flows over time we can use their makespan (i.e. the time the last particle arrives at the sink) as one measure of their quality. Alternatively, we can also look at the total travel times. For a single edge $e$ we can calculate the total travel time spend on this edge by $\int_0^\infty f^+_e(\theta)c_e(\theta)d\theta$, where $c_e(\theta) \coloneqq \tau_e + \frac{q_e(\theta)}{\nu_e}$ denotes the \emph{current} or \emph{instantaneous travel time} at time $\theta$ over edge $e$. 
Thus, we get the total travel time spend on an edge $e$ by taking for every time $\theta$ the product of the current inflow rate into this edge and the travel time experienced by all particles entering at this time and then ``summing'' over all times. The sum over all edges then gives us the total travel times in the whole network.
\begin{defn}\label{def:totalTravelTime}
	Let $f$ be a terminating feasible flow. Then we define the \emph{total travel times for $f$} by
		\[\sumOfTraveltimes(f) \coloneqq \sum_{e \in E}\int_0^\infty c_e(\theta)f^+_e(\theta) d\theta.\]
\end{defn}

The following lemma gives two alternative ways to define the same value. The latter, in particular, will be useful in the rest of the paper and has the following intuitive interpretation (see \cite[Section 5]{PredictionEquilibrium2021}): For every time $\theta$ we take the total flow volume in the network $F^{\Delta}$ at that time and then integrate these values for all times.

\begin{lemma}\label{lemma:SumOfTraveltimesAltDef}\label{lemma:SumOfTravelTimesDefWithEdgeLoad}
	For any terminating feasible flow we have
		\[\sumOfTraveltimes(f) = \int_{0}^{\termTime(f)}\theta\cdot Z'(\theta)d\theta - \sum_{v \in V\setminus\{t\}}\int_{0}^{\theta_0} \theta \cdot u_v(\theta)d\theta = \int_0^{\termTime(f)} F^\Delta(\theta)d\theta.\]
\end{lemma}

The proof of this lemma is a somewhat technical but mostly straight forward calculation (see \cref{appendix} for the details). The main idea is to first show that for each individual edge $e$ we have 
	\[\int_0^{\termTime(f)}c_e(\theta)f_e^+(\theta)d\theta = \int_0^{\termTime(f)}(\theta+\tau_e)f_e^-(\theta+\tau_e) + \theta f_e^+(\theta)d\theta.\]
From this the lemma's first equality follows by a telescope sum argument and flow conservation at the nodes. The second equality can then be shown using integration by parts.

\begin{remark}
	In \cite{BhaskarFA15} \citeauthor{BhaskarFA15} define another closely related objective called the \emph{total delay} of a flow by
		\[D(f) \coloneqq \int_0^{\termTime(f)}\theta Z'(\theta).\]
	Intuitively, this is just the sum of all particles' arrival times (in contrast to the total travel time which is the sum of all particles' \emph{travel} times). They then show (\cite[Theorem 10]{BhaskarFA15}) that for dynamic equilibria the total delay price of anarchy is at most twice the makespan price of anarchy. Since their proof of this result in fact holds for any feasible flow (not just dynamic equilibria) it also applies to IDE. Thus, any bound on the makespan PoA for IDE also immediately gives us an asymptotically equal bound on the total delay PoA for IDE.
\end{remark}

\paragraph{IDE Flows and their PoA.}

Following \cite{GHS18} we define an IDE flow as a feasible flow with the property that whenever a particle arrives at a node $v \neq t$, it can only enter an edge that is the first edge on a currently shortest $v$-$t$ path, i.e. a shortest path with respect to the instantaneous travel time $c_e(\theta)$.
In order to make this more formal we use time dependent node labels $\ell_{v}(\theta)$ corresponding to current shortest path distances from $v$ to the sink $t$. For $v\in V$ and $\theta\in \R_{\geq 0}$, define
\begin{equation}
	\ell_{v}(\theta)\coloneqq
	\begin{cases} 
		0, & \text{ for } v=t\\
		\min\limits_{e=vw\in E} \{\ell_{w}(\theta)+c_{e}(\theta)\}, & \text{ else.}
	\end{cases}
\end{equation}
We say that an edge $e=vw$ is \emph{active} at time $\theta$, if $ \ell_{v}(\theta) = \ell_{w}(\theta)+c_{e}(\theta)$ and we denote the set of active edges by $E_\theta\subseteq E$. We call a $v$-$t$ path $P$ an \emph{active $v$-$t$ path at time $\theta$}, if all edges of $P$ are active for $i$ at $\theta$ or, equivalently, $\sum_{e \in P}c_e(\theta) = \ell_{v}(\theta)$. For differentiation we call paths that are minimal with respect to the transit times $\tau$ \emph{physical shortest paths}.

\begin{defn}\label{Def:IDE}
	A feasible flow over time $f$ is an \emph{instantaneous dynamic equilibrium (IDE)}, if for all  $\theta\in \R_{\geq 0}$ and $e\in E$ it satisfies 
	\begin{align}\label[cons]{eq:DefIDE-OnlyUseSP}
	f_{e}^+(\theta)>0 \Rightarrow e\in E_\theta.
	\end{align}
\end{defn}

Since in an IDE flow particles act selfishly and without cooperation we should expect that IDE flows will not be optimal with regard to objectives like makespan or total travel time. To quantify this differences we will use the price of anarchy, which we define as follows: For any instance $\mathcal{N} = (G,(\nu_e)_{e \in E},(\tau_e)_{e \in E},(u_v)_{v \in V\setminus\set{t}}, t)$ we define the worst case makespan of an IDE flow in $\mathcal{N}$ as
	\[\termTime[IDE](\mathcal{N}) \coloneqq \sup\set{\termTime(f) | f \text{ an IDE flow in } \mathcal{N}}\]
and the optimal makespan in $\mathcal{N}$ as
	\[\termTime[OPT](\mathcal{N}) \coloneqq \inf\set{\termTime(f) | f \text{ a feasible flow in } \mathcal{N}}.\]
In the same way we define the worst case total travel time as
	\[\sumOfTraveltimes[IDE](\mathcal{N}) \coloneqq \sup\set{\sumOfTraveltimes(f) | f \text{ an IDE flow in } \mathcal{N}}\]
and the optimal total travel times as
	\[\sumOfTraveltimes[OPT](\mathcal{N}) \coloneqq \inf\set{\sumOfTraveltimes(f) | f \text{ a feasible flow in } \mathcal{N}}.\]
	
Note that we use supremum and infemum here instead of maximum and minimum as the set of all IDE flows/of all feasible flows is not necessarily compact as a subset of the appropriate function space. Thus, it is a priori not clear whether a maximum or minimum is attained. However, since we are only interested in a bound on the worst case ratio of worst IDE to optimal flow, this is also not important for the further analysis. To define this bound we again take the supremum over all ratios of worst case IDE to optimal feasible flow over the set of all instances with the same total flow volume and total edge length. This then gives us the parametrized price of anarchy, which -- in other words -- is just a tight bound on how much worse an IDE can be with regards to the given objective compared to an optimal flow.

\begin{defn}\label{def:PoA}
	For any pair of whole numbers $(U,\tau)$ we define the \emph{makespan/total travel times Price of Anarchy} (\MPoA/\TPoA) for instances with total flow volume $U$ and total edge length $\tau$ as
		\[\MPoA(U,\tau) \coloneqq \sup\Set{\frac{\termTime[IDE](\mathcal{N})}{\termTime[OPT](\mathcal{N})} | \mathcal{N} \text{ an instance with } \!\!\!\sum_{v \in V\setminus\set{t}}U_v(\theta_0) = U, \sum_{e \in E}\tau_e = \tau}\]		
	and
		\[\TPoA(U,\tau) \coloneqq \sup\Set{\frac{\sumOfTraveltimes[IDE](\mathcal{N})}{\sumOfTraveltimes[OPT](\mathcal{N})} | \mathcal{N} \text{ an instance with } \!\!\!\sum_{v \in V\setminus\set{t}}U_v(\theta_0) = U, \sum_{e \in E}\tau_e = \tau},\]		
	respectively.
\end{defn}

\begin{remark}
	At first it might seem strange that the $\PoA$ depends only on $U$ and $\tau$ while being independent of the capacities $\nu_e$. However, this is only the case here because we always assume that all capacities are at least $1$ throughout this paper. In order to transfer our results to networks with arbitrary capacities one has to rescale the network and, in particular, replace $U$ by $\frac{1}{\nu_{\min}}U$, where $\nu_{\min} \coloneqq \min\set{\nu_e | e \in E}$.
\end{remark}

%!TEX root = ../article-PoA.tex

\section{Upper Bounds}\label{sec:UpperBounds}

In this section we will show upper bounds for the makespan as well as the total travel times for IDE flows in terms of $\tau(G) \coloneqq \sum_{e \in E}\tau_e$ and $U \coloneqq \sum_{v \in V\setminus\set{t}}\int_{0}^{\theta_0}u_v(\theta)d\theta$. From this we can then derive upper bounds for the respective $\PoA$s.
Before we turn to the termination results, however, we need an additional technical lemma that is a strengthening of \cite[Lemma 4.2]{GHS18} and gives an upper bound for the length of time a volume of flow can take before it leaves the edge.

\begin{lemma}\label{lemma:FlowOnEdgesLeavesFastVar}
	Let $f$ be a feasible flow over time, $\theta_1 \in \IR_{\geq 0}$, $e \in E$ and $0 \leq \lambda \leq \edgeLoad[e](\theta_1)$. Then the cumulative outflow from edge $e$ between times $\theta_1$ and $\theta_2 := \theta_1+\frac{\lambda}{\nu_e} + \tau_e$ will be at least $\lambda$, i.e.,
	\[ F^-_e(\theta_2) - F^-_e(\theta_1) \geq \lambda. \]
\end{lemma}

Intuitively, this lemma holds true because all particles which are on edge $e$ at time $\theta_1$ but have not left this edge by time $\theta_1 + \tau_e$ must have been in the queue of this edge at time $\theta_1$. And since queues operate at capacity (\cref{eq:FeasibleFlow-QueueOpAtCap}) these particles leave the queue and enter the edge itself at a rate of $\nu_e$ (leaving it $\tau_e$ time steps later).

\begin{proof}
	By way of contradiction we assume $F^-_e(\theta_2) - F^-_e(\theta_1) < \lambda \leq \edgeLoad[e](\theta_1) = F^+_e(\theta_1) - F^-_e(\theta_1)$. From this, we get $F^-_e(\theta_2) < F^+_e(\theta_1) \leq F^+_e(\theta_2-\tau_e)$ and, since $F^+_e$ and $F_e^-$ are non-decreasing, the same holds true for all $\theta \in [\theta_1+\tau_e,\theta_2]$ in the place of $\theta_2$. Thus, we have $q_e(\theta) = F^+_e(\theta) - F^-_e(\theta+\tau_e) > 0$ for all $\theta \in [\theta_1,\theta_2 - \tau_e]$ and, by \cref{eq:FeasibleFlow-QueueOpAtCap}, $f^-(\theta+\tau_e) = \nu_e$ for all such $\theta$.	
	But this leads to a contradiction as follows
	\begin{align*}
		F^-_e(\theta_2) - F^-_e(\theta_1) &\geq F^-_e(\theta_2) - F^-_e(\theta_1+\tau_e) = \int_{\theta_1+\tau_e}^{\theta_2} f^-_e(\theta)d\theta \\ 			&= \int_{\theta_1}^{\theta_2-\tau_e} f^-_e(\theta+\tau_e)d\theta = (\theta_2-\tau_e-\theta_1)\nu_e = \lambda.\qedhere
	\end{align*}
\end{proof}
We will now look at acyclic graphs and extend the above result to the whole graph, i.e. derive an upper bound on the time flow particles can take to reach the sink. This, in particular, then provides a general termination bound for \emph{all} feasible flows in acyclic graphs in terms of $U$ and $\tau(P_{\max})$, where the latter denotes the physical length of a longest $v$-$t$ path.

\begin{lemma}\label{lemma:TerminationAcyclicNetworksV2}
	In an acyclic network, every feasible flow over time $f$ satisfies
		\[Z(\theta) \geq Z(\zeta) + \min\set{\edgeLoad(\zeta),\theta-\zeta-\tPmax}\]
	for all times $\theta \geq \zeta \geq 0$, that is as long as there is some volume of flow in the network at some time $\zeta$ this volume arrives at the sink at most $\tPmax$ time steps later at a rate of at least $1$. In particular, any feasible flow in an acyclic network terminates before $\theta_0 + \tPmax + U$. 
\end{lemma}

\begin{proof}
	For this proof we will assume that all edge travel times are $1$. If that is not the case the network can be easily modified without changing the flow dynamics by replacing each edge $e$ with $\tau_e$ consecutive edges of travel time $1$ each. Furthermore, we assume that the sink $t$ has no outgoing edges since in an acyclic graph such edges can never be used by a feasible flow.
	
	To show the \namecref{lemma:TerminationAcyclicNetworksV2} we first need some notation: For every node $v \in V$ we define the \emph{pessimistic remaining distance}
		$\distP{v} \coloneqq \max\set{\abs{P} | P \text{ is a $v$-$t$ path}}$ 
	as the physical length of the longest $v$-$t$ path (cf. \Cref{fig:TerminationAcyclic}). Since our network is acyclic, $\distP{v}$ is well defined and every edge $e=(u,v)$ satisfies $\distP{u} > \distP{v}$. For any $k = 1,\dots, \tPmax$ we define
		$V_k \coloneqq \Set{v \in V \setminus \set{t} | \distP{v} = k}$ and $V_{\leq k} \coloneqq \Set{v \in V \setminus \set{t} | \distP{v} \leq k}$,
	the sets of all nodes with pessimistic distance equal to and at most $k$, respectively. Furthermore, we define the set $E^{>k}_{\leq k} \coloneqq \set{(u,v) \in E | \distP{u} > k \geq \distP{v}}$ of all edges crossing level $k$ as well as the subsets $E^{k+1}_{\leq k} \coloneqq \set{(u,v) \in E | \distP{u} = k+1, k \geq \distP{v}}$ and $E^{>k}_{k} \coloneqq \set{(u,v) \in E | \distP{u} > k = \distP{v}}$. Finally, for any $k = 0,1, \dots, \tPmax$ let	
		\[\pClosEL{k}(\theta) \coloneqq \sum_{e \in E^{>k}_{\leq k}}F^-_e(\theta) + \sum_{v \in V_{\leq k}}U_v(\theta)\]
	denote the total volume of flow that currently has a pessimistic remaining distance of less than $k$ (including all the flow volume already at the sink). 
	\begin{figure}[h]\centering
		\begin{adjustbox}{max width=.5\textwidth}
			\begin{tikzpicture}
	\node[color=gray]() at (-1,-2) {$\tilde{d}:$};
	\draw[color=gray,ultra thick, dashed] (0,-1.5) -- (0,1.5);	
	\node[color=gray]() at (0,-2) {$5$};
	\draw[color=gray,ultra thick, dashed] (2,-1.5) -- (2,1.5);	
	\node[color=gray]() at (2,-2) {$4$};
	\draw[color=gray,ultra thick, dashed] (4,-1.5) -- (4,1.5);	
	\node[color=gray]() at (4,-2) {$3$};
	\draw[color=gray,ultra thick, dashed] (6,-1.5) -- (6,1.5);	
	\node[color=gray]() at (6,-2) {$2$};
	\draw[color=gray,ultra thick, dashed] (8,-1.5) -- (8,1.5);	
	\node[color=gray]() at (8,-2) {$1$};
	\draw[color=gray,ultra thick, dashed] (10,-1.5) -- (10,1.5);	
	\node[color=gray]() at (10,-2) {$0$};

	\node[vertex,fill=white] (1) at (0,0) {};
	\node[vertex,fill=white] (2) at (2,1) {};
	\node[vertex,fill=white] (2') at (2,-1) {};
	\node[vertex,fill=white] (3) at (4,0) {};
	\node[vertex,fill=white] (4) at (6,0) {};
	\node[vertex,fill=white] (5) at (8,1) {};
	\node[vertex,fill=white] (5') at (8,-1) {};
	\node[namedVertex,fill=white] (t) at (10,0) {$t$};

	\draw[edge,ultra thick] (1)	-- (2);
	\draw[edge,ultra thick] (2)	-- (3);
	\draw[edge] (2')	-- (3);
	\draw[edge] (2')	-- (5');
	\draw[edge,ultra thick] (3)	-- (4);
	\draw[edge,ultra thick] (4)	-- (5');
	\draw[edge] (4)	-- (t);
	\draw[edge] (5)	-- (t);
	\draw[edge,ultra thick] (5')	-- (t);
\end{tikzpicture}
		\end{adjustbox}
		\caption{An example network with the nodes placed according to their pessimistic distance $\distP{v}$. The bold edges form the path $P_{\max}$ in this network. \Cref{lemma:TerminationAcyclicNetworksV2:claim:LBforPessFlow} states that as long as there is remaining flow volume to the left of any level $k$ this flow volume crosses this level at most $\tPmax-k$ time steps later at a rate of at least $1$.}\label{fig:TerminationAcyclic}
	\end{figure}
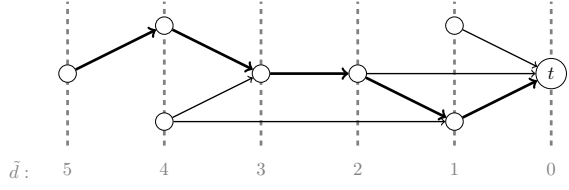
	The following claim will then immediately imply the \namecref{lemma:TerminationAcyclicNetworksV2}.
	\begin{claim}\label{lemma:TerminationAcyclicNetworksV2:claim:LBforPessFlow}
		For any $k=0,1, \dots, \tPmax$ and all times $\theta \geq \zeta \geq 0$ we have
		\begin{align}\label{eq:TermInAcyclicNetworkFlowVolumeCloseToSink}
			\pClosEL{k}(\theta) \geq M(\theta,\zeta,k) \coloneqq \min\set{\pClosEL{\tPmax}(\zeta), \pClosEL{k}(\zeta) + \theta-\zeta-\tPmax + k}
		\end{align}
	\end{claim}
	
	\begin{proofClaim}[Proof of \Cref{lemma:TerminationAcyclicNetworksV2:claim:LBforPessFlow}]
		First note, that it suffices to show this inequality for all $\theta \in [\zeta + \tPmax - k,\zeta + \tPmax - k + \pClosEL{\tPmax}(\zeta) - \pClosEL{k}(\zeta)]$. For any time $\theta$ before that interval we have $M(\theta,\zeta,k) = \pClosEL{k}(\zeta)$, which is at most $\pClosEL{k}(\theta)$ since $\theta \geq \zeta$ and $\pClosEL{k}$ is non-decreasing. After the interval we have $M(\theta,\zeta,k) = \pClosEL{\tPmax}(\zeta)$ and, thus, if \eqref{eq:TermInAcyclicNetworkFlowVolumeCloseToSink} holds at the end of the interval it holds forever after that (again because $\pClosEL{k}$ is non-increasing).
		
		For $\theta \in [\zeta + \tPmax - k,\zeta + \tPmax - k + \pClosEL{\tPmax}(\zeta) - \pClosEL{k}(\zeta)]$ we will now show \eqref{eq:TermInAcyclicNetworkFlowVolumeCloseToSink} by downward induction on $k$ starting with $k = \tPmax$.
		
		\begin{description}
			\item[$\bm{k = \tPmax}$:] Here \eqref{eq:TermInAcyclicNetworkFlowVolumeCloseToSink} holds trivially as $\pClosEL{\tPmax}$ is non-decreasing and, thus, the inequality $\pClosEL{\tPmax}(\theta) \geq \pClosEL{\tPmax}(\zeta) \geq M(\theta,\zeta,\tPmax)$ is true for all $\theta \geq \zeta$.
			\item[$\bm{k \to k-1}$:] We assume that \eqref{eq:TermInAcyclicNetworkFlowVolumeCloseToSink} holds for some fixed $k$ and all $\theta \geq \zeta$. By way of contradiction assume that \eqref{eq:TermInAcyclicNetworkFlowVolumeCloseToSink} does \emph{not} hold for $k-1$, i.e. there exists some times $\zeta \in \IR_{\geq 0}$ and $\theta \in [\zeta + \tPmax - (k-1),\zeta + \tPmax - (k-1) + \pClosEL{\tPmax}(\zeta) - \pClosEL{k-1}(\zeta)]$ with
				\[\pClosEL{k-1}(\theta) < M(\theta,\zeta,k-1).\]
			We define 
				\[\bar{\theta} \coloneqq \inf\set{\theta \geq \zeta + \tPmax - (k-1) | \pClosEL{k-1}(\theta) < M(\theta,\zeta,k-1)}.\]
			Since \eqref{eq:TermInAcyclicNetworkFlowVolumeCloseToSink} holds for $\theta = \zeta + \tPmax - (k-1)$ and both sides of the inequality are differentiable on the relevant interval, this implies that here exists some $\varepsilon > 0$ such that for all $\theta \in (\bar{\theta},\bar{\theta}+\varepsilon)$ we have
				\[\sum_{e \in E^{>k-1}_{\leq k-1}}f^-_e(\theta) = \frac{\partial}{\partial \theta}\pClosEL{k-1}(\theta) < \frac{\partial}{\partial \theta}M(\theta,\zeta,k-1) = 1.\]
			Thus, for all $e \in E^{>k-1}_{\leq k-1}$ and any $\theta \in (\bar{\theta},\bar{\theta}+\varepsilon)$ we have $f^-_e(\theta) < 1 \leq \nu_e$. By \cref{eq:FeasibleFlow-QueueOpAtCap} (queues operate at capacity) we get 
				\begin{align}\label{eq:lemma:TerminationAcyclicNetworksV2:claim:LBforPessFlow}
					F^+_e(\theta-1) - F^-_e(\theta) = F^+_e(\theta-\tau_e) - F^-_e(\theta) = q_e(\theta-\tau_e) = 0,
				\end{align}
			which in turn implies
				\begin{align*}
					\pClosEL{k-1}(\theta) &= \sum_{\mathclap{e \in E^{>k}_{\leq k-1}}}F^-_e(\theta) + \sum_{\mathclap{e \in E^{k}_{\leq k-1}}}F^-_e(\theta) + \sum_{\mathclap{v \in V_{\leq k-1}}}U_v(\theta) \\
					&\overset{\text{\eqref{eq:lemma:TerminationAcyclicNetworksV2:claim:LBforPessFlow}}}{=} \sum_{\mathclap{e \in E^{>k}_{\leq k-1}}}F^-_e(\theta) + \sum_{\mathclap{e \in E^{k}_{\leq k-1}}}F^+_e(\theta-1) + \sum_{v \in V_{\leq k-1}}U_v(\theta) \\
					&\overset{\text{\eqref{eq:FeasibleFlow-FlowConservation}}}{=} \sum_{\mathclap{e \in E^{>k}_{\leq k-1}}}F^-_e(\theta) + \sum_{\mathclap{e \in E^{>k}_{k}}}F^-_e(\theta-1) + \sum_{\mathclap{v \in V_k}}U_v(\theta-1) + \sum_{\mathclap{v \in V_{\leq k-1}}}U_v(\theta) \\
					&\geq \sum_{\mathclap{e \in E^{>k}_{\leq k-1}}}F^-_e(\theta-1) + \sum_{\mathclap{e \in E^{>k}_{k}}}F^-_e(\theta-1) + \sum_{\mathclap{v \in V_k}}U_v(\theta-1) + \sum_{\mathclap{v \in V_{\leq k-1}}}U_v(\theta-1) \\
					&=\pClosEL{k}(\theta-1) \geq M(\theta-1,\zeta,k) = M(\theta,\zeta,k-1),
				\end{align*}
			for all $\theta \in (\bar{\theta},\bar{\theta}+\varepsilon)$, where the last inequality holds by induction. But this is a contradiction to the definition of $\bar{\theta}$. \qedhere
		\end{description}
	\end{proofClaim}
	The first part of the \namecref{lemma:TerminationAcyclicNetworksV2} now follows from \Cref{lemma:TerminationAcyclicNetworksV2:claim:LBforPessFlow} by choosing $k=0$:
		\begin{align*}
			Z(\theta) 	&= \pClosEL{0}(\theta) \overset{\text{\Cref{lemma:TerminationAcyclicNetworksV2:claim:LBforPessFlow}}}{\geq} \min\set{\pClosEL{\tPmax}(\zeta), \pClosEL{0}(\zeta) + \theta-\zeta-\tPmax} \\
						&= \min\set{\edgeLoad(\zeta)+Z(\zeta), Z(\zeta) + \theta-\zeta-\tPmax} \\
						&= Z(\zeta) + \min\set{\edgeLoad(\zeta), \theta-\zeta-\tPmax}.
		\end{align*}
	For the second part we additionally choose $\zeta=\theta_0$ and $\theta=\theta_0+\tPmax+U$ leading to
		\begin{align*}
			Z(\theta) 	&\geq Z(\theta_0) + \min\set{\edgeLoad(\theta_0), U} \geq \min\set{Z(\theta_0) + \edgeLoad(\theta_0), U} = \min\set{U, U} = U,
		\end{align*}
	which immediately implies 
		\begin{align*}
			\edgeLoad(\theta) = U(\theta) - Z(\theta) \leq U - U = 0,
		\end{align*}
	i.e. $f$ has terminated by time $\theta = \theta_0+\tPmax+U$.
\end{proof}

It turns out, the same termination bound also holds for multi-commodity flows in acyclic networks (i.e. flows with multiple sinks).
	
\begin{lemma}\label{lemma:TerminationAcyclicNetworksMulti}
	In an acyclic network, every feasible multi-commodity flow over time $f$ satisfies
		\[Z(\theta) \geq Z(\zeta) + \min\set{\edgeLoad(\zeta),\theta-\zeta-\tPmax}\]
	for all times $\theta \geq \zeta \geq 0$, where $Z$ is aggregated over all sinks. In particular, any such flow terminates before $\theta_0 + \tPmax + U$. 
\end{lemma}	

\begin{proof}
	The proof is essentially the same as for the single-commodity case since the flow rates of a multi-commodity flow aggregated over all commodities satisfy the feasibility constraints of a single commodity flow (see \cite{GHS18} for a formal definition of multi-commodity flows). We only have to adjust the definition of the pessimistic remaining distance by setting 
		$$\distP{v} \coloneqq \max\set{\abs{P} | P \text{ is a $v$-$t$ path for some sink } t \in V}.$$
	Furthermore, to ensure that we still have $Z(\theta) = F_{\leq 0}(\theta)$ for all times $\theta$, we add a new sink node for every sink node which currently does not have a pessimistic distance of $0$. We then connect the old sink node (which does not act as a sink any more) to the new one by an edge of infinite capacity and physical travel time of $1$. Note, that this does not change the length of the longest path towards a sink in the network and only increases the termination time. Thus, a bound on the makespan in the new network is also a bound on the makespan in the old network.
\end{proof}

\ifsubm
\begin{remark}
	Shortly after this paper's initial conference version \citeauthor*{AtomicBound2021} provided in \cite[Theorem 2]{AtomicBound2021} a matching termination bound for the atomic routing model, i.e. the discretized version of flows over time, using a completely different proof approach. Since \citeauthor*{Convergence2021} have shown in \cite[Theorem 5.5]{Convergence2021} that any flow over time can be considered as limit of finer and finer discretizations (for which then the result from \cite{AtomicBound2021} applies), this yields an alternative proof for the second part of \Cref{lemma:TerminationAcyclicNetworksMulti}. 
\end{remark}
\else
\begin{remark}
	The bound on the termination time for flows over time in acyclic networks given in this lemma is exactly the same as the one found by \citeauthor*{AtomicBound2021} in \cite[Theorem 2]{AtomicBound2021} for the atomic routing model, i.e. the discretized version of flows over time. Since \citeauthor*{Convergence2021} have shown in \cite[Theorem 5.5]{Convergence2021} that any flow over time can be considered as limit of finer and finer discretizations (for which then the result from \cite{AtomicBound2021} applies), this yields an alternative proof for the second part of \Cref{lemma:TerminationAcyclicNetworksV2}. Since both results actually hold for multi-source multi-sink settings, this even shows a slightly more general result -- namely, that in acyclic networks the bound of $\theta_0 + \tPmax + U$ holds even in the multi-commodity setting.
\end{remark}

\begin{proof}[Proof sketch.]
	First note that we can assume without loss of generality that all edge travel times and capacities are $1$. This can be achieved by replacing every edge $e$ by $\nu_e$ parallel paths each consisting of $\tau_e$ edges with capacity and travel time equal to $1$. Then every feasible flow in the original network can be translated into a feasible flow in the transformed network such that the arrival times of all particles at all nodes (and in particular at its sink) remain the same.
	
	Now, let $f$ be any such feasible flow. Then, by \cite[Theorem 5.5]{Convergence2021}, there exists a sequence $(f_n)_{n \in \INs}$ of atomic packet routings such that in packet routing $f_n$ the time steps have length $\frac{1}{n}$ and the packets a size of $\frac{1}{n^2}$. Furthermore for any particle of flow $f$ the arrival time of its corresponding packet in the packet routings at any node converges to its arrival time in the flow $f$. In particular, denoting by $\termTime(f_n)$ the time the last packet arrives at it sink in $f_n$, we have $\lim_{n \to \infty} \termTime(f_n) = \termTime(f)$. 

	Thus, it suffices to give a bound on the $\termTime(f_n)$. So, let $n$ be fixed. Then, by again rescaling the network and the flow, we can obtain a routing with unit sized packets and unit time steps. The edge capacities are then $n$ and by further subdividing edges into paths we can also guarantee that all of them have a travel time of $1$. Furthermore, by adding for every packet a path of lengths $n\xi$ ending with its source node, where $\xi$ is the time it enters the network in $f_n$, and defining the start of this path as its new source, we can also ensure that all packets are enter the network simultaneously at time $0$. Now, take the last packet arriving at the sink, then by \cite[Theorem 2]{AtomicBound2021} we know that its arrival time is at most $L + \floor{\frac{\floor{n^2U}-1}{n}}$, where $L$ is the length of the longest possible path this particle can take, i.e. we have $L \leq n\theta_0 + n\tPmax$. Thus, we get
		\begin{align*}
			\termTime(f) 
				&= \lim_{n \to \infty}\termTime(f_n) \leq \lim_{n \to \infty}\left(L + \floor{\frac{\floor{n^2U}-1}{n}}\right) \\
				&\leq \theta_0 + \tPmax + \lim_{n \to \infty}\frac{1}{n}\floor{\frac{\floor{n^2U}-1}{n}} = \theta_0 + \tPmax + U.\qedhere
		\end{align*}
\end{proof}
\fi

Similarly to the proof of termination in \cite[Theorem 4.6]{GHS18} we will apply our result for feasible flows in acyclic graphs to IDE flows in general graphs by using the fact (\cite[Lemma 4.4]{GHS18}) that, whenever the total flow volume in a subgraph is small enough, only the physically shortest paths in this subgraph can be active. Since these edges form an acyclic subgraph, for an IDE flow we can then apply \Cref{lemma:TerminationAcyclicNetworksV2}. 	

For the following proof we will look at a particular type of subgraph, which we will call a sink-like subgraph: a subgraph containing all physically shortest paths from its nodes towards the sink, with a sufficiently low flow volume at the beginning of some interval as well as a low inflow into this subgraph over the course of said interval.

\begin{defn}
	An induced subgraph $T \subseteq G$ is a \emph{sink-like subgraph} on an interval $[\theta_1,\theta_2]$ if the following two properties hold:
	\begin{itemize}
		\item For each node $v \in V(T)$ all physically shortest $v$-$t$ paths are contained in $T$.
		\item The total flow volume present in $T$ at the beginning of the interval plus the cumulative inflow into this subgraph during the interval is less than $\frac{1}{2}$, i.e. $T$ satisfies
			\[\vol{T}{\theta_1}{\theta_2} \coloneqq \sum_{e \in E(T)}\edgeLoad[e](\theta_1) +  \sum_{e \in \delta^-_T}\inflowOver{\theta_1}{\theta_2}  + \sum_{v \in V(T)\setminus\set{t}}\int_{\theta_1}^{\theta_2}u_v(\theta)d\theta < \frac{1}{2}.\]
	\end{itemize}
\end{defn} 

Here $\delta^-_T \coloneqq \set{vw \in E | v \notin V(T), w \in V(T)}$ denotes the set of edges entering the subgraph $T$.
We will now show that inside a sink-like subgraph only physically shortest paths towards the sink can be active.

\begin{lemma}\label{lemma:InSinkLikeOnlyShortestPathsAreActive}
	Let $T$ be a sink-like subgraph on an interval $[\theta_1,\theta_2]$. Then during this interval only physically shortest paths towards $t$ can be active.	
\end{lemma}

\begin{proof}
	Let $v \in V(T)$ be a node in $T$ and $\hat{\theta} \in [\theta_1,\theta_2]$, then for any physically shortest $v$-$t$ path $P$ we have:
		\begin{align*}
			&\sum_{e \in P}\edgeLoad[e](\hat{\theta}) \leq \sum_{e \in E(T)}\edgeLoad[e](\hat{\theta}) 
			\overset{\text{\Cref{lem:GeIsTotalFlowInGraph}}}{=}\quad \sum_{e \in \delta^-_T}F_e^-(\hat{\theta}) + \sum_{v \in V(T) \setminus \set{t}}U_v(\hat{\theta}) - \sum_{e \in \delta^+_T}F_e^+(\hat{\theta}) - Z(\hat{\theta}) \\
			&\quad\quad\leq \sum_{e \in \delta^-_T}F_e^-(\theta_1) + \sum_{e \in \delta^-_T}\inflowOver{\theta_1}{\theta_2}  + \sum_{v \in V(T)\setminus\set{t}}\int_{\theta_1}^{\hat\theta}u_v(\theta)d\theta + \sum_{v \in V(T) \setminus \set{t}}U_v(\theta_1) \\
				&\quad\quad\quad\quad\quad- \sum_{e \in \delta^+_T}F_e^+(\theta_1) - Z(\theta_1) \\
			&\quad\quad\overset{\mathclap{\text{\Cref{lem:GeIsTotalFlowInGraph}}}}{=}\quad \sum_{e \in E(T)}\edgeLoad[e](\theta_1)+  \sum_{e \in \delta^-_T}\inflowOver{\theta_1}{\theta_2}  + \sum_{v \in V(T)\setminus\set{t}}\int_{\theta_1}^{\theta_2}u_v(\theta)d\theta<  \frac{1}{2}.
		\end{align*}
	Thus, by \cite[Lemma 4.4]{GHS18} we get that all active $v$-$t$ paths are also physically shortest paths towards $t$ (we have $v_{\min}, \tauMinDiff \geq 1$, since both travel times and capacities are whole numbers here).
\end{proof}

Together with \Cref{cor:FlowVolumeReduction} and \Cref{lemma:TerminationAcyclicNetworksV2} this implies that any IDE flow will terminate once the whole graph is sink-like.

\begin{corollary}\label{cor:TerminationIfGIsSinkLike}
	Let $f$ be an IDE flow and $\tilde{\theta} \geq \theta_0$ such that the whole graph $G$ is sink-like at time $\tilde{\theta}$.
	Then, the flow terminates before $\tilde{\theta} + \tPmax + \frac{1}{2}$.
\end{corollary}

\begin{proof}
	\Cref{cor:FlowVolumeReduction} together with $\delta_G^+=\delta_G^-=\emptyset$ shows 
		\[\vol{G}{\tilde{\theta}}{\theta} = \edgeLoad(\tilde{\theta}) = \vol{G}{\tilde{\theta}}{\tilde{\theta}} < \frac{1}{2}\]
	for all times $\theta \geq \tilde{\theta}$, i.e. $G$ remains sink-like forever after $\tilde{\theta}$. So, by \Cref{lemma:InSinkLikeOnlyShortestPathsAreActive}, after $\tilde{\theta}$ the flow $f$ can only use edges on physically shortest paths towards $t$. As those edges form a time independent acyclic subgraph this shows that $f$ only uses an acyclic subgraph of $G$ and by \Cref{lemma:TerminationAcyclicNetworksV2} such a flow terminates after an additional time of at most $\tPmax + \vol{G}{\tilde{\theta}}{\tilde{\theta}} < \tPmax + \frac{1}{2}$.
\end{proof}

To get an upper bound on the makespan of an IDE flow it now suffices to find a large enough time horizon such that it contains at least one point in time where the whole graph is sink-like. To determine such a time, we first show that if we have a sink-like subgraph over a sufficiently long period of time, we can extend this subgraph to a larger sink-like subgraph over a slightly smaller subinterval. Note, that the proof of \cite[Theorem 4.6]{GHS18} uses a similar strategy, but is non-constructive and, therefore, only establishes the existence of a termination time without revealing anything about the length of this time. Thus, a more thorough analysis is needed here.

\begin{lemma}\label{lemma:ExtensionOfSinkLikeSubgraphs}
	Let $T \subsetneq G$ be an induced subgraph, $v$ the closest node to $t$ not in $T$ and $T'$ the subgraph of $G$ induced by $V(T)\cup\set{v}$. Let $\theta_1 \in \IR_{\geq 0}$ be some time, $\theta_2 \coloneqq \theta_1+\sum_{e \in E\setminus E(T)}(\tau_e+ \frac{1}{2\nu_e})$ and $\theta'_2 \coloneqq \theta_1+\sum_{e \in E\setminus E(T')}(\tau_e+ \frac{1}{2\nu_e})$. 
	
	If $T$ is sink-like on $[\theta_1,\theta_2]$, then $T'$ is a sink-like subgraph on $[\theta_1,\theta'_2]$.
\end{lemma}

\begin{proof}
	\newcommand{\colInT}{green!70!blue!70}
	\newcommand{\colVtoT}{red!50!blue!80}
	\newcommand{\colVtoG}{black}
	\newcommand{\colGtoV}{blue!70}
	\newcommand{\colGtoT}{blue!30}
	\newcommand{\colTtoV}{red}
	\newcommand{\colTtoG}{black}
	Since it is clear that $T'$ fulfills the first property of being sink-like (by the choice of $v$), we only need to show that $\vol{T'}{\theta_1}{\theta_2'} \leq \vol{T}{\theta_1}{\theta_2}$, from which the lemma follows immediately (as $T$ is sink-like on $[\theta_1,\theta_2]$). More precisely we will show that the flow volume on edges between $v$ and $T$ (i.e. edges in $E(T') \setminus E(T) = (\delta_v^+ \cap \delta_T^-) \cup (\delta_v^- \cap \delta_T^+)$) at time $\theta_1$ as well as the cumulative inflow into $v$ over the interval $[\theta_1,\theta_2']$ is already accounted for by the cumulative inflow into $T$ on the interval $[\theta_1,\theta_2]$ via edges from $v$ to $T$. This is formalized in the following three claims:

	\begin{figure}[ht!]\centering
		\begin{minipage}[c]{0.45\textwidth}
			\begin{adjustbox}{max width=\textwidth}
				\newcommand{\colTbackground}{black!20!white}

\tikzset{
	styleVtoT/.style={edge,dashed,\colVtoT},
	styleTtoV/.style={edge,dashdotted,\colTtoV},
	styleInT/.style={edge,\colInT},
	styleGtoT/.style={edge,loosely dotted,\colGtoT},
	styleGtoV/.style={edge,dotted,\colGtoV},
	styleTtoG/.style={edge,loosely dotted,\colTtoG},
	styleVtoG/.style={edge,loosely dotted,\colVtoG}
}

\begin{tikzpicture}[scale=1.0]	
	\node () at (1,4) {\Large$G$};
	\fill[rounded corners=10, color=\colTbackground] (3.5,4.5) rectangle (8.5, -.5);
	\node() at (8,4) {\Large$T$};
	\draw[rounded corners=10, ultra thick, color=\colTbackground] (3.5,4.5) -- (8.5,4.5) -- (8.5, -.5) -- (3.5,-.5) -- (1.2,2) -- cycle;
	\node() at (3,2.5) {\Large$T'$};

	\node[namedVertex, fill=white] (t) at (8, 2) {$t$};
	\node[namedVertex, fill=white] (v) at (2,2) {$v$};
	\node[vertex, fill=white] (1) at (4,4) {};
	\node[vertex, fill=white] (2) at (7,4) {};
	\node[vertex, fill=white] (3) at (5,2) {};
	\node[vertex, fill=white] (4) at (4,0) {};
	\node[vertex, fill=white] (5) at (8,0) {};
	
	\draw[styleVtoT] (v) -- (3);
	\draw[styleVtoT] (v) -- (4);
	\draw[styleVtoG] (v) -- ++(-1,-1);
	\draw[styleGtoV,<-] (v) -- ++(-1.5,0);
	\draw[styleGtoV,<-] (v) -- ++(-1,1);
	
	\draw[styleInT] (1) -- (2);
	\draw[styleInT] (1) -- (3);
	\draw[styleTtoV] (1) -- (v);
	\draw[styleTtoG] (1) -- ++(-.8,.8);
	\draw[styleGtoT,<-] (1) -- ++(-1.5,0);
	
	\draw[styleInT] (2) -- (3);
	
	\draw[styleInT] (3) -- (t);
	
	\draw[styleInT] (4) -- (3);
	\draw[styleGtoT,<-] (4) -- ++(-1.5,0);
	
	\draw[styleInT] (5) -- (4);
	\draw[styleInT] (5) -- (t);

\end{tikzpicture}
			\end{adjustbox}
		\end{minipage}\hfill
		\begin{minipage}[c]{0.5\textwidth}
			\caption{A sink-like subgraph $T$ and a closest node $v \in V\setminus V(T)$ as in the statement of \Cref{lemma:ExtensionOfSinkLikeSubgraphs}. By \Cref{lemma:ExtensionOfSinkLikeSubgraphs:Claim:ReachingVFast} all flow on the {\color{\colTtoV}dash-dotted edge} from $T$ to $v$ will reach $v$ before time $\theta_v$. By \Cref{lemma:ExtensionOfSinkLikeSubgraphs:Claim:FromVOnlyToT} between $\theta_1$ and $\theta_v$ all flow reaching $v$ (either via the {\color{\colGtoV}dotted} or via the {\color{\colTtoV}dash-dotted edges}) will travel towards $T$ from there (i.e. enter one of the {\color{\colVtoT}dashed edges}). By \Cref{lemma:ExtensionOfSinkLikeSubgraphs:Claim:ReachingTFast} the {\color{\colVtoT}dashed edges} will never carry a larger flow volume than $\frac{1}{2}$ between $\theta_1$ and $\theta_v$ and all flow particles using these edges within this time interval will reach $T$ before $\theta_2$.
			}\label{fig:ExtendingSinkLike}
		\end{minipage}
	\end{figure}
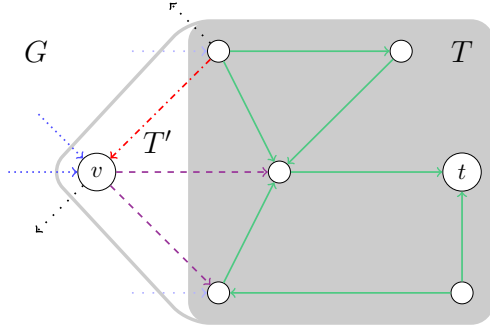
	
	\begin{claim}\label{lemma:ExtensionOfSinkLikeSubgraphs:Claim:ReachingVFast}
		All flow on edges from $T$ to $v$ ({\color{\colTtoV}dash-dotted edges} in \Cref{fig:ExtendingSinkLike}) at time $\theta_1$ reaches $v$ before $\theta_v \coloneqq \theta_1+\sum_{e \in E\setminus E(T')}(\tau_e+ \frac{1}{2\nu_e})+\sum_{e \in \delta_v^-\cap\delta_T^+}(\tau_e+ \frac{1}{2\nu_e}) \leq \theta_2$, i.e.
			\[\edgeLoad[e](\theta_1) \leq \inflowOver{\theta_1}{\theta_v}\text{ for all } e \in \delta_T^+ \cap \delta_v^-.\]
	\end{claim}

	\begin{claim}\label{lemma:ExtensionOfSinkLikeSubgraphs:Claim:FromVOnlyToT}
		All flow reaching $v$ (due to the network inflow at $v$ or from $T$ or $G\setminus T'$, i.e. via the {\color{\colTtoV}dash-dotted} or via the {\color{\colGtoV}dotted} edges in \Cref{fig:ExtendingSinkLike}) between $\theta_1$ and $\theta_v$ will enter an edge towards $T$ ({\color{\colVtoT}dashed} edges in \Cref{fig:ExtendingSinkLike}), i.e.
			\[\sum_{e \in \delta_v^-}\inflowOver{\theta_1}{\theta_v} + \int_{\theta_1}^{\theta_v}u_v(\theta)d\theta
				= \sum_{e \in \delta_v^+\cap\delta_T^-}\outflowOver{\theta_1}{\theta_v}.\]
	\end{claim}

	\begin{claim}\label{lemma:ExtensionOfSinkLikeSubgraphs:Claim:ReachingTFast}
		For any edge from $v$ to $T$ ({\color{\colVtoT}dashed} edges in \Cref{fig:ExtendingSinkLike}) the volume of flow currently traveling on this edge at any time $\theta \in [\theta_1,\theta_v]$ is less than $\frac{1}{2}$, i.e.
			\[\edgeLoad[e](\theta) < \frac{1}{2} \text{ for all } e \in \delta_v^+\cap\delta_T^- \text{ and } \theta \in [\theta_1,\theta_v].\]
		Additionally all this flow will reach $T$ before $\theta_2$, i.e.
			\[\edgeLoad[e](\theta_1) + \outflowOver{\theta_1}{\theta_v} \leq \inflowOver{\theta_1}{\theta_2} \text{ for all } e \in \delta_v^+\cap\delta_T^-.\]
	\end{claim}

	From \Cref{lemma:ExtensionOfSinkLikeSubgraphs:Claim:FromVOnlyToT,lemma:ExtensionOfSinkLikeSubgraphs:Claim:ReachingTFast,lemma:ExtensionOfSinkLikeSubgraphs:Claim:ReachingVFast} we then directly get
	\Crefformat{claim}{Cl. #2#1#3}
	\begin{align*}
		&\vol{T'}{\theta_1}{\theta'_2} \leq \vol{T'}{\theta_1}{\theta_v} \\
		&= \sum_{e \in E(T')}\edgeLoad[e](\theta_1)+  \sum_{e \in \delta^-_{T'}}\inflowOver{\theta_1}{\theta_v} +\!\!\! \sum_{v \in V(T')\setminus\set{t}}\int_{\theta_1}^{\theta_v}u_v(\theta)d\theta \\
		&\quad= {\color{\colInT}\sum_{e \in E(T)}\edgeLoad[e](\theta_1)} + {\color{\colTtoV}\sum_{e \in \delta^+_T\cap\delta^-_v}\edgeLoad[e](\theta_1)} + {\color{\colVtoT}\sum_{e \in \delta_v^+\cap\delta_T^-}\edgeLoad[e](\theta_1)} + {\color{\colGtoT}\sum_{e \in \delta_T^-\setminus\delta_v^+}\inflowOver{\theta_1}{\theta_v}} \\
			&\quad\quad\quad+ {\color{\colGtoV}\sum_{e \in \delta_v^-\setminus\delta_T^+}\inflowOver{\theta_1}{\theta_v}} +\!\!\! \sum_{v \in V(T')\setminus\set{t}}\int_{\theta_1}^{\theta_v}u_v(\theta)d\theta \\
		&\quad\overset{\mathclap{\text{\Cref{lemma:ExtensionOfSinkLikeSubgraphs:Claim:ReachingVFast}}}}{\leq} 
		{\color{\colInT}\sum_{e \in E(T)}\edgeLoad[e](\theta_1)} + {\color{\colTtoV}\sum_{e \in \delta^+_T\cap\delta^-_v}\inflowOver{\theta_1}{\theta_v}} + {\color{\colVtoT}\sum_{e \in \delta_v^+\cap\delta_T^-}\edgeLoad[e](\theta_1)} + {\color{\colGtoT}\sum_{e \in \delta_T^-\setminus\delta_v^+}\inflowOver{\theta_1}{\theta_v}} \\
			&\quad\quad\quad+ {\color{\colGtoV}\sum_{e \in \delta_v^-\setminus\delta_T^+}\inflowOver{\theta_1}{\theta_v}} +\!\!\! \sum_{v \in V(T')\setminus\set{t}}\int_{\theta_1}^{\theta_v}u_v(\theta)d\theta \\
		&\quad\overset{\mathclap{\text{\Cref{lemma:ExtensionOfSinkLikeSubgraphs:Claim:FromVOnlyToT}}}}{=} 
		{\color{\colInT}\sum_{e \in E(T)}\edgeLoad[e](\theta_1)} + \!\!\!{\color{\colVtoT}\sum_{e \in \delta_v^+\cap\delta_T^-}\edgeLoad[e](\theta_1)} +\!\!\! {\color{\colGtoT}\sum_{e \in \delta_T^-\setminus\delta_v^+}\inflowOver{\theta_1}{\theta_v}} +\!\!\! {\color{\colVtoT}\sum_{e \in \delta_v^+\cap\delta_T^-}\outflowOver{\theta_1}{\theta_v}} \\
			&\quad\quad\quad+\!\!\! \sum_{v \in V(T)\setminus\set{t}}\int_{\theta_1}^{\theta_v}u_v(\theta)d\theta  \\
		&\quad\overset{\mathclap{\text{\Cref{lemma:ExtensionOfSinkLikeSubgraphs:Claim:ReachingTFast}}}}{\leq} 
		{\color{\colInT}\sum_{e \in E(T)}\edgeLoad[e](\theta_1)} + {\color{\colGtoT}\sum_{e \in \delta_T^-\setminus\delta_v^+}\inflowOver{\theta_1}{\theta_v}} +{\color{\colVtoT}\sum_{e \in \delta_v^+\cap\delta_T^-}\inflowOver{\theta_1}{\theta_2}}  +\!\!\! \sum_{v \in V(T)\setminus\set{t}}\int_{\theta_1}^{\theta_v}u_v(\theta)d\theta \\
		&\quad\leq\sum_{e \in E(T)}\edgeLoad[e](\theta_1) +  \sum_{e \in \delta_T^-}\inflowOver{\theta_1}{\theta_2}  + \sum_{v \in V(T)\setminus\set{t}}\int_{\theta_1}^{\theta_2}u_v(\theta)d\theta = \vol{T}{\theta_1}{\theta_2} < \frac{1}{2},
	\end{align*}\Crefformat{claim}{Claim~#2#1#3}
	proving 
	that $T'$ is indeed sink-like on $[\theta_1,\theta_2']$. The proofs of these claims are relatively straightforward calculations using \cite[Lemma 4.4]{GHS18} and \Cref{lemma:FlowOnEdgesLeavesFastVar}. Note, that the proofs have to be done in reverse order, as the proof of \Cref{lemma:ExtensionOfSinkLikeSubgraphs:Claim:FromVOnlyToT} uses \Cref{lemma:ExtensionOfSinkLikeSubgraphs:Claim:ReachingTFast} and the proof of \Cref{lemma:ExtensionOfSinkLikeSubgraphs:Claim:ReachingVFast} uses both \Cref{lemma:ExtensionOfSinkLikeSubgraphs:Claim:FromVOnlyToT,lemma:ExtensionOfSinkLikeSubgraphs:Claim:ReachingTFast}.
	\begin{proofClaim}[Proof of \Cref{lemma:ExtensionOfSinkLikeSubgraphs:Claim:ReachingTFast}]
		Let $e \in \delta_v^+\cap\delta_T^-$ be any edge from $v$ to $T$. For the first part we assume by way of contradiction that there exists a time $\theta' \in [\theta_1,\theta_v]$ such that $\edgeLoad[e](\theta') \geq \frac{1}{2}$. By \Cref{lemma:FlowOnEdgesLeavesFastVar} we then have $F_e^-(\theta'+\frac{1}{2\nu_e} + \tau_e) - F_e^-(\theta') \geq \frac{1}{2}$ which, together with $\theta'+\frac{1}{2\nu_e} + \tau_e \leq \theta_v+\frac{1}{2\nu_e} + \tau_e \leq \theta_2$, shows that
		\[\vol{T}{\theta_1}{\theta_2} \geq \int_{\theta_1}^{\theta_2}f_e^-(\theta)d\theta \geq \int_{\theta'}^{\theta'+\frac{1}{2\nu_e} + \tau_e}f_e^-(\theta)d\theta = F_e^-(\theta'+\frac{1}{2\nu_e} + \tau_e) - F_e^-(\theta') \geq \frac{1}{2},\]
		a contradiction to $T$ being sink-like on the interval $[\theta_1,\theta_2]$. Thus, we indeed have $\edgeLoad[e](\theta') < \frac{1}{2}$ for all $\theta' \in [\theta_1,\theta_v]$.
		
		Now, using this first part of the Claim for $\theta' = \theta_v$ we get $\theta_v+\frac{\edgeLoad[e](\theta_v)}{\nu_e}+\tau_e < \theta_v + \frac{1}{2\nu_e} \leq \theta_2$ and therefore
		\begin{align}\label{lemma:ExtensionOfSinkLikeSubgraphs:Claim:ReachingTFast:eq1}
		\edgeLoad[e](\theta_v) \overset{\text{\Cref{lemma:FlowOnEdgesLeavesFastVar}}}{\leq} F_e^-(\theta_v+\frac{\edgeLoad[e](\theta_v)}{\nu_e}+\tau_e)-F_e^-(\theta_v) \leq F_e^-(\theta_2)-F_e^-(\theta_v).
		\end{align}
		Finally we get
		\begin{align*}
		\outflowOver{\theta_1}{\theta_v} + \edgeLoad[e](\theta_1)
		&= F_e^+(\theta_v)-F_e^+(\theta_1) + F_e^+(\theta_1) - F_e^-(\theta_1) \\
		&= F_e^+(\theta_v)-F_e^-(\theta_v) + F_e^-(\theta_v) - F_e^-(\theta_1) \\
		&= \edgeLoad[e](\theta_v) + F_e^-(\theta_v) - F_e^-(\theta_1)\\ 
		&\overset{\text{\eqref{lemma:ExtensionOfSinkLikeSubgraphs:Claim:ReachingTFast:eq1}}}{\leq} F_e^-(\theta_2)-F_e^-(\theta_v)+ F_e^-(\theta_v) - F_e^-(\theta_1)\\
		&=  F_e^-(\theta_2)-F_e^-(\theta_1) = \inflowOver{\theta_1}{\theta_2}.\qedhere
		\end{align*}		
	\end{proofClaim}
	
	\begin{proofClaim}[Proof of \Cref{lemma:ExtensionOfSinkLikeSubgraphs:Claim:FromVOnlyToT}]
		Since we already know, that $T'$ fulfills the first property for being sink-like, any physically shortest $v$-$t$ path consists of one edge in $\delta^+_v\cap \delta^-_T$ and after this only edges inside $T$. Combining the first part of \Cref{lemma:ExtensionOfSinkLikeSubgraphs:Claim:ReachingTFast} and the second property of being sink-like (for $T$) we get, that any such path always carries a flow volume of less than $1$ for the whole interval $[\theta_1,\theta_v]$. By \cite[Lemma 4.4]{GHS18} this means that all active $v$-$t$ paths are also physically shortest and therefore $\delta_v^+ \cap E_\theta \subseteq \delta_v^+ \cap \delta^-_T$. Together with the flow conservation constraint \eqref{eq:FeasibleFlow-FlowConservation} this immediately gives us
		\[
		\sum_{e \in \delta_v^-}\inflowOver{\theta_1}{\theta_v} + \int_{\theta_1}^{\theta_v}u_v(\theta) d\theta = \sum_{e \in \delta_v^+}\outflowOver{\theta_1}{\theta_v} = \sum_{e \in \delta_v^+ \cap \delta^-_T}\outflowOver{\theta_1}{\theta_v}.\qedhere
		\]
	\end{proofClaim}

	\begin{proofClaim}[Proof of \Cref{lemma:ExtensionOfSinkLikeSubgraphs:Claim:ReachingVFast}]
		Let $e \in \delta_T^+\cap\delta_v^-$ be any edge from $T$ to $v$. As in the proof of \Cref{lemma:ExtensionOfSinkLikeSubgraphs:Claim:ReachingTFast} we first show that, $\edgeLoad[e](\theta_1) < \frac{1}{2}$. So, we assume $\edgeLoad[e](\theta_1) \geq \frac{1}{2}$ and get $F_e^-(\theta_1+\frac{1}{2\nu_e} + \tau_e) - F_e^-(\theta_1) \geq \frac{1}{2}$ from  \Cref{lemma:FlowOnEdgesLeavesFastVar}. Together with  $\theta_1+\frac{1}{2\nu_e} + \tau_e \leq \theta_1+\frac{1}{2\nu_e} + \tau_e \leq \theta_2$ this shows
		\begin{align*}
		\vol{T}{\theta_1}{\theta_2} 
			&\geq \sum_{e \in \delta_v^+\cap\delta_T^-}\inflowOver{\theta_1}{\theta_2} \overset{\text{\Cref{lemma:ExtensionOfSinkLikeSubgraphs:Claim:ReachingTFast}}}{\geq} \sum_{e \in \delta_v^+\cap\delta_T^-}\outflowOver{\theta_1}{\theta_v} \\
			&\overset{\mathclap{\text{\Cref{lemma:ExtensionOfSinkLikeSubgraphs:Claim:FromVOnlyToT}}}}{=}\quad\sum_{e \in \delta_v^-}\inflowOver{\theta_1}{\theta_v} + \int_{\theta_1}^{\theta_v}u_v(\theta) d\theta \geq \inflowOver{\theta_1}{\theta_v} \\
			&\geq F_e^-(\theta_1+\frac{1}{2\nu_e} + \tau_e) - F_e^-(\theta_1) \geq \frac{1}{2},
		\end{align*}
		which is a contradiction to $T$ being sink-like on $[\theta_1,\theta_2]$. So, we do indeed have $\edgeLoad[e](\theta_1) < \frac{1}{2}$, thus $\theta_1+\frac{\edgeLoad[e](\theta_1)}{\nu_e}+\tau_e\leq \theta_1+\frac{1}{2\nu_e}+\tau_e \leq \theta_v$ and therefore
		\[
		\edgeLoad[e](\theta_1) \overset{\text{\Cref{lemma:FlowOnEdgesLeavesFastVar}}}{\leq} F_e^-(\theta_1+\frac{\edgeLoad[e](\theta_1)}{\nu_e}+\tau_e)-F_e^-(\theta_1) \leq F_e^-(\theta_v)-F_e^-(\theta_1) = \inflowOver{\theta_1}{\theta_v}\qedhere
		\]
	\end{proofClaim}
	This concludes the proof of the lemma.
\end{proof}

\begin{theorem}\label{thm:Termination_SingleSink2}
	For multi-source single-sink networks, any IDE flow over time terminates before $\hat{\theta} \coloneqq \theta_0 + 2U \sum_{e\in E}(\tau_e+\frac{1}{2\nu_e}) + \tPmax + \frac{1}{2}$ and has total travel times of at most $\ceil{3U}U(2\tPmax+\sum_{e\in E}(\tau_e+\frac{1}{2\nu_e})+1)$.
\end{theorem}

\begin{proof}
	Starting with the subgraph consisting only of the sink node $t$ (which trivially contains all shortest paths towards $t$) and iteratively applying \Cref{lemma:ExtensionOfSinkLikeSubgraphs} we immediately get
	\begin{claim}\label{thm:Termination_SingleSink2:Claim}
		If for some interval $[a,b]$ of length at least $\sum_{e\in E}(\tau_e+\frac{1}{2\nu_e})$ the sink node $t$ has a total cumulative inflow of less than $\frac{1}{2}$, then the whole graph is sink-like for $[a,b-\sum_{e\in E}(\tau_e+\frac{1}{2\nu_e})]$.
		\hfill$\blacksquare$
	\end{claim}
	
	Since all flow reaching $t$ vanishes from the network there can be at most $2U$ (pairwise disjoint) intervals of length $\sum_{e\in E}(\tau_e+\frac{1}{2\nu_e})$ with inflow of at least $\frac{1}{2}$ into $t$. Thus, there must be some time $\theta_0 \leq \tilde{\theta} \leq \theta_0 + 2U \sum_{e\in E}(\tau_e+\frac{1}{2\nu_e})$ which is the beginning of an interval of length $\sum_{e\in E}(\tau_e+\frac{1}{2\nu_e})$ with total cumulative inflow of less than $\frac{1}{2}$ into $t$. So, by \Cref{thm:Termination_SingleSink2:Claim}, the whole graph is sink-like at $\tilde{\theta}$, which, by \Cref{cor:TerminationIfGIsSinkLike}, implies that the flow terminates before $\tilde{\theta} + \tPmax + \frac{1}{2} \leq \hat{\theta}$. This shows the bound on the makespan.

	For the total travel times we partition the time interval $[0,\termTime(f)]$ into $\ceil{3U}$ intervals $[a_0,a_1], \dots, [a_{\ceil{3U}-1},a_{\ceil{3U}}]$ by setting $a_0 \coloneqq 0$ and then recursively defining
	\begin{align*}
		a_{k+1} \coloneqq \sup\Set{\theta \in [a_k,\termTime(f)] \mid Z(\theta)-Z(a_k) \leq \frac{1}{3}} \text{ for } k=1, \dots, \ceil{3U}.
	\end{align*}
	Since $Z$ is continuous all these intervals except for possibly the last one have a cumulative inflow of exactly $\frac{1}{3}$ into the sink and, therefore, together they form a partition of $[0,\termTime(f)]$. We now distinguish two different kinds of intervals: We say an interval $[a_{i-1},a_i]$ is \emph{short}, if $a_i - a_{i-1} \leq c \coloneqq \sum_{e\in E}(\tau_e+\frac{1}{2\nu_e})+\tPmax+\frac{1}{2}$ and \emph{long} otherwise. For any such long interval $[a_{i-1},a_i]$ \Cref{thm:Termination_SingleSink2:Claim} guarantees that the whole graph is sink-like on $[a_{i-1},a_i-\sum_{e\in E}(\tau_e+\frac{1}{2\nu_e})]$. Therefore, by \Cref{lemma:InSinkLikeOnlyShortestPathsAreActive} flow particles only use physically shortest paths towards the sink, which together form a fixed acyclic subgraph. Thus, for this interval we can apply \Cref{lemma:TerminationAcyclicNetworksV2}, i.e. for any time $\theta \in [a_{i-1},a_i-\sum_{e\in E}(\tau_e+\frac{1}{2\nu_e})-\tPmax-\frac{1}{2}] = [a_{i-1},a_i-c]$ we have 
		\begin{align}\label{eq:boundOnArrivingFlow}
			Z\left(\theta+\tPmax+\frac{1}{2}\right) \geq Z(\theta) + \min\Set{\edgeLoad(\theta),\frac{1}{2}} = Z(\theta) + \edgeLoad(\theta) 
		\end{align}
	where $\edgeLoad(\theta) \leq \vol{G}{\theta}{\theta} < \frac{1}{2}$ follows from the fact that $G$ is sink-like at $\theta$. This allows us to find the following upper bound on the total travel times incurred during the first part of any long interval $[a_{i-1},a_i]$:
	\begin{align*}
		&\int_{a_{i-1}}^{a_i-c}\edgeLoad(\theta)d\theta 
			\overset{\text{\eqref{eq:boundOnArrivingFlow}}}{\leq} \int_{a_{i-1}}^{a_i-c}Z\left(\theta+\tPmax+\frac{1}{2}\right) - Z(\theta)d\theta \\
		&\quad\quad= \int_{a_{i-1}}^{a_i-c}Z\left(\theta+\tPmax+\frac{1}{2}\right)d\theta - \int_{a_{i-1}}^{a_i-c}Z(\theta)d\theta \\
		&\quad\quad=\int_{a_{i-1}+\tPmax+\frac{1}{2}}^{a_i-c+\tPmax+\frac{1}{2}}Z(\theta)d\theta - \int_{a_{i-1}}^{a_i-c}Z(\theta)d\theta \\
		&\quad\quad= \int_{a_i}^{a_i-c+\tPmax+\frac{1}{2}}Z(\theta)d\theta - \int_{a_{i-1}}^{a_{i-1}+\tPmax+\frac{1}{2}}Z(\theta)d\theta \\
		&\quad\quad\leq \int_{a_i}^{a_i-c+\tPmax+\frac{1}{2}}Z(\theta)d\theta \overset{\text{\Cref{lem:GeIsTotalFlowInGraph}}}{\leq} \int_{a_i}^{a_i-c+\tPmax+\frac{1}{2}}Ud\theta = U\left(\tPmax+\frac{1}{2}\right).
	\end{align*}
	Together with the trivial upper bounds for the second parts of long intervals
	\begin{align*}
		\int_{a_i-c}^{a_i}\edgeLoad(\theta)d\theta \overset{\text{\Cref{lem:GeIsTotalFlowInGraph}}}{\leq} \int_{a_i-c}^{a_i}Ud\theta = Uc
	\end{align*}
	as well as for any short interval $[a_{i-1},a_i]$
	\begin{align*}
		\int_{a_{i-1}}^{a_i}\edgeLoad(\theta)d\theta \overset{\text{\Cref{lem:GeIsTotalFlowInGraph}}}{\leq} \int_{a_{i-1}}^{a_i}Ud\theta = (a_i-a_{i-1})U \leq Uc
	\end{align*}
	this gives us the desired bound on the sum of travel time for the whole flow, where we denote by  $K \coloneqq \set{i \in [\ceil{3U}] \mid a_i - a_{i-1} > \frac{1}{3}}$ the set of all long intervals:
	\begin{align*}
		\sumOfTraveltimes(f) 
			&\overset{\text{\Cref{lemma:SumOfTravelTimesDefWithEdgeLoad}}}{=} \int_0^{\termTime(f)}\edgeLoad(\theta)d\theta \\
			&= \sum_{i \in K}\int_{a_{i-1}}^{a_i-c}\edgeLoad(\theta)d\theta + \sum_{i \in K}\int_{a_i-c}^{a_i}\edgeLoad(\theta)d\theta + \sum_{i \in [\ceil{3U}]\setminus K}\int_{a_{i-1}}^{a_i}\edgeLoad(\theta)d\theta \\
			&\leq \sum_{i \in K}U\left(\tPmax+\frac{1}{2}\right) + \sum_{i \in [\ceil{3U}]}Uc \leq \ceil{3U}U\left(2\tPmax+1+\sum_{e\in E}(\tau_e+\frac{1}{2\nu_e})\right)\qedhere
	\end{align*}

\end{proof}

\begin{remark}
	This means that for any single-sink network any IDE flow terminates within $\BigO\big(U\tau(G)\big)$ and has total travel times in $\BigO\big(U^2\tau(G)\big)$.
\end{remark}

Since $\termTime[OPT]$ is trivially bounded below by $\theta_0+1$ and $\sumOfTraveltimes[OPT]$ by $U$ this immediately leads to the following upper bound on the \MPoA{} and \TPoA{} for IDE flows:

\begin{theorem}\label{thm:PoAUppperBound}
	For any pair of integers $(U,\tau)$ we have $\MPoA(U,\tau), \TPoA(U,\tau) \in \bigO(U\tau)$. \qed
\end{theorem}

%!TEX root = ../article-PoA.tex

\section{Lower Bounds}\label{sec:LowerBounds}

It is easy to see that a general bound for the makespan cannot be better than $\BigO(U+\tau(G))$, since any feasible flow in the network consisting of one source node with an inflow rate of $U$ over the interval $[0,1]$, one sink node and a single edge between the two nodes with capacity $1$ and some travel time $\tau$ terminates by $U+\tau(G)$. Similarly, an upper bound on the total travel times certainly cannot be better than $\BigO(U(U+\tau(G)))$.
In the following, we will construct a family of instances, parameterized by $K,L \in \INs$, that provide non-trivial lower bounds on the makespan and total travel times
in single-sink networks of order $\Omega(U\cdot \log(\tau(G)))$ and $\Omega(U^2 \cdot\log(\tau(G)))$, respectively.

For any given pair of positive integers $K,L \in \INs$ the instance is of the form sketched in \Cref{fig:SlowTermGraphIntuitive}, with $u_{3^K+1}$ as source node and $t$ as its sink. The graph has a ``width'' (i.e. length of the horizontal paths from $u_1$ to $u_{3^K+1})$ of $3^{K+1}$ and a ``height'' (length of the vertical paths from nodes $u_i, v_i$ and $w_i$ to $t$) of $\approx K3^{K+1}$. All edges on the horizontal path (including the one edge back to $u_1$) have a capacity of $2U$ with $U \approx L3^{K+1}$, while all the edges on the vertical paths have capacities of either $1$ or $3$. 

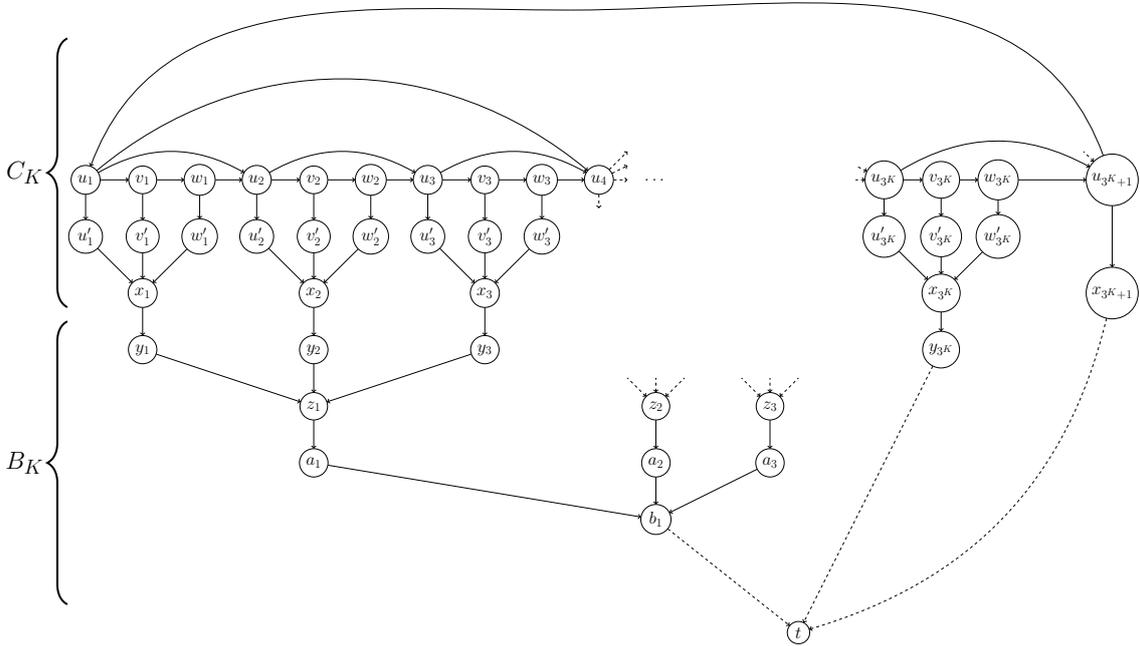
\begin{figure}[ht!]\centering
	\begin{adjustbox}{max width=\textwidth}
		\begin{tikzpicture}[font=\LARGE]%\tikzstyle{edge} = [draw,-{>[scale=1.7]},thick]
	%\useasboundingbox (-2.5,7.5) rectangle (40.5,-14.5);
	
	\node[namedVertex](u1) at (0,0) {$u_1$};
	\node[namedVertex](v1) at (2,0) {$v_1$};
	\node[namedVertex](w1) at (4,0) {$w_1$};
	\node[namedVertex](u2) at (6,0) {$u_2$};
	\node[namedVertex](v2) at (8,0) {$v_2$};
	\node[namedVertex](w2) at (10,0) {$w_2$};
	\node[namedVertex](u3) at (12,0) {$u_3$};
	\node[namedVertex](v3) at (14,0) {$v_3$};
	\node[namedVertex](w3) at (16,0) {$w_3$};
	\node[namedVertex](u4) at (18,0) {$u_4$};
	\node() at (20,0) {$\dots$};
	\node[namedVertex](u3k) at (28,0) {$u_{3^K}$};
	\node[namedVertex](v3k) at (30,0) {$v_{3^K}$};
	\node[namedVertex](w3k) at (32,0) {$w_{3^K}$};
	\node[namedVertex](u3k1) at (36,0) {$u_{3^K+1}$};

	\node[namedVertex](u1s) at (0,-2) {$u_1'$};
	\node[namedVertex](v1s) at (2,-2) {$v_1'$};
	\node[namedVertex](w1s) at (4,-2) {$w_1'$};
	\node[namedVertex](u2s) at (6,-2) {$u_2'$};
	\node[namedVertex](v2s) at (8,-2) {$v_2'$};
	\node[namedVertex](w2s) at (10,-2) {$w_2'$};
	\node[namedVertex](u3s) at (12,-2) {$u_3'$};
	\node[namedVertex](v3s) at (14,-2) {$v_3'$};
	\node[namedVertex](w3s) at (16,-2) {$w_3'$};

	\node[namedVertex](u3ks) at (28,-2) {$u_{3^K}'$};
	\node[namedVertex](v3ks) at (30,-2) {$v_{3^K}'$};
	\node[namedVertex](w3ks) at (32,-2) {$w_{3^K}'$};

	\node[namedVertex](x1) at (2,-4) {$x_1$};
	\node[namedVertex](x2) at (8,-4) {$x_2$};
	\node[namedVertex](x3) at (14,-4) {$x_3$};

	\node[namedVertex](x3k) at (30,-4) {$x_{3^K}$};
	\node[namedVertex](x3k1) at (36,-4) {$x_{3^K+1}$};

	\node[namedVertex](y1) at (2,-6) {$y_1$};
	\node[namedVertex](y2) at (8,-6) {$y_2$};
	\node[namedVertex](y3) at (14,-6) {$y_3$};
	
	\node[namedVertex](y3k) at (30,-6) {$y_{3^K}$};

	\node[namedVertex](z1) at (8,-8) {$z_1$};
	
	\node[namedVertex](z2) at (20,-8) {$z_2$};
	
	\node[namedVertex](z3) at (24,-8) {$z_3$};
	
	\node[namedVertex](a1) at (8,-10) {$a_1$};
	
	\node[namedVertex](a2) at (20,-10) {$a_2$};

	\node[namedVertex](a3) at (24,-10) {$a_3$};	
	
	\node[namedVertex](b1) at (20,-12) {$b_1$};

	%\node[namedVertex](z3k) at (30,-8) {$z_{3^{K-1}}$};

	\node[namedVertex](t) at (25,-16) {$t$};
	
	\path[edge] (u1) -- (u1s);
	\path[edge] (v1) -- (v1s);
	\path[edge] (w1) -- (w1s);
	\path[edge] (u2) -- (u2s);
	\path[edge] (v2) -- (v2s);
	\path[edge] (w2) -- (w2s);
	\path[edge] (u3) -- (u3s);
	\path[edge] (v3) -- (v3s);
	\path[edge] (w3) -- (w3s);
	\path[edge,dashed] (u4) -- +(0,-1);
	\path[edge] (u3k) -- (u3ks);
	\path[edge] (v3k) -- (v3ks);
	\path[edge] (w3k) -- (w3ks);
	\path[edge] (u1) -- (v1);
	\path[edge] (v1) -- (w1);
	\path[edge] (w1) -- (u2);
	\path[edge] (u2) -- (v2);
	\path[edge] (v2) -- (w2);
	\path[edge] (w2) -- (u3);
	\path[edge] (u3) -- (v3);
	\path[edge] (v3) -- (w3);
	\path[edge] (w3) -- (u4);
	\path[edge] (u3k) -- (v3k);
	\path[edge] (v3k) -- (w3k);
	\path[edge] (w3k) -- (u3k1);

	\path[edge] (u1s) -- (x1);
	\path[edge] (v1s) -- (x1);
	\path[edge] (w1s) -- (x1);
	\path[edge] (u2s) -- (x2);
	\path[edge] (v2s) -- (x2);
	\path[edge] (w2s) -- (x2);
	\path[edge] (u3s) -- (x3);
	\path[edge] (v3s) -- (x3);
	\path[edge] (w3s) -- (x3);
	\path[edge] (u3ks) -- (x3k);
	\path[edge] (v3ks) -- (x3k);
	\path[edge] (w3ks) -- (x3k);

	\path[edge] (u3k1) -- (x3k1);
	
	\path[edge] (x1) -- (y1);
	\path[edge] (x2) -- (y2);
	\path[edge] (x3) -- (y3);	
	
	\path[edge] (x3k) -- (y3k);
	
	\path[edge] (y1) -- (z1);
	\path[edge] (y2) -- (z1);		
	\path[edge] (y3) -- (z1);
	
	\path[edge,<-,dashed] (z2) -- +(-1,1);
	\path[edge,<-,dashed] (z2) -- +(0,1);
	\path[edge,<-,dashed] (z2) -- +(1,1);	

	\path[edge,<-,dashed] (z3) -- +(-1,1);
	\path[edge,<-,dashed] (z3) -- +(0,1);
	\path[edge,<-,dashed] (z3) -- +(1,1);

	%\path[edge] (y3k) -- (z3k);	

%	\path[edge] (u3k1) to [bend right=135] (u1);	

	\path[edge,dashed] (u4) -- +(1,0);
	\path[edge,dashed] (u4) -- +(1,.5);
	\path[edge,dashed] (u4) -- +(1,1);
	\path[edge,dashed,<-] (u3k) -- +(-1,0);
	\path[edge,dashed,<-] (u3k) -- +(-1,.5);
	\path[edge,dashed,<-] (u3k1) -- +(-1,1);

	\node[](temp) at (17,6) {};
	\path[edge,-] (u3k1) to [out=110,in=0] (temp.center);
	\path[edge] (temp.center) to [out=180,in=70] (u1);
	
	\path[edge] (u1) to [bend left=30] (u2);
	\path[edge] (u2) to [bend left=30] (u3);
	\path[edge] (u3) to [bend left=30] (u4);
	\path[edge] (u3k) to [bend left=30] (u3k1);
	
	\path[edge] (u1) to [bend left=40] (u4);
	
	\path[edge] (z1) -- (a1);
	\path[edge] (z2) -- (a2);
	\path[edge] (z3) -- (a3);

	\path[edge] (a1) -- (b1);
	\path[edge] (a2) -- (b1);
	\path[edge] (a3) -- (b1);
	
	\path[edge,dashed] (b1) -- (t);
%	\path[edge,dashed] (x2) --node[left]{$P_2$} (t);
%	\path[edge,dashed] (x3) --node[left]{$P_3$} (t);
	\path[edge,dashed] (y3k) -- (t);
	\path[edge,dashed] (x3k1) to[bend left=30] (t);
	
	\draw [decorate,decoration={brace,amplitude=20pt},xshift=-4pt,yshift=0pt,line width=2pt]
		(-.5,-4.5) -- (-.5,5) node [midway,xshift=-1.5cm] 
		{\Huge $C_K$};

	\draw [decorate,decoration={brace,amplitude=20pt},xshift=-4pt,yshift=0pt,line width=2pt]
		(-.5,-15) -- (-.5,-5) node [midway,xshift=-1.5cm] 
		{\Huge $B_K$};

\end{tikzpicture}
	\end{adjustbox}
	\caption{A network with a total edge travel time of $\approx 3^{2K}$ that, given an inflow volume of $\approx L3^{K+1}$, has a makespan of more than $KL3^{K+1}$.}\label{fig:SlowTermGraphIntuitive}
\end{figure} 

If we let flow enter at node $u_{3^K+1}$ at a rate of $2U$ over the interval $[-0.5,0]$, we will observe the following behavior: 
At first, all flow enters the direct downwards path towards the sink $t$ until a queue of length $1$  has built up on the first edge of this path. After that almost all flow will enter the edge towards $u_1$ and some flow will go downwards to keep the queue length constant. Assuming $U \gg 1$, most of the flow will travel towards $u_1$ and arrive there one time step later with a slightly lower inflow rate and over a slightly shorter interval as at $u_{3^K+1}$.
At $u_1$ the same flow split happens again: First all flow enters the edge to $u'_1$ until a queue of sufficient length to induce a waiting time of $1$ has formed, then most of flow travels towards the next node $v_1$. Similarly, at all the following nodes on the horizontal path, this pattern repeats, i.e., a small amount of flow (of volume $\approx 3$) starts traveling downwards while most of the flow is diverted further to the right. Finally the main block of flow arrives at $u_{3^K+1}$ and is diverted back to $u_1$ (having lost a total volume of $\approx 3^{K+1}$ to the downwards paths). By the time this flow arrives at $u_1$ again, the flow particles that traveled along the edges $u_1u'_1$, $v_1v'_1$ and $w_1w'_1$ join up again at node $x_1$ in such a way that they form a queue of length $\approx 3$ on the edge $y_1z_1$.
This queue is long enough to divert the main block of flow away towards $u_2$ (over the direct edge $u_1u_2$ of length $3$). This pattern, again, repeats at all subsequent nodes $u_i$ until the main flow finally arrives at node $u_{3^K+1}$ (having lost no additional flow volume) and is again diverted back to $u_1$. This time, the flow particles from the queues on the edges $y_1z_1, y_2z_1$ and $y_3z_1$ met at node $z_1$ and now form a queue of length $\approx 9$ on edge $a_1b_1$. Thus, our main flow can now be diverted away directly to node $u_4$ and so on. 
This way, the main amount of flow can travel along the horizontal path for $\approx K$ times without losing a significant amount of flow until all the flow on the vertical paths finally reaches the sink $t$. After that the pattern described until now repeats. Thus, flow remains in the network until at least time $\approx 3^{K+1} K \frac{U}{3^{K+1}} = 3^{K+1} K L$. From this we obtain the following theorem.

\begin{theorem}\label{ex:SlowTerminationExample}\label{thm:SlowTermination}
	Given any pair $K,L \in \INs$, there exists an instance $G_{K,L}$ with $\tau(G_{K,L}) \in \BigO(3^{2K})$ and $U_{K,L} \in \BigO(L3^{K})$ such that there exists an IDE flow that does not terminate before $L(K+1)(3^{K+1}+1)$ and has total travel times of at least $ K3^{2K+3}\frac{1}{2L}\big((L-1)^3-(L-1)\big)$. 
\end{theorem}

	\newcommand{\inU}[2]{\overline{\mathrm{in}}_{#1}^{#2}}
	\newcommand{\inL}[2]{\underline{\mathrm{in}}_{#1}^{#2}}
	\newcommand{\outU}[2]{\overline{\mathrm{out}}_{#1}^{#2}}
	\newcommand{\outL}[2]{\underline{\mathrm{out}}_{#1}^{#2}}
	\newcommand{\queueU}[2]{\overline{q}_{#1}^{#2}}
	\newcommand{\queueL}[2]{\underline{q}_{#1}^{#2}}
	
	Before we can construct the instances for the formal proof of \Cref{thm:SlowTermination} we first need two observations. The first one describes the flow split at the nodes $u_i$, $v_i$ and $w_i$ on the horizontal path:
	
	\begin{obs}\label{obs:SlowTermFlowSplit}
		Given a graph as in \Cref{fig:SlowTermSplit} with $\tau(P_v) = \tau(P_u)+1$, a flow $f$ and some time $\theta_0$ such that there are no queues on the paths $P_u$, $P_v$ between $\theta_0-x$ and $\theta_0$, no flow on the edges $uv$ and $uu'$ at time $\theta_0$, and a constant edge outflow rate from edge $e$ of $f^-_e(\theta) = y$ over the interval $[\theta_0-x, \theta_0]$, where $y > \nu_{uu'}$ and $0.5 \geq x \geq \varepsilon \coloneqq \frac{\nu_{uu'}}{y-\nu_{uu'}}$. Then the flow split described in \Cref{fig:SlowTermSplit} is an IDE. Note, that even if the outflow of edge $e$ is not exactly the one given in \Cref{fig:SlowTermSplit}, but still bounded below by it and has its support in $[\theta_0-0.5,\theta_0]$, then the inflows into $uv$ and $P_u$ are also bounded below by the inflow functions given in \Cref{fig:SlowTermSplit} and their supports are in $[\theta_0-0.5,\theta_0]$ and $[\theta_0+1-0.5,\theta_0+2]$, respectively.
		
		\begin{minipage}{.5\textwidth}\vspace{1em}
			\begin{adjustbox}{max width=\textwidth}
				\begin{tikzpicture}
	\useasboundingbox (-5.5,4) rectangle (4.5,-8);
	
	\node[namedVertex] (u) at (0,0) {$u$};
	\node[namedVertex] (v) at (3,0) {$v$};
	\node[namedVertex] (us) at (0,-3) {$u'$};
	\node[namedVertex] (t) at (1.5,-7) {$t$};
	
	\path[edge,<-] (u) --node[above](e){$e$} +(-2,0);
	\path[edge] (u) --node[above](uv){$\nu_{uv}\geq y$}node[below]{$\tau_{uv}=1$} (v);
	\path[edge] (u) --node[above,sloped](uus){$\nu_{uu'}=1$}node[below,sloped]{$\tau_{uu'}=1$} (us);
	
	\path[edge,dashed] (us) -- node[left](Pu){$P_u$} (t);
	\path[edge,dashed] (v) -- node[right]{$P_v$} (t);
	
	\node() at ($(e)+(-2,2)$) {\begin{tikzpicture}[anchor=center,scale=.8,solid,black,
		declare function={
			f(\x)= and(\x>=1, \x<2) * 8;
		}
		]
		\begin{axis}[xmin=.5,xmax=2.5,ymax=9, ymin=-.5, samples=500,width=5cm,height=5cm,
					xtick={1,2},xticklabels={$\theta_0-x$,$\theta_0$},
					ytick={0,8},yticklabels={$0$,$y$}]
			\addplot[blue, ultra thick,domain=0:3] {f(x)};
		\end{axis}
		\node[blue]() at (2,2) {$f^-_e$};
	\end{tikzpicture}};
	\node() at ($(uus)+(-3.5,-.5)$) {\begin{tikzpicture}[anchor=center,scale=.8,solid,black,
		declare function={
			f(\x)= and(\x>=1, \x<1.2) * 8 + and(\x>=1.2,\x<2)*2;
		}
		]
		\begin{axis}[xmin=.5,xmax=2.5,ymax=8.5, ymin=-.5, samples=500,width=5cm,height=5cm,
					xtick={1.2,2},xticklabels={$\theta_0-x+\varepsilon$,$\theta_0$},
					ytick={0,2,8},yticklabels={$0$,$\nu_{uu'}$,$y$}]
			\addplot[blue, ultra thick,domain=0:3] {f(x)};
		\end{axis}
		\node[blue]() at (2,2) {$f^+_{uu'}$};
	\end{tikzpicture}};
	\node() at ($(uv)+(-.5,2)$) {\begin{tikzpicture}[anchor=center,scale=.8,solid,black,
		declare function={
			f(\x)= and(\x>=1.2,\x<2)*6;
		}
		]
		\begin{axis}[xmin=.5,xmax=2.5,ymax=8.5, ymin=-.5, samples=500,width=5cm,height=5cm,
					xtick={1.2,2},xticklabels={$\theta_0-x+\varepsilon$,$\theta_0$},
					ytick={0,6},yticklabels={$0$,$y-\nu_{uu'}$}]
			\addplot[blue, ultra thick,domain=0:3] {f(x)};
		\end{axis}
		\node[blue]() at (2,2) {$f^+_{uv}$};
	\end{tikzpicture}};
	\node() at ($(Pu)+(-3.1,-.5)$) {\begin{tikzpicture}[anchor=center,scale=.8,solid,black,
		declare function={
			f(\x)= and(\x>=1,\x<4)*2;
		}
		]
		\begin{axis}[xmin=.5,xmax=4.5,ymax=4.5, ymin=-.5, samples=500,width=7cm,height=3.5cm,
					xtick={1,4},xticklabels={$\theta_0-x+1$,$\theta_0+2$},
					ytick={0,2},yticklabels={$0$,$\nu_{uu'}$}]
			\addplot[blue, ultra thick,domain=0:4.5] {f(x)};
		\end{axis}
		\node[blue]() at (2,1.4) {$f^+_{P_u}$};
	\end{tikzpicture}};
\end{tikzpicture}
			\end{adjustbox}
		\end{minipage}%
		\begin{minipage}{.47\textwidth}
			
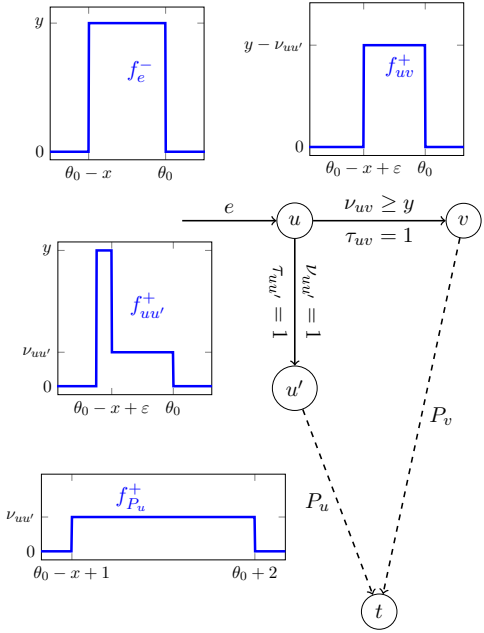
\captionof{figure}{Given the outflow rate of edge $e$ and no flow on any other edge at time $\theta_0-x$, the following flow split is an IDE: At first all flow arriving at $u$ enters the edge towards $u'$, starting to built up a queue on this edge (since $y > \nu_{uu'}$). By time $\theta_0-x+\varepsilon$ the queue reaches a length of $\nu_{uu'}$ and, thus, the path $uv,P_v$ has now the same instantaneous travel time as the path $uu',P_u$. Therefor, from now on flow enters the edge towards $u'$ at rate $\nu_{uu'}$ (to keep the queue at the current length) and the rest of the flow travels towards $v$.}\label{fig:SlowTermSplit}
		\end{minipage}
	\end{obs}
	
	The second observation describes the evolution of the queues on the vertical paths (except for the one on the top most edges, e.g. $u_iu'_i$):
	
	\begin{obs}\label{obs:StepFlowInto1Edge}
		Given an edge $e$ with $\nu_e=\tau_e=1$, a time $\theta_0$ with $\edgeLoad[e](\theta_0)=0$ and some $k \in \INo$. If between $\theta_0$ and $\theta_0+14\cdot 3^k$ the edge inflow rate $f_e^+$ into $e$ is bounded by the two functions $\inL{\theta_0}{k}$ and $\inU{\theta_0}{k}$, i.e. the blue continuous and the red dashed graph in the topmost diagram of \Cref{fig:SlowTermEdgeWithCap1}, then the queue length on this edge is bounded by $\queueL{\theta_0}{k}$ and $\queueU{\theta_0}{k}$ over this time interval (blue continuous and red dashed line in the bottom diagram of \Cref{fig:SlowTermEdgeWithCap1}) and the edge outflow rate $f_e^-$ by $\outL{\theta_0}{k}$ and $\outU{\theta_0}{k}$ (blue continuous and red dashed line in the middle diagram of \Cref{fig:SlowTermEdgeWithCap1}) over the interval $[\theta_0,\theta_0+14\cdot3^k+1]$.
		
		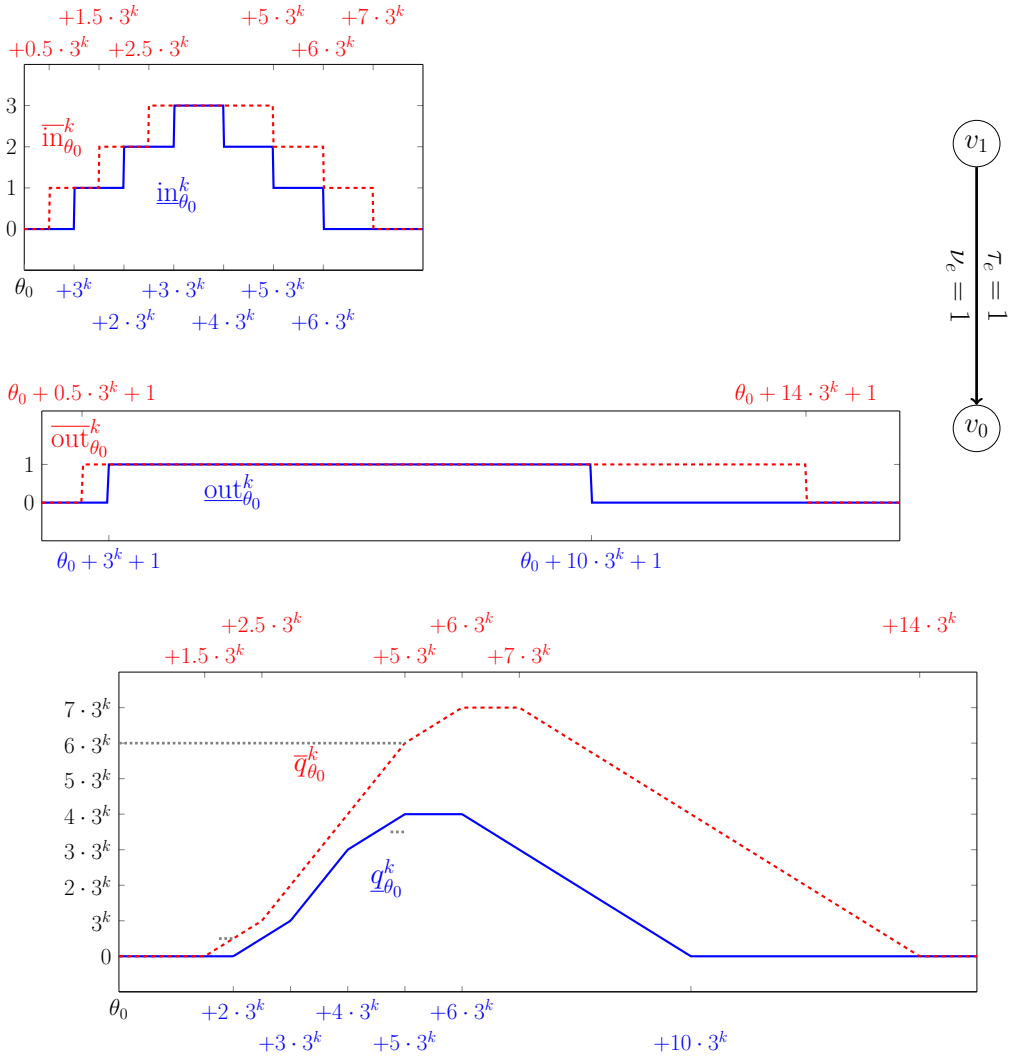
\begin{figure}[ht!]\centering
			\begin{adjustbox}{max width=.9\textwidth}
				\begin{tikzpicture}[font=\LARGE]
	\node[namedVertex](1) at (0,0) {$v_1$};
	\node[namedVertex](0) at (0,-6) {$v_0$};
	
	\path[edge,ultra thick] (1) --node[below,sloped](e){$\nu_e=1$}node[above,sloped](e){$\tau_e=1$} (0);
	
	\node() at ($(1)+(-16,-.5)$) {\begin{tikzpicture}[anchor=center,scale=.8,solid,black,
		declare function={
			f(\x)= and(\x>=1, \x<2) * 1 + and(\x>=2, \x<=3) * 2 + and(\x>=3, \x<4) * 3  + and(\x>=4, \x<5) * 2  + and(\x>=5, \x<6) * 1;
			g(\x)= and(\x>=.5, \x<1.5) * 1 + and(\x>=1.5, \x<=2.5) * 2 + and(\x>=2.5, \x<5) * 3  + and(\x>=5, \x<6) * 2  + and(\x>=6, \x<7) * 1;
		}
		]
		\begin{axis}[xmin=0,xmax=8,ymax=4, ymin=-1, samples=500,width=12cm,height=7cm,
					xtick={1,2,3,4,5,6},xticklabels={$+3^k$,$+2\cdot3^k$,$+3\cdot3^k$,$+4\cdot3^k$,$+5\cdot 3^k$,$+6\cdot3^k$},
					x tick label style={yshift={mod(\ticknum,2)*(-1.5em)}},
					ytick={0,1,2,3},yticklabels={$0$,$1$,$2$,$3$},
					x tick label style={color=blue},
					axis x line*=bottom]
			\addplot[blue, ultra thick,domain=0:8] {f(x)};
		\end{axis}
		\begin{axis}[xmin=0,xmax=8,ymax=4, ymin=-1, samples=500,width=12cm,height=7cm,
				xtick={0.5,1.5,2.5,5,6,7},xticklabels={$+0.5\cdot3^k$,$+1.5\cdot3^k$,$+2.5\cdot3^k$,$+5\cdot 3^k$,$+6\cdot3^k$,$+7\cdot3^k$},
				x tick label style={yshift={mod(\ticknum,2)*(1.5em)}},
				axis x line*=top,
				hide y axis,
				x tick label style={color=red},
				xlabel near ticks]
			\addplot[red, ultra thick,domain=0:8,dashed] {g(x)};
		\end{axis}
		\begin{axis}[xmin=0,xmax=8,ymax=4, ymin=-1, samples=500,width=12cm,height=7cm,
				xtick={0},xticklabels={$\theta_0$},
				axis x line*=bottom,
				hide y axis,
				x tick label style={color=black}]
		\end{axis}
		
		\node[red]() at (1,3.5) {$\inU{\theta_0}{k}$};
		\node[blue]() at (4,2) {$\inL{\theta_0}{k}$};
	\end{tikzpicture}};

	\node() at ($(0)+(-11,-1)$) {\begin{tikzpicture}[anchor=center,scale=.8,solid,black,
		declare function={
			f(\x)= and(\x>=1.25, \x<=10.25) * 1;
			g(\x)= and(\x>=0.75, \x<=14.25) * 1;
		}
		]
		\begin{axis}[xmin=0,xmax=16,ymax=2.4, ymin=-1, samples=500,width=24cm,height=5cm,
					xtick={1.25,10.25},xticklabels={$\theta_0+3^k+1$,$\theta_0+10\cdot 3^k+1$},
					ytick={0,1},yticklabels={$0$,$1$},
					x tick label style={color=blue},
					axis x line*=bottom]
			\addplot[blue, ultra thick,domain=0:16] {f(x)};
		\end{axis}
		\begin{axis}[xmin=0,xmax=16,ymax=2.4, ymin=-1, samples=500,width=24cm,height=5cm,
				xtick={0.75,14.25},xticklabels={$\theta_0+0.5\cdot3^k+1$,$\theta_0+14\cdot 3^k+1$},
				axis x line*=top,
				hide y axis,
				x tick label style={color=red},
				xlabel near ticks]
			\addplot[red, ultra thick,domain=0:16,dashed] {g(x)};
		\end{axis}
		\node[blue]() at (5,1.3) {$\outL{\theta_0}{k}$};
		\node[red]() at (1,2.7) {$\outU{\theta_0}{k}$};
	\end{tikzpicture}};

	\node() at ($(e)+(-10,-11.5)$) {\begin{tikzpicture}[anchor=center,scale=.8,solid,black,
		declare function={
			f(\x)= and(\x>=2, \x<3) * (\x-2) + and(\x>=3, \x<4) * (2*(\x-3)+1) + and(\x>=4, \x<5) * (\x-4+3) + and(\x>=5, \x<6) * 4  + and(\x>=6, \x<10) * (-(\x-6)+4);
			g(\x)= and(\x>=1.5, \x<2.5) * (\x-1.5) + and(\x>=2.5, \x<5) * (2*(\x-2.5)+1) + and(\x>=5, \x<6) * (\x-5+6) + and(\x>=6, \x<7) * 7  + and(\x>=7, \x<14) * (-(\x-7)+7);
			h(\x)=and(\x>=1.5,\x<=2) * 0.5 + and(\x>=4.5,\x<=5) * 3.5;
			hh(\x) =6;
		}
		]
		\begin{axis}[xmin=0,xmax=15,ymax=8, ymin=-1, samples=500,width=24cm,height=10cm,
				xtick={2,3,4,5,6,10},xticklabels={$+2\cdot3^k$,$+3\cdot3^k$,$+4\cdot3^k$,$+5\cdot 3^k$,$+6\cdot3^k$,$+10\cdot3^k$},
				x tick label style={yshift={mod(\ticknum,2)*(-1.5em)}},
				ytick={0,1,2,3,4,5,6,7},yticklabels={$0$,$3^k$,$2\cdot3^k$,$3\cdot3^k$,$4\cdot3^k$,$5\cdot3^k$,$6\cdot3^k$,$7\cdot3^k$},
				x tick label style={color=blue},
				axis x line*=bottom]
			\addplot[blue, ultra thick,domain=0:16] {f(x)};
		\end{axis}
		\begin{axis}[xmin=0,xmax=15,ymax=8, ymin=-1, samples=500,width=24cm,height=10cm,
				xtick={1.5,2.5,5,6,7,14},xticklabels={$+1.5\cdot3^k$,$+2.5\cdot3^k$,$+5\cdot 3^k$,$+6\cdot3^k$,$+7\cdot3^k$,$+14\cdot3^k$},
				x tick label style={yshift={mod(\ticknum,2)*(1.5em)}},
				axis x line*=top,
				hide y axis,
				x tick label style={color=red},
				xlabel near ticks]
			\addplot[red, ultra thick,domain=0:16,dashed] {g(x)};
		\end{axis}		
		\begin{axis}[xmin=0,xmax=15,ymax=8, ymin=-1, samples=500,width=24cm,height=10cm,
				xtick={0},xticklabels={$\theta_0$},
				axis x line*=bottom,
				hide y axis,
				x tick label style={color=black}]
			\addplot[gray, line width=2pt,domain=1.75:2,dotted] {h(x)};
			\addplot[gray, line width=2pt,domain=4.75:5,dotted] {h(x)};
			\addplot[gray, line width=2pt,domain=0:5,dotted] {hh(x)};
		\end{axis}

		\node[blue]() at (7,3) {$\queueL{\theta_0}{k}$};
		\node[red]() at (5,6) {$\queueU{\theta_0}{k}$};
	\end{tikzpicture}};

\end{tikzpicture}
			\end{adjustbox}
			\caption{Given an edge with capacity $1$, travel time $1$ and an inflow rate $f^+_e$ between $\inL{\theta_0}{k}$ and $\inU{\theta_0}{k}$, the queue length on this edge will be between $\queueL{\theta_0}{k}$ and $\queueU{\theta_0}{k}$ and the outflow $f_e^-$ between $\outL{\theta_0}{k}$ and $\outU{\theta_0}{k}$.}\label{fig:SlowTermEdgeWithCap1}
		\end{figure}		
		Note, that in particular the waiting time on $e$ is less than $0.5\cdot3^k$ on the interval $[\theta_0+2\cdot3^k-0.5,\theta_0+2\cdot3^k]$, at least $3.5\cdot3^k$ on the interval $[\theta_0+5\cdot3^k-0.5,\theta_0+5\cdot3^k]$ and never more than $6\cdot3^k$ on the interval $[\theta_0,\theta_0+5\cdot3^k]$ (see gray dotted lines in \Cref{fig:SlowTermEdgeWithCap1}). Additionally, the rate of change of the waiting time is bounded by $2$ for the whole interval $[\theta_0,\theta_0+14\cdot 3^k]$.
	\end{obs}
	
	We can now proceed to construct the instances $G_{K,L}$ for the proof of \Cref{ex:SlowTerminationExample}:
	
	\begin{proof}[Proof of \Cref{ex:SlowTerminationExample}]
		Our instance consists of two parts: A cycling gadget $C_K$ containing the horizontal paths and a blocking gadget $B_K$ containing the vertical paths. We first construct $B_K$ inductively: Namely, for any $k \leq K$ we define a gadget $B_{K,k}$ (see \Cref{fig:SlowTermGadgetBK}) with the following properties:
		
		\begin{itemize}
			\item There are $3^k$ edges $e_1, \dots, e_{3^k}$ leading into the gadget and one edge $e_0$ coming out of it. For each of the edges $e_j$, $j \geq 1$ there exists one unique path $P_j^k$ from its head to the tail of edge $e_0$. Those paths are not necessarily edge disjoint, but all have the same physical travel time of $\tau(P_j^k) = (k-1)(3^{K+1}+1) + 1 - 3^{K-k} + 3^{K-1}$.
			\item The sum of the travel times of all edges inside $B_{K,k}$ is given by $\tau(B_{K,k}) = \sum_{k'=1}^k \big(4\cdot 3^{k'-1} + (3^{k'+1}+2)3^{K-1}\big)-11\cdot 3^{K-1}-1$.
			\item If for some time $\theta_0$ there is no flow inside the gadget and after $\theta_0$ for each edge $e_j$ the outflow $f^-_{e_j}$ of this edge lies between $\inL{\theta_0+(j-1)3^{K-k+1}}{K-k}$ and $\inU{\theta_0+(j-1)3^{K-k+1}}{K-k}$, then all flow has left path $P_j^k$ by $(k-1)(3^K+1)+(j'-1)3^K+14\cdot3^{K-1}+1$, where $j' \in \set{1,2,3}$ is chosen such that $(j'-1)3^{k-1} < j \leq j'3^{k-1}$. Furthermore, before that the paths $P_j^k$ exhibit the \emph{$(K,k,j,\theta_0)$-blocking property}, which is:
			
			For any $k' \in \set{0, \dots, k-1}$, $j' \in \set{0, \dots, 3^{k-k'}-1}$ with $j \in \set{j'3^{k'}\!\!+1, \dots, (j'+1)3^{k'}}$ and $\hat{\theta} \coloneqq \theta_0+k'\cdot(3^{K+1}+1)+j'3^{K-k+k'+1}$ the waiting times on $P_j^k$ are
			\begin{itemize}
				\item at most $0.5\cdot 3^{K-k+k'}$ on the interval $[\hat{\theta}+2\cdot 3^{K-k}-0.5,\hat{\theta}+2\cdot 3^{K-k}]$, 
				\item at least $3.5\cdot 3^{K-k+k'}$ on the interval $\hat{\theta}+5\cdot3^{K-k}-0.5,\hat{\theta}+5\cdot3^{K-k}]$
				\item at most $6\cdot 3^{K-k+k'}$ on the interval $[\hat{\theta},\hat{\theta}+5\cdot3^{K-k}]$ and the rate at which the queue-length changes in never more than $2$.
			\end{itemize}
		\end{itemize}
		
		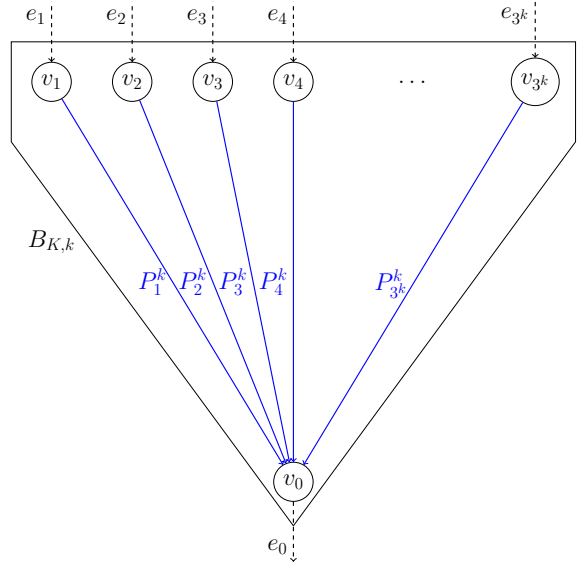
\begin{figure}[ht!]\centering
			\begin{adjustbox}{max width=.5\textwidth}
				\begin{tikzpicture}[font=\LARGE]
	\node () at (1,-5) {\LARGE $B_{K,k}$};
	
	\draw (0,-2.5) -- (0,0) -- (14,0) -- (14,-2.5) -- (7,-12.1) --cycle;
	
	\node[namedVertex] (1) at (1,-1) {$v_1$};
	\node[namedVertex] (2) at (3,-1) {$v_2$};
	\node[namedVertex] (3) at (5,-1) {$v_3$};
	\node[namedVertex] (4) at (7,-1) {$v_4$};
	\node () at (10,-1) {$\dots$};
	\node[namedVertex] (3k) at (13,-1) {$v_{3^k}$};
	
	\path[edge,<-, dashed] (1) -- node[left, near end]{$e_1$} +(0,2);
	\path[edge,<-, dashed] (2) -- node[left, near end]{$e_2$} +(0,2);
	\path[edge,<-, dashed] (3) -- node[left, near end]{$e_3$} +(0,2);
	\path[edge,<-, dashed] (4) -- node[left, near end]{$e_4$} +(0,2);
	\path[edge,<-, dashed] (3k) -- node[left, near end]{$e_{3^k}$} +(0,2);
	
	\node[namedVertex] (0) at (7,-11) {$v_0$};
	\path[edge, dashed] (0) -- node[left,near end]{$e_0$} +(0,-2);
	
	\path[edge,blue] (1) --node[left]{$P_1^k$} (0);
	\path[edge,blue] (2) --node[left]{$P_2^k$} (0);
	\path[edge,blue] (3) --node[left]{$P_3^k$} (0);
	\path[edge,blue] (4) --node[left]{$P_4^k$} (0);
	\path[edge,blue] (3k) --node[left]{$P_{3^k}^k$} (0);

\end{tikzpicture}
			\end{adjustbox}
			\caption{The Blocking-Gadget $B_{K,k}$}\label{fig:SlowTermGadgetBK}
		\end{figure}	
		
		For $k=1$ such a gadget just consists of three edges $v_1v_0$, $v_2v_0$ and $v_3v_0$ each with capacity $1$ and travel time $1$ (see \Cref{fig:SlowTermGadgetBK1}). 
		
		Thus, $\tau(B_{K,1}) = 3$. With an outflow between $\inL{\theta_0+(j-1)3^{K}}{K-1}$ and $\inU{\theta_0+(j-1)3^{K}}{K-1}$ on each of the edges $e_j, j=1,2,3$ \Cref{obs:StepFlowInto1Edge} immediately shows, that the edge $v_jv_0$ (which is the path $P^1_j$ for this gadget) has the $(K,1,j,\theta_0)$-blocking property. Note, that it also follows from \Cref{obs:StepFlowInto1Edge} that the inflow into edge $e_0$ lies between 
		\[\outL{\theta_0}{K-1}+\outL{\theta_0+3^{K}}{K-1}+\outL{\theta_0+2\cdot3^{K}}{K-1} = \inL{\theta_0+1+3^{K-1}-3^K}{K}\]
		and 
		\[\outU{\theta_0}{K-1}+\outU{\theta_0+3^{K}}{K-1}+\outU{\theta_0+2\cdot3^{K}}{K-1} = \inU{\theta_0+1+3^{K-1}-3^K}{K}.\]
		In particular, all flow has left the path $P^1_j$ by time $\theta_0+(j-1)3^K+14\cdot 3^{K-1}+1$. The specific form of the inflow into edge $e_0$ is not important for $B_{K,1}$, but will be used for the induction step of the construction.
		
		\begin{minipage}{.3\textwidth}
			\vspace{.5em}
			\begin{adjustbox}{max width=\textwidth}
				\begin{tikzpicture}
	\node () at (1,-3) {\LARGE $B_{K,1}$};
	
	\draw (0,-1) -- (0,0) -- (6,0) -- (6,-1) -- (3,-4) --cycle;
	
	\node[namedVertex] (1) at (1,-1) {$v_1$};
	\node[namedVertex] (2) at (3,-1) {$v_2$};
	\node[namedVertex] (3) at (5,-1) {$v_3$};
	
	\path[edge,<-, dashed] (1) -- node[left, near end]{$e_1$} +(0,2);
	\path[edge,<-, dashed] (2) -- node[left, near end]{$e_2$} +(0,2);
	\path[edge,<-, dashed] (3) -- node[left, near end]{$e_3$} +(0,2);
	
	\node[namedVertex] (v0) at (3,-3) {$v_0$};
	
	\path[edge, blue] (1) --node[below,sloped]{$\nu_e=1$} (v0);
	\path[edge, blue] (2) --node[below,sloped]{$\nu_e=1$} (v0);
	\path[edge, blue] (3) --node[below,sloped]{$\nu_e=1$} (v0);
	
	\path[edge, dashed] (v0) -- node[left,near end]{$e_0$} +(0,-2);
\end{tikzpicture}
			\end{adjustbox}
			\vspace{.5em}
		\end{minipage}\hspace{.03\textwidth}%
		\begin{minipage}{.64\textwidth}
			
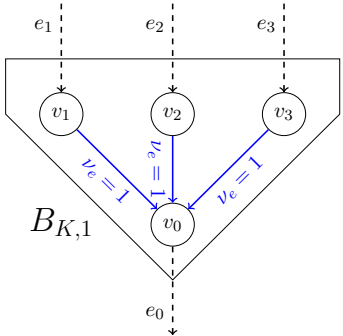
\captionof{figure}{The blocking gadget $B_{K,1}$ (the first building block for the inductive construction of the general Blocking-Gadget $B_{K,k}$.)}\label{fig:SlowTermGadgetBK1}
		\end{minipage}	
		
		\begin{minipage}{.55\textwidth}
			\vspace{.5em}
			\begin{adjustbox}{max width=\textwidth}
				\begin{tikzpicture}[font=\large]
	\node () at (1,-6) {\LARGE $B_{K,k}$};
	\draw (0,-3.5) -- (0,0) -- (14,0) -- (14,-3.5) -- (7,-12.5) --cycle;
	
	\draw[dashed] (.25,-1.25) -- (.25,-.25) -- (5.75,-.25) -- (5.75,-1.25) -- (3,-4.25) --cycle;
	\node () at (.6,-2.6) {\LARGE $B_{K,1}$};	
	
	\node[namedVertex] (1) at (1,-1) {$v_1$};
	\node[namedVertex] (2) at (3,-1) {$v_2$};
	\node[namedVertex] (3) at (5,-1) {$v_3$};
	\node[namedVertex] (4) at (7,-1) {$v_4$};
	\node () at (10,-1) {$\dots$};
	\node[namedVertex] (3k) at (13,-1) {$v_{3^k}$};
	
	\path[edge,<-, dashed] (1) -- node[left, near end]{$e_1$} +(0,2);
	\path[edge,<-, dashed] (2) -- node[left, near end]{$e_2$} +(0,2);
	\path[edge,<-, dashed] (3) -- node[left, near end]{$e_3$} +(0,2);
	\path[edge,<-, dashed] (4) -- node[left, near end]{$e_4$} +(0,2);
	\path[edge,<-, dashed] (3k) -- node[left, near end]{$e_{3^k}$} +(0,2);
	
	\node[namedVertex,blue] (w1) at (3,-3) {$w_1$};
	\node[namedVertex,blue] (w2) at (9,-3) {$w_2$};
	\node[namedVertex,blue] (w3k) at (11,-3) {$w_{3^{k-1}}$};
	
	\path[edge, blue] (1) --node[below,sloped]{$\nu_e=1$} (w1);
	\path[edge, blue] (2) --node[below,sloped]{$\nu_e=1$} (w1);
	\path[edge, blue] (3) --node[below,sloped]{$\nu_e=1$} (w1);
	\path[edge, blue] (4) --node[below,sloped]{$\nu_e=1$} (w2);
	\path[edge, blue] (3k) --node[below,sloped]{$\nu_e=1$} (w3k);	
		
	\draw (2.5,-5.5) -- (2.5,-4.5) -- (11.5,-4.5) -- (11.5,-5.5) -- (7,-11) --cycle;
	\node () at (7,-5.5) {\LARGE $B_{K,k-1}$};
	
	\node[vertex] (1b) at (3,-5) {};
	\node[vertex] (2b) at (9,-5) {};
	\node[vertex] (3kb) at (11,-5) {};
	
	\path[edge, blue] (w1) --node[left]{$e'_1$} (1b);
	\path[edge, blue] (w2) --node[left]{$e'_2$} (2b);
	\path[edge, blue] (w3k) --node[left]{$e'_{3^{k-1}}$} (3kb);
	
	\node[vertex] (0b) at (7,-10) {};

	\path[edge,blue,dashed] (1b) --node[right,near start]{$P_1^{k-1}$} (0b);
	\path[edge,blue,dashed] (2b) --node[left,near start]{$P_2^{k-1}$} (0b);
	\path[edge,blue,dashed] (3kb) --node[left,near start]{$P_{3^{k-1}}^{k-1}$} (0b);	
	
	\node[namedVertex] (0) at (7,-11.5) {$v_0$};
	\path[edge, blue] (0b) -- (0);
	\path[edge, dashed] (0) -- node[left,near end]{$e_0$} +(0,-2);
	
\end{tikzpicture}
			\end{adjustbox}
			\vspace{.5em}
		\end{minipage}\hspace{.03\textwidth}%
		\begin{minipage}{.39\textwidth}
			
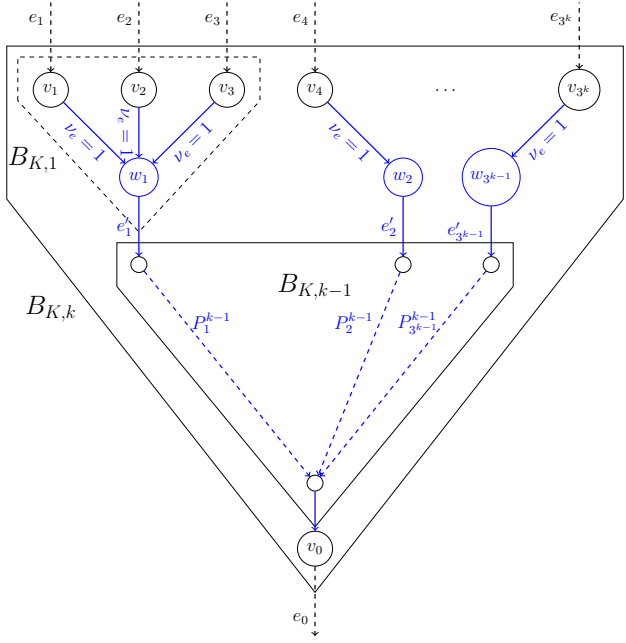
\captionof{figure}{The inductive construction the blocking gadget $B_{K,k}$ from $3^{k-1}$ copies of the gadget $B_{K,1}$ and one of gadget $B_{K,k-1}$. The connecting edges $e'_{j'}$ have capacity $3$ and travel time $3^{K+1}+3^{K-k+1}-3^{K-k}$.}\label{fig:SlowTermGadgetBKkInductionStep}
		\end{minipage}	
		
		Analogous to the proof for gadget $B_{K,1}$ it follows from \Cref{obs:StepFlowInto1Edge} that, given an outflow between $\inL{\theta_0+(j-1)3^{K-k+1}}{K-k}$ and $\inU{\theta_0+(j-1)3^{K-k+1}}{K-k}$ from edges $e_j$ the blocking property holds for $k'=0$ and all $j', j$ and the inflow into the edges $e'_{j'}$ ($j' \in \set{1,\dots,3^{k-1}}$) lies between $\inL{\theta_0+(j'-1)3\cdot 3^{K-k+1}+1+3^{K-k}-3^{K-k+1}}{K-k+1}$ and $\inU{\theta_0+(j'-1)3\cdot 3^{K-k+1}+1+3^{K-k}-3^{K-k+1}}{K-k+1}$. 
		Since those edges have capacity $3$ their outflow then has the same form, just $\tau_{e'_{j'}} = 3^{K+1}+3^{K-k+1}-3^{K-k}$ time steps later. So, by induction, $B_{K,k-1}$ has the $(K,k-1,j',\theta_0+3^{K+1}+1)$-blocking property, which together proves the $(K,k,j)$-blocking property for $P^k_j$ ($j \in \set{j'3+1,j'3+2,j'3+3}$) for all $k'\geq 1$. Additionally all flow has left $P_{j'}^{k-1}$ by time $(3^K+1)+(k-2)(3^K+1)+(j''-1)3^K+14\cdot3^{K-1}+1$, where $j'' \in \set{1,2,3}$ such that $(j''-1)3^{k-1-1} < j' \leq j''3^{k-1-1}$ (i.e. $(j''-1)3^{k-1} < j' \leq (j''-1)3^{k-1}$) by induction and, thus, also $P_j^k$ by this time. Finally we have 
		\[\tau(P_j^k) = 1 + 3^{K+1}+3^{K-k+1}-3^{K-k} + \tau(P_{j'}^{k-1}) = (k-1)(3^{K+1}+1) + 1 - 3^{K-k} + 3^{K-1}\]
		and the total travel time of all edges inside $B_{K,k}$ is 
		\begin{align*}
		\tau(B_{K,k}) 	&= 3^{k-1}\tau(B_{K,1}) + \sum_{j'=1}^{3^{k-1}}\tau_{e'_{j'}} + \tau(B_{K,k-1}) \\
		&= 3^k + 3^{k-1}(1 + 3^{K+1} + 3^{K-k+1} - 3^{K-k}) \\
			&\quad\quad+\sum_{k'=1}^{k-1} \big(4\cdot 3^{k'-1} + (3^{k'+1}+2)3^{K-1}\big)-11\cdot 3^{K-1}-1\\
		&= \sum_{k'=1}^k \big(4\cdot 3^{k'-1} + (3^{k'+1}+2)3^{K-1}\big)-11\cdot 3^{K-1}-1.
		\end{align*}
		This concludes the construction of the blocking gadget $B_K \coloneqq B_{K,K}$.
		
		Next we define $U_{K,L} \coloneqq \frac{1}{2}\big((L-1)(1+3^{K+2}+2K)+ 4L3^{K+2} + 3^{K+2} + 1\big)$, which will be the total flow volume in $G_{K,L}$. Additionally $2U_{K,L}$ will be the maximal flow rate occurring in $G_{K,L}$. Now we construct the cycling gadget $C_K$ (see \Cref{fig:SlowTermGadgetC}): It consists of 
		\begin{itemize}
			\item nodes $u_j,u'_j,v_j,v'_j,w_j,w'_j,x_j$ for $j \in \set{1,2, \dots, 3^K}$ and $u_{3^K+1}$ and $x_{3^K+1}$,
			\item edges $u_jv_j$, $v_jw_j$ and $w_ju_{j+1}$ for $j \in \set{1, 2, \dots, 3^K}$, each with capacity $2U_{K,L}$ and travel time $1$,
			\item edges $u_ju'_j$, $v_jv'_j$ and $w_jw'_j$ for $j \in \set{1, 2, \dots, 3^K}$, each with capacity 3 and travel time $1$,
			\item edges $u'_jx_j$, $v'_jx_j$ and $w'_jx_j$ for $j \in \set{1, 2, \dots, 3^K}$, each with capacity 1 and travel time $1$,
			\item edges $u_{(j-1)3^k+1}u_{j3^k+1}$ for $k \in \set{0, 1, \dots, K-1}$, $j \in \set{1,2, \dots, 3^{K-k}}$, each with capacity $2U_{K,L}$ and travel time $3^k$
			\item and edges $u_{3^K+1}x_{3^K+1}$ with capacity $1$ and travel time $2$ and $u_{3^K+1}u_1$ with capacity $2U_{K,L}$ and travel time $1$.
		\end{itemize}
		
		The total physical travel time of all edges in $C_K$ is therefor $\tau(C_K) = 3\cdot 3^K+3\cdot 3^K + 3\cdot3^K + K3^{K+1} + 3 = (3+K)3^{K+1}+3$.
		
		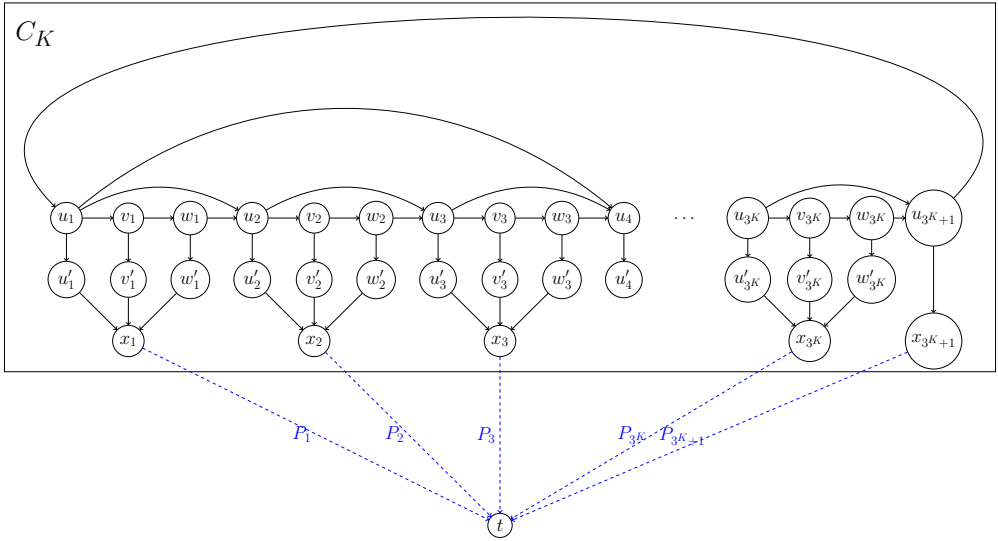
\begin{figure}[ht!]\centering
			\begin{adjustbox}{max width=.9\textwidth}
				\begin{tikzpicture}[font=\LARGE]
	\useasboundingbox (-2.5,7.5) rectangle (30.5,-10.5);

	\draw (-2,7) rectangle (30,-5);
	\node () at (-1,6) {\Huge $C_{K}$};
	
	\node[namedVertex](u1) at (0,0) {$u_1$};
	\node[namedVertex](v1) at (2,0) {$v_1$};
	\node[namedVertex](w1) at (4,0) {$w_1$};
	\node[namedVertex](u2) at (6,0) {$u_2$};
	\node[namedVertex](v2) at (8,0) {$v_2$};
	\node[namedVertex](w2) at (10,0) {$w_2$};
	\node[namedVertex](u3) at (12,0) {$u_3$};
	\node[namedVertex](v3) at (14,0) {$v_3$};
	\node[namedVertex](w3) at (16,0) {$w_3$};
	\node[namedVertex](u4) at (18,0) {$u_4$};
	\node() at (20,0) {$\dots$};
	\node[namedVertex](u3k) at (22,0) {$u_{3^K}$};
	\node[namedVertex](v3k) at (24,0) {$v_{3^K}$};
	\node[namedVertex](w3k) at (26,0) {$w_{3^K}$};
	\node[namedVertex](u3k1) at (28,0) {$u_{3^K+1}$};

	\node[namedVertex](u1s) at (0,-2) {$u_1'$};
	\node[namedVertex](v1s) at (2,-2) {$v_1'$};
	\node[namedVertex](w1s) at (4,-2) {$w_1'$};
	\node[namedVertex](u2s) at (6,-2) {$u_2'$};
	\node[namedVertex](v2s) at (8,-2) {$v_2'$};
	\node[namedVertex](w2s) at (10,-2) {$w_2'$};
	\node[namedVertex](u3s) at (12,-2) {$u_3'$};
	\node[namedVertex](v3s) at (14,-2) {$v_3'$};
	\node[namedVertex](w3s) at (16,-2) {$w_3'$};
	\node[namedVertex](u4s) at (18,-2) {$u_4'$};

	\node[namedVertex](u3ks) at (22,-2) {$u_{3^K}'$};
	\node[namedVertex](v3ks) at (24,-2) {$v_{3^K}'$};
	\node[namedVertex](w3ks) at (26,-2) {$w_{3^K}'$};

	\node[namedVertex](x1) at (2,-4) {$x_1$};
	\node[namedVertex](x2) at (8,-4) {$x_2$};
	\node[namedVertex](x3) at (14,-4) {$x_3$};

	\node[namedVertex](x3k) at (24,-4) {$x_{3^K}$};
	\node[namedVertex](x3k1) at (28,-4) {$x_{3^K+1}$};
	
	\node[namedVertex](t) at (14,-10) {$t$};
	
	\path[edge] (u1) -- (u1s);
	\path[edge] (v1) -- (v1s);
	\path[edge] (w1) -- (w1s);
	\path[edge] (u2) -- (u2s);
	\path[edge] (v2) -- (v2s);
	\path[edge] (w2) -- (w2s);
	\path[edge] (u3) -- (u3s);
	\path[edge] (v3) -- (v3s);
	\path[edge] (w3) -- (w3s);
	\path[edge] (u4) -- (u4s);
	\path[edge] (u3k) -- (u3ks);
	\path[edge] (v3k) -- (v3ks);
	\path[edge] (w3k) -- (w3ks);
	\path[edge] (u1) -- (v1);
	\path[edge] (v1) -- (w1);
	\path[edge] (w1) -- (u2);
	\path[edge] (u2) -- (v2);
	\path[edge] (v2) -- (w2);
	\path[edge] (w2) -- (u3);
	\path[edge] (u3) -- (v3);
	\path[edge] (v3) -- (w3);
	\path[edge] (w3) -- (u4);
	\path[edge] (u3k) -- (v3k);
	\path[edge] (v3k) -- (w3k);
	\path[edge] (w3k) -- (u3k1);

	\path[edge] (u1s) -- (x1);
	\path[edge] (v1s) -- (x1);
	\path[edge] (w1s) -- (x1);
	\path[edge] (u2s) -- (x2);
	\path[edge] (v2s) -- (x2);
	\path[edge] (w2s) -- (x2);
	\path[edge] (u3s) -- (x3);
	\path[edge] (v3s) -- (x3);
	\path[edge] (w3s) -- (x3);
	\path[edge] (u3ks) -- (x3k);
	\path[edge] (v3ks) -- (x3k);
	\path[edge] (w3ks) -- (x3k);

	\path[edge] (u3k1) -- (x3k1);	
	
	\path[edge] (u3k1) to [bend right=135] (u1);	
	
	\path[edge] (u1) to [bend left=30] (u2);
	\path[edge] (u2) to [bend left=30] (u3);
	\path[edge] (u3) to [bend left=30] (u4);
	\path[edge] (u3k) to [bend left=30] (u3k1);
	
	\path[edge] (u1) to [bend left=40] (u4);

	\path[edge,blue,dashed] (x1) --node[left]{$P_1$} (t);
	\path[edge,blue,dashed] (x2) --node[left]{$P_2$} (t);
	\path[edge,blue,dashed] (x3) --node[left]{$P_3$} (t);
	\path[edge,blue,dashed] (x3k) --node[left]{$P_{3^K}$} (t);
	\path[edge,blue,dashed] (x3k1) --node[left]{$P_{3^K+1}$} (t);
	
\end{tikzpicture}
			\end{adjustbox}
			\caption{The cycling gadget $C_K$. All horizontal edges (forming the cycle) have a capacity of $2U_{K,L}$, the vertical edges $u_ju'_j$, $v_jv'_j$ and $w_jw'_j$ have a capacity of $3$ and all other edges have a capacity of $1$. The paths $P_j$ connecting $C_K$ to the sink $t$ are not part of the gadget itself, but are required to be of equal length.}\label{fig:SlowTermGadgetC}
		\end{figure}
		
		We will now assume that the nodes $x_j$ from gadget $C_K$ are connected to the sink node $t$ via paths $P_j$ of equal length (see \Cref{fig:SlowTermGadgetC}) and start at some time $\theta_0$ with an inflow rate of $y\geq \max\set{3^{K+2}+1,2K}$ over the interval $[\theta_0-x,\theta_0]$ with $0.5 \geq x\geq \max\set{\frac{1}{y-1}+\sum_{i=1}^{3^{K+1}}\frac{3}{y-1-3i},\sum_{i=0}^{K-1}\frac{6\cdot3^i}{y-1-3i}}$ at node $u_{3^K+1}$ and an otherwise empty gadget $C_K$. We describe two possible evolutions of an IDE from here on out:
		
		If the paths $P_j$ have no waiting time between time $\theta_0+3j-3$ and $\theta_0+3j$ then the following flow dynamic inside $C_K$ is an IDE up to time $\theta_0+3^{K+1}$ (Charging-Phase):
		
		\begin{itemize}
			\item By \Cref{obs:SlowTermFlowSplit} the flow arriving at $u_{3^K+1}$ splits in such a way, that it enters the edge towards $u_1$ at a rate of $y-1$ over the interval $[\theta_0-x+\frac{1}{y-1},\theta_0]$ and arrives at $u_1$ over the interval $[\theta_0-x+\frac{1}{y-1}+1,\theta_0+1]$.
			\item For $j \in \set{1,2, \dots, 3^K}$: Flow arrives at node $u_j$ at a rate of $y-1-3^2(j-1)$ over the interval $[\theta_0-x+\frac{1}{y-1}+\sum_{i=1}^{3(j-1)}\frac{3}{y-1-3i}+3(j-1)+1,\theta_0+3(j-1)+1]$. This flow splits according to \Cref{obs:SlowTermFlowSplit} between the edges towards $u'_j$ and $v_j$: Namely, in such a way, that flow arrives at $u'_j$ at rate $3$ over the interval $[\theta_0-x+\frac{1}{y-1}+\sum_{i=1}^{3(j-1)}\frac{3}{y-1-3i}+3(j-1)+2,\theta_0+3(j-1)+3]$ and at $v_j$ at a rate of $y-1-3^2(j-1)-3$ over the interval $[\theta_0-x+\frac{1}{y-1}+\sum_{i=1}^{3(j-1)+1}\frac{3}{y-1-3i}+3(j-1)+2,\theta_0+3(j-1)+2]$. The same split then happens at $v_j$ and (one time step later) at $w_j$, so that finally over the interval $[\theta_0-x+\frac{1}{y-1}+\sum_{i=1}^{3j}\frac{3}{y-1-3i}+3j+1,\theta_0+3j+1]$ the flow reaches $u_{j+1}$ at a rate of $y-1-3^2j$.
		\end{itemize}	
		This flow then has the following two properties:
		\begin{itemize}
			\item For $j \in \set{1,2,\dots,3^K}$ the flow over the nodes $u'_j$, $v'_j$ and $w'_j$ arrives at $x_j$ at a rate between $\inL{\theta_0+3(j-1)+2}{0}$ and $\inL{\theta_0+3(j-1)+2}{0}$.
			\item Flow arrives again at node $u_{3^K+1}$ at a rate of $y-1-3^{K+2}$ over the interval $[\theta_0-x+\frac{1}{y-1}+\sum_{i=1}^{3^{K+1}}\frac{3}{y-1-3i}+3^{K+1}+1,\theta_0+3^{K+1}+1]$
		\end{itemize}
		If, on the other hand, the paths $P_j$, $j \in \set{1,2,\dots,3^K}$ exhibit the $(K,K,j,\theta_0-4)$-blocking property then the following flow evolution inside $C_K$ is an IDE up to $\theta_0+K(3^{K+1}+1)-1$ (Cycling-Phase):
		\begin{itemize}
			\item For $k \in \set{0,1,2,\dots,K-1}$
			\begin{itemize}
				\item Over the interval $[\theta_0+k(3^{K+1}+1)-x+\sum_{i=0}^{k-1}\frac{6\cdot 3^i}{y-1-3i},\theta_0+k(3^{K+1}+1)]$ flow arrives at $u_{3^K+1}$ with rate at least $y-2k$. Similar to \Cref{obs:SlowTermFlowSplit} at first all flow enters the edge towards $x_{3^K+1}$ until a queue of sufficient length has built up, such that the path over $u_1$ has the same instantaneous travel time. After this point, flow enters edge $u_{3^K+1}x_{3^K+1}$ at a rate of one more than the current change of waiting time on the path $P_1$ (which is at most $2$) and the rest of the flow enters the edge towards $u_1$. Since $P_1$ has the $(K,K,1,\theta_0-4)$-blocking property, the current waiting time is at most $6\cdot3^{k}$, so at most a flow volume of $6\cdot3^{k}+3x$ is lost. The rest of the flow arrives at node $u_1$ at least at a rate of $y-3(k+1)$ over the interval $[\theta_0+k(3^{K+1}+1)-x+\sum_{i=0}^{k}\frac{6\cdot 3^i}{y-1-3i}+1,\theta_0+k(3^{K+1}+1)+1]$.
				\item For $j \in \set{1,2,\dots,3^{K-k}-1}$: All flow arrives at node $u_{(j-1)3^k+1}$ within the interval $[\theta_0+k(3^{K+1}+1)-x+1+(j-1)3^{k+1},\theta_0+k(3^{K+1}+1)+1+(j-1)3^{k+1}]$. Since the paths $P_{j'}$ have the $(K,K,j',\theta_0-4)$-blocking property the path  $u_{(j-1)3^k+1}, u_{j3^k+1}, u'_{j3^k+1},$ $x_{j3^k+1}, P_{j3^k+1}$ has a shorter current travel time (current waiting time $\leq 0.5\cdot 3^k$) than $u_{(j-1)3^k+1},u'_{(j-1)3^k+1},x_{(j-1)3^k+1},P_{(j-1)3^k+1}$ or any of the paths inbetween (current waiting time $\geq 3.5\cdot 3^k$). Thus, the edge $u_{(j-1)3^k+1},u_{j3^k+1}$ is active and we can send all flow over this edge, so it will arrive at node $u_{j3^k+1}$ within the interval $[\theta_0+k(3^{K+1}+1)+1+j3^{k+1}-x,\theta_0+k(3^{K+1}+1)+1+j3^{k+1}]$.
				\item For $j = 3^{K-k}$: All flow arrives at node $u_{3^K-3^k+1}$ within the interval $[\theta_0+k(3^{K+1}+1)+1+(3^{K-k}-1)3^{k+1}-x,\theta_0+k(3^{K+1}+1)+1+(3^{K-k}-1)3^{k+1}]$. Since no new flow has arrived at $u_{3^K+1}$ since $\theta_0+k(3^{K+1}+1)$ the queue of at most $6\cdot 3^k \leq (3^{K-k}-1)3^{k+1}$ on the edge $u_{3^K+1}x_{3^K+1}$ has vanished by now. Thus, again, $u_{3^K-3^k+1},u_{3^K+1},x_{3^K+1},P_{3^K+1}$ has the currently shortest instantaneous travel time and we can send all flow over the edge towards $u_{3^K+1}$ where it will arrive within the interval $[\theta_0+(k+1)(3^{K+1}+1)-x,\theta_0+(k+1)(3^{K+1}+1)]$.
			\end{itemize} 
		\end{itemize}
		This flow then has the property, that flow finally arrives at $u_{3^K+1}$ with rate of at least $y-2K$ over at least the interval $[\theta_0+K(3^{K+1}+1)-x+\sum_{i=0}^{K-1}\frac{6\cdot 3^i}{y-1-3i},\theta_0+K(3^{K+1}+1)]$.
		
		We can now construct the network $G_{K,L}$ by connecting the gadgets $C_K$ and $B_K$ as indicated in \Cref{fig:SlowTermGraph}. 	The sum of all edge length in $G_{K,L}$ is then
		\begin{align*}
		\tau(G_{K,L}) 	&= \underbrace{\tau(C_K)}_{=(3+K)3^{K+1}+3} + \underbrace{\tau(B_K)}_{3^k+(1+3^K)\sum_{k'=1}^{k-1}3^{k'}} + \underbrace{\sum_{j=1}^{3^K}\tau_{e_j}}_{=3^K(3^{K+1}-5)} + \underbrace{\tau_{e_0}}_{=1} + \underbrace{\tau_{e_{3^K+1}}}_{=3^K-5 + (k-1)(3^K+1) + 1 +1}
		\end{align*}
		and, thus, in $\BigO(3^{2K})$.
		
		\begin{minipage}{.45\textwidth}
			\vspace{.5em}
			\begin{adjustbox}{max width=\textwidth}
				\begin{tikzpicture}
	\draw (0,1) rectangle (10,-1);
	\node () at (2,0.5) {\LARGE $C_{K}$};
	
	\node[namedVertex](x1) at (1,0) {$x_1$};
	\node[namedVertex](x2) at (3,0) {$x_2$};
	\node(xd) at (5,0) {$\dots$};
	\node[namedVertex](x3k) at (7,0) {$x_{3^K}$};
	\node[namedVertex](x3k1) at (9,0) {$x_{3^K+1}$};
	
	\node () at (4,-5) {\LARGE $B_{K}$};
	\draw (0,-3.5) -- (0,-2) -- (8,-2) -- (8,-3.5) -- (4,-8) --cycle;
	
	\node[namedVertex] (v1) at (1,-3) {$v_1$};
	\node[namedVertex] (v2) at (3,-3) {$v_2$};
	\node (vd) at (5,-3) {$\dots$};
	\node[namedVertex] (v3k) at (7,-3) {$v_{3^k}$};	
	
	\node[namedVertex] (v0) at (4,-6.5) {$v_0$};
	
	\path[edge,dashed] (v1) -- (v0);
	\path[edge,dashed] (v2) -- (v0);
	\path[edge,dashed] (v3k) -- (v0);
	
	\node[namedVertex](t) at (4,-9) {$t$};
	
	\path[edge] (x1) --node[left]{$e_1$} (v1);
	\path[edge] (x2) --node[left]{$e_2$} (v2);
	\path (xd) --node[left]{$\dots$} (vd);
	\path[edge] (x3k) --node[left]{$e_{3^K}$} (v3k);
	
	\path[edge] (v0) --node[left, near end]{$e_0$} (t);
	\path[edge] (x3k1) to [bend left=30] node[left]{$e_{3^K+1}$} (t);
	
\end{tikzpicture}
			\end{adjustbox}
			\vspace{.5em}
		\end{minipage}\hspace{.03\textwidth}%
		\begin{minipage}{.5\textwidth}
			
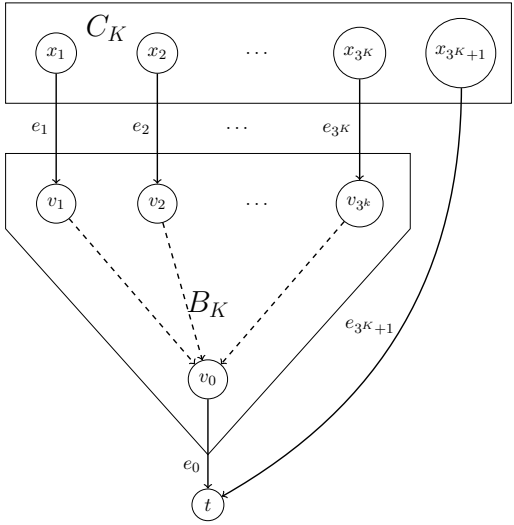
\captionof{figure}{The whole graph $G_{K,L}$ constructed by combining cycling gadget $C_K$ with blocking gadget $B_k$. The edges $e_j$ connecting both gadgets have a capacity of $3$ and a physical travel time of $3^{K+1}-5$. Edge $e_0$ connecting $B_K$ with the sink node $t$ has capacity $3$ and travel time $1$ and edge $e_{3^K+1}$ connecting node $x_{3^{K+1}+1}$ from gadget $C_K$ with $t$ has capacity $1$ and travel time equal to $\tau_{e_1}+\tau(P_1)+1$.}\label{fig:SlowTermGraph}
		\end{minipage}		
		
		Furthermore we define a constant inflow rate of $2U_{K,L} = (L-1)(1+3^{K+2}+2K)+ 4L3^{K+2} + 3^{K+2} + 1$ at node $u_{3^K+1}$ over the interval $[-0.5,0]$ (thus, $U_{K,L} \in \BigO(L3^K)$). For all $\ell \in \set{0,\dots, L-1}$, define
		\begin{itemize}
			\item $\theta_0^{\ell} \coloneqq \ell (K+1)(3^{K+1}+1))$,
			\item $y_{\ell} \coloneqq 2U_{K,L} -\ell(1+3^{K+2}+2K) \geq 4L3^{K+2} + 3^{K+2} + 1$ and
			\item $x_{\ell} \coloneqq \frac{1}{2} - \frac{\ell}{2L}$.
		\end{itemize} 
		Then for all $y \geq y_{L-1} = 4L3^{K+2} + 3^{K+2} + 1$ we have
		\begin{align}
			\frac{1}{y-1} &+ \sum_{i=1}^{3^{K+1}}\frac{3}{y-1-3i} + \sum_{i=0}^{K-1}\frac{6\cdot 3^i}{y-1-3i} 
				\leq \frac{1}{4L3^{K+2}} + \sum_{i=1}^{3^{K+1}}\frac{3}{4L3^{K+2}} + \sum_{i=0}^{K-1}\frac{6\cdot 3^i}{4L3^{K+2}} \notag\\		
				&= \frac{1}{4L3^{K+2}}(1+3\cdot 3^{K+1} + 6\frac{3^K-1}{3-1}) \leq \frac{2\cdot 3^{K+1}}{4L3^{K+2}} = \frac{1}{2L}. \label{eq:BoundForIntervalReduction}
		\end{align}
		We can now describe an IDE up to time $L(3^{K+1}+1+K(3^{K+1}+1))$ as follows: 
		\begin{itemize}
			\item At all times $\theta_0^{\ell}$ gadget $C_K$ is at the start of a Charging-Phase. Flow arrives at node $u_{3^K+1}$ within the interval $[\theta_0^{\ell}-0.5,\theta_0^{\ell}]$ and at a rate of at least $y_{\ell}$ over the interval $[\theta_0^{\ell}-x_{\ell},\theta_0^{\ell}]$. Thus, for all $j \in \set{1,2,\dots,3^K}$ emerges from edge $e_j$ at a rate between $\inL{\theta_0^{\ell}+3(j-1)-3+3^{K+1}}{0}$ and $\inU{\theta_0^{\ell}+3(j-1)-3+3^{K+1}}{0}$. Therefore the paths $P_j$ inside gadget $B_K$ will exhibit the $(K,K,j,\theta_0^{\ell}+3^{K+1}+1-4)$-blocking property. 
			\item At times $\theta_0^{\ell}+3^{K+1}+1$ gadget $C_K$ is at the start of a Cycling-Phase. Flow arrives at node $u_{3^K+1}$ within the interval $[\theta_0^{\ell}+3^{K+1}+1-0.5,\theta_0^{\ell}+3^{K+1}+1]$ and at a rate of at least $y_{\ell} -1-3^{K+2}$ over the interval $[\theta_0^{\ell}+3^{K+1}+1-x_{\ell}+\frac{1}{y-1}+\sum_{i=1}^{3^{K+1}}\frac{3}{y-1-3i},\theta_0^{\ell}+3^{K+1}+1]$. Thus, flow will again arrive at node $u_{3^K+1}$ within the interval $[\theta_0^{\ell}+3^{K+1}+1+3^K+K(3^{K+1}+1)-0.5,\theta_0^{\ell}+3^{K+1}+1+K(3^{K+1}+1)]$ and at a rate of at least $y_{\ell}-1-3^{K+2}-2K=y_{\ell+1}$ over the interval $[\theta_0^{\ell}+3^{K+1}+1+3^K+K(3^{K+1}+1)-x_{\ell}+\frac{1}{y-1}+\sum_{i=1}^{3^{K+1}}\frac{3}{y-1-3i}+\sum_{i=0}^{K-1}\frac{6\cdot 3^i}{y-1-3i},\theta_0^{\ell}+3^{K+1}+1+3^K+K(3^{K+1}+1)] \overset{\text{\eqref{eq:BoundForIntervalReduction}}}{\supseteq} [\theta_0^{\ell+1}-x_{\ell+1},\theta_0^{\ell+1}]$.
		\end{itemize}	
		Extending this flow to an IDE flow on $\IR_{\geq 0}$ (which is always possible by \cite[Lemma 5.3 and Theorem 5.5]{GHS18}), gives us an IDE flow, that does not terminate before $\theta_0^L \coloneqq \theta_0^{\ell} \coloneqq L(K+1)(3^{K+1}+1)$. This IDE flow also has the property that at time $\theta_0^{\ell}$ there is a flow volume of at least $x_{\ell} y_{\ell}$ still in the network (the flow that just arrived at node $u_{3^K+1}$), i.e. $\edgeLoad(\theta) \geq x_{\ell} y_{\ell}$ for all $\theta \in [\theta_0^{\ell-1},\theta_0^{\ell}]$. Thus, we get the following lower bound on the total travel times:
			\begin{align*}
				\sumOfTraveltimes(f) 
					&\overset{\text{\Cref{lemma:SumOfTravelTimesDefWithEdgeLoad}}}{\geq}
						\sum_{\ell=1}^{L-1}(\theta_0^{\ell}-\theta_0^{\ell-1})\cdot x_{\ell}\cdot y_{\ell} 
						\geq \sum_{\ell=1}^{L-1}K3^{K+1}\cdot(L-1-\ell)3^{K+2}\cdot \frac{L-\ell}{2L}\\
				 	&= K3^{2K+3}\frac{1}{2L}\sum_{\ell'=1}^{L-1}(\ell'-1)\ell' = K3^{2K+3}\frac{(L-1)^3-(L-1)}{2L}\qedhere
			\end{align*}
	\end{proof}

In order to derive from \Cref{ex:SlowTerminationExample} a lower bound on the price of anarchy, we will slightly modify the instance used there in such a way that the makespan of the worst case IDE flow remains asymptotically the same, while there exists an optimal makespan in $\bigO(1)$ and optimal total times in $\BigO(U)$. 

\begin{theorem}\label{thm:PoALowerBound}
	For any pair of positive integers $(U,\tau)$ with $U \geq 2\tau$, we have $\MPoA(U,\tau) \in \bigOm(U\log \tau)$ as well as $\TPoA(U,\tau) \in \bigOm(U\log \tau)$.
\end{theorem}

\begin{proof}
	Wlog assume that there exist positive integers $L \geq 3^K$ such that $U = (L+3^K)3^K$ and $\tau = 3^{2K}$. Then we take the graph $G_{K,L}$ constructed in the proof of \Cref{thm:SlowTermination} and modify it by adding two new vertices $s$ and $v$ and four new edges as indicating in \Cref{fig:PoA-LowerBound-Graph}.	
		
	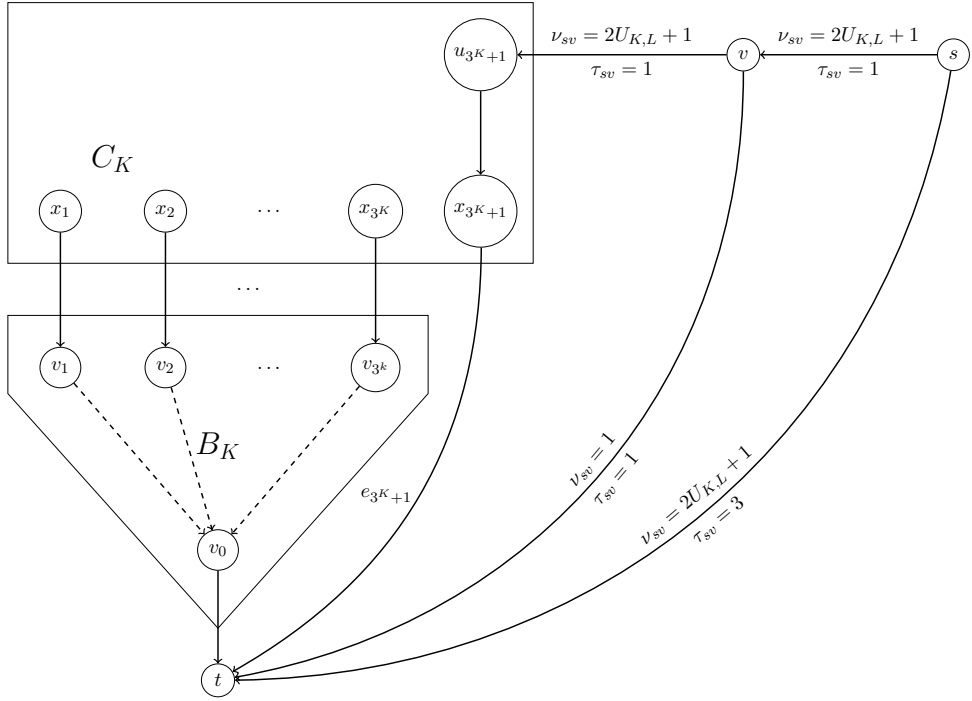
\begin{figure}[ht!]\centering
		\begin{adjustbox}{max width=.85\textwidth}
			\begin{tikzpicture}
	\draw (0,4) rectangle (10,-1);
	\node () at (2,1) {\LARGE $C_{K}$};
	
	\node[namedVertex](x1) at (1,0) {$x_1$};
	\node[namedVertex](x2) at (3,0) {$x_2$};
	\node(xd) at (5,0) {$\dots$};
	\node[namedVertex](x3k) at (7,0) {$x_{3^K}$};
	\node[namedVertex](x3k1) at (9,0) {$x_{3^K+1}$};
	
	\node[namedVertex](u3k1) at (9,3) {$u_{3^K+1}$};
		
	\node () at (4,-4.5) {\LARGE $B_{K}$};
	\draw (0,-3.5) -- (0,-2) -- (8,-2) -- (8,-3.5) -- (4,-8) --cycle;
	
	\node[namedVertex] (v1) at (1,-3) {$v_1$};
	\node[namedVertex] (v2) at (3,-3) {$v_2$};
	\node (vd) at (5,-3) {$\dots$};
	\node[namedVertex] (v3k) at (7,-3) {$v_{3^k}$};	
	
	\node[namedVertex] (v0) at (4,-6.5) {$v_0$};
	
	\path[edge,dashed] (v1) -- (v0);
	\path[edge,dashed] (v2) -- (v0);
	\path[edge,dashed] (v3k) -- (v0);
	
	\node[namedVertex](t) at (4,-9) {$t$};
	
	\path[edge] (x1) -- (v1);
	\path[edge] (x2) -- (v2);
	\path (xd) --node[left]{$\dots$} (vd);
	\path[edge] (x3k) -- (v3k);
	
	\path[edge] (v0) -- (t);
	\path[edge] (x3k1) to [bend left=30] node[left]{$e_{3^K+1}$} (t);
	
	\node[namedVertex](v) at (14,3) {$v$};
	\node[namedVertex](s) at (18,3) {$s$};
	
	\path[edge] (u3k1) -- (x3k1);

	\path[edge] (s) -- node[above]() {$\nu_{sv}=2U_{K,L}+1$}node[below]() {$\tau_{sv}=1$} (v);
	\path[edge] (v) -- node[above]() {$\nu_{sv}=2U_{K,L}+1$}node[below]() {$\tau_{sv}=1$} (u3k1);
	
	\path[edge] (s) to[bend left=40] node[above,sloped]() {$\nu_{sv}=2U_{K,L}+1$}node[below,sloped]() {$\tau_{sv}=3$} (t);
	\path[edge] (v) to[bend left=40] node[above,sloped]() {$\nu_{sv}=1$}node[below,sloped]() {$\tau_{sv}=1$} (t);
	
\end{tikzpicture}
		\end{adjustbox}
		\caption{The instance from the proof of \Cref{thm:SlowTermination} (on the left), two additional nodes $s$ and $v$ as well as four new edges. 
				If flow enters the network at node $s$ at a rate of at most $2U_{K,L}+1$ before time $\theta=0$, the whole network can be cleared before time $\theta=3$, if all flow particles use the direct edge $st$ towards the sink $t$. In an IDE flow, however, this is not possible as this edge is not active at the start (since the path $s,v,t$ is strictly shorter). Thus, in an IDE flow all flow particles will travel towards $v$. Since the capacity of edge $vt$ is very small, a queue will built up on this edge and divert most of the flow further to the left into node $u_{3^K+1}$ of the cycling gadget $C_K$.
			}\label{fig:PoA-LowerBound-Graph}
	\end{figure}

	We now use $s$ as our new source node and a constant inflow rate of $2U_{K,L}+1$ over the interval $[-0.5-\varepsilon,0]$, where $\varepsilon \coloneqq \frac{4+\tau_{e_{3^K+1}}}{2U_{K,L}} \leq \frac{1}{2}$. The optimal makespan of the resulting network $\mathcal{N}$ is then at most $3$, since we can just send all flow on the direct edge from $s$ to $t$, where the last particle arrives $3$ time steps later. Consequently, the optimal total travel times are bounded by $3\cdot (2U_{K,L}+1)$.
	
	In an IDE flow, however, all flow particles travel first to $v$ (as the path $s$-$v$-$t$ is shorter than the path $s$-$t$), where they enter the edge towards $t$. This continues until a queue of length $\tau_{e_{3^K+1}}+3$ has built up at edge $vt$ and the edge towards $u_{3^K+1}$ becomes active, which happens at time $\theta=0.5$. From there on, flow splits between the two edges, entering edge $vt$ at a rate of $1$ and edge $vu_{3^K+1}$ at a rate of $2U_{K,L}$ throughout the interval $[0.5,1]$. Thus, over the interval $[1.5,2]$ flow arrives at node $u_{3^K+1}$ at a rate of $2U_{K,L}$. Continuing with the flow described before the statement of \Cref{thm:SlowTermination} is then again an IDE flow. This shows that $\termTime[IDE](\mathcal{N}) \geq L(K+1)(3^{K+1}+1)$ and $\sumOfTraveltimes[IDE](\mathcal{N}) \geq K3^{2K+3}\frac{1}{2L}\big((L-1)^3-(L-1)\big)$.
	
	From $\tau(\mathcal{N}) = \tau(G_{K,L}) + 4 \in \bigO(3^{2K})$ and $U_{\mathcal{N}} = \Big(\frac{1}{2}+\varepsilon\Big)(2U_{K,L} + 1)  \in \bigO(L3^K + 3^{2K}) = \bigO(L3^K)$
	we get $3^{2K} \in \bigOm(\tau(\mathcal{N}))$ and $L3^K \in \bigOm(U_{\mathcal{N}})$, which implies
		$\termTime[IDE](\mathcal{N}) \geq LK(3^K+1) \in \bigOm\big(U_{\mathcal{N}}\log\tau(\mathcal{N})\big)$.
	Thus, in particular, 
		$\MPoA(U,\tau) \geq \frac{\termTime[IDE](\mathcal{N})}{\termTime[OPT](\mathcal{N})} \in \bigOm(U\log\tau)$.
	Similarly, for the total travel times we get
		$\TPoA(U,\tau) \geq \frac{\sumOfTraveltimes[IDE](\mathcal{N})}{\sumOfTraveltimes[OPT](\mathcal{N})} \in \bigOm\Big(\frac{K3^{2K}L^2}{(L+3^K)3^K}\Big) = \bigOm(LK3^K) = \bigOm(U\log\tau)$.
\end{proof}

Expanding the network constructed in the proof of \Cref{thm:PoALowerBound} into an acyclic network yields in an instance with constant optimal makespan, but IDE makespan of $\tPmax \gg \tPmin$, where $\tPmin$ is the physical length of a shortest path from the source to the sink node. From this, we get the following tighter bounds for acyclic networks:

\begin{theorem}
	Denoting by $\MPoA\big|_{\mathrm{acyclic}}$ the makespan price of anarchy when only ranging over \emph{acyclic} networks we get
		\[\MPoA\big|_{\mathrm{acyclic}} \in  \bigOm(\tPmax) \cap \bigO(U + \tPmax).\]
	These bounds hold even for multi-commodity flows.
\end{theorem}

\begin{proof}
	We obtain the lower bound by expanding the network from the proof of \Cref{thm:PoALowerBound} into an acyclic network. The upper bound follows immediately from \Cref{lemma:TerminationAcyclicNetworksMulti}.
\end{proof}

\begin{remark}
	For multi-commodity flows in general networks the price of anarchy (both makespan and total travel times) is unbounded as there are instances wherein flow cycles forever (cf. \cite[Theorem 6.1]{GHS18}).
\end{remark}

%!TEX root = ../article-PoA.tex

\section{Conclusions and Open Questions}

We studied the efficiency of IDE flows and derived the first upper and lower bounds on the price of anarchy of IDE flows with respect to two different objectives (makespan and total travel times). These bounds are of order $\bigO(U\tau)$ and $\Omega(U\log \tau)$, respectively. Comparing these bounds to the (conjectured) constant bound of $\frac{e}{e-1}$ for dynamic equilibria in the full information setting (cf.~Correa, Cristi and Oosterwijk~\cite{CCO19})  shows in some sense a ``price of shortsightedness''.  While  instantaneous dynamic equilibria may be significantly less efficient than dynamic equilibria, in many situations this might be a price one has to pay as the full information needed for dynamic equilibria might just not be available. 

Generally, it would  be interesting to test the different equilibria on real instances and see how their efficiency compares there. A large-scale computational study seems more feasible for IDEs compared to  dynamic equilibria, as already calculating a single $\alpha$-extension is much more difficult for the full information model, while it is easy
for IDE flows using a simple water-filling procedure (see \cite{GHS18}). Indeed,
for  calculating a single extension phase the only positive result 
for dynamic equilibria is based  on a recent work of  Kaiser
\cite{Kaiser20} showing that  for series-parallel graphs a single phase can be computed
in polynomial time. The question of whether a finite number of such extensions is enough to compute a complete equilibrium flow is still open for dynamic equilibria, while we were able to answer this question positively for IDE flows in \cite{GH20Algo}.

Another promising topic for future research is to find out whether our upper bound results can also be applied to different flow models. Those could be extensions of the Vickrey bottleneck model studied in this paper (e.g. adding spillback as in \cite{Sering2019}) or different ways of modelling congestion (e.g. by using differential equations as in \cite{BoyceRL95}). We hope that the very modular structure of our proofs in \cref{sec:UpperBounds} (as represented in \Cref{fig:LogicalStructureSec3}) should make it fairly easy to check which parts still apply to a different model and which have to be adjusted. More concretely, there are two main places at which the special properties of the Vickrey bottleneck model come into play: First, we need some lower bound on the amount of flow volume required to divert other flow away from a physically shortest paths (which is guaranteed by \cite[Lemma 4.4]{GHS18} in our case). This property most likely holds in many different flow models, maybe just with different values for the lower bound leading to different constants in the derived upper bounds on makespan and total travel times. Second, in the proofs of \Cref{lemma:FlowOnEdgesLeavesFastVar,lemma:TerminationAcyclicNetworksV2} we make use of the exact way queues operate in the Vickrey model (i.e. governed by \cref{eq:FeasibleFlow-QueueOpAtCap}) in order to provide upper bounds on how long flow volume can stay on an edge or to the left of any cut in an acyclic network, respectively. Again, it seems rather intuitive that many flow models will exhibit such a property guaranteeing similar upper bounds. The rest of the proofs then just relies on these three lemmas as well as fundamental properties of flows like flow conservation and, thus, should be quite universally applicable.

\section*{Acknowledgments} 

An extended abstract of this paper appeared in the proceedings of the
International Conference on Web and Internet Economics
WINE 2020: Web and Internet Economics, pp 237-251.

The research of the authors was funded by the Deutsche Forschungsgemeinschaft (DFG, German Research Foundation) - HA 8041/1-1 and HA 8041/4-1.

We thank Kathrin Gimmi for sparking the initial idea that allowed us to prove the explicit upper bound on the termination time of IDE flows. We are also grateful to the anonymous reviewers for their valuable feedback.

\clearpage

\bibliographystyle{plainnat}
\bibliography{literature}

\clearpage
\appendix
%!TEX root = ../article-PoA.tex

\section{List of symbols}

\begin{longtable}{lp{2.5cm}p{9.3cm}}
	\textbf{Symbol}			& \textbf{Name}
		& \textbf{Description} \\\hline\endhead
	$G = (V,E)$							& directed graph
		& a directed graph with node set $V$ and edge set $E$ \\
	$\delta_v^+ \subseteq E$			& outgoing edges
		& the set of edges leaving node $v$, i.e. $\delta_v^+ \coloneqq \set{e \in E | e = vw}$ \\ 
	$\delta_v^- \subseteq E$			& incoming edges
		& the set of edges entering node $v$, i.e. $\delta_v^- \coloneqq \set{e \in E | e = uv}$ \\
	$\delta_W^+ \subseteq E$			& outgoing edges
		& the set of edges leaving the subgraph induced by $W$, i.e. $\delta_W^+ \coloneqq \set{wv \in E | w \in W, v \in V\setminus W}$ \\ 
	$\delta_W^- \subseteq E$			& incoming edges
		& the set of edges entering the subgraph induced by $W$, i.e. $\delta_W^- \coloneqq \set{vw \in E | v \in V\setminus W, w \in W}$ \\
	$\tau_e \in \INs$ 				& physical travel time 
		& the physical travel time on edge $e$\\
	$\nu_e \in \INs$ 				& edge capacity
		& the capacity of edge $e$\\
	$t \in V$					& sink (node)
		& the common destination of all particles \\
	$\theta \in \IR_{\geq 0}$ 				& time
		& the time is usually denoted by $\theta$ \\
	$u_v(\theta)\in \IR_{\geq 0}$			& network inflow rate
		& inflow rate of particles at node $v$ at time $\theta$ \\
	$U_v(\theta)\in \IR_{\geq 0}$			& cumulative network inflow
		& the cumulative network inflow at node $v$ up to time $\theta$, i.e. $U_v(\theta) \coloneqq \int_{0}^{\theta} u_v(\zeta)d\zeta$ \\
	$\network$								& network 
		& a network $\network =(G,(\nu_e)_{e \in E},(\tau_e)_{e \in E},(u_v)_{v \in V\setminus\set{t}}, t)$ consisting of a graph, edge travel times, edge capacities, network inflow rates and a designated sink node \\
	$f = (f^+,f^-)$			& flow
		& a flow given by a tuple of edge inflow rates $f^+$ and edge outflow rates $f^-$\\
	$f^+_{e}(\theta)\in \IR_{\geq 0}$ 	& (edge) inflow rate
		& the rate at which particles of flow $f$ enter edge $e$ at time $\theta$ \\
	$f^-_{e}(\theta)\in \IR_{\geq 0}$ 	& (edge) outflow rate
		& the rate at which particles of flow $f$ leave edge $e$ at time $\theta$\\
	$F^+_{e}(\theta) \in \IR_{\geq 0}$ & cumulative (edge) inflow 
		& the total volume of flow having entered edge $e$ up to time $\theta$, i.e. $F_{e}^+(\theta) \coloneqq \int_{0}^{\theta}f_{e}^+(\zeta)d\zeta$ \\ 
	$F^-_{e}(\theta) \in \IR_{\geq 0}$ & cumulative (edge) outflow 
		& the total volume of flow having left edge $e$ up to time $\theta$, i.e. $F_{e}^-(\theta) \coloneqq \int_{0}^{\theta}f_{e}^-(\zeta)d\zeta$\\
	$\edgeLoad[e](\theta) \in \IR_{\geq 0}$ & edge load 
		& the flow volume on edge $e$ at time $\theta$, i.e. $\edgeLoad[e](\theta) \coloneqq F^+_e(\theta) - F^-_e(\theta)$\\
	$\edgeLoad(\theta) \in \IR_{\geq 0}$ & total flow volume 
		& the total volume of flow in the network at time $\theta$, i.e. $\edgeLoad(\theta) \coloneqq \sum_{e \in E}\edgeLoad[e](\theta)$\\
	$Z(\theta) \in \IR_{\geq 0}$ & flow volume at the sink
		& the total volume which has arrived at the sink by time $\theta$, i.e. $Z(\theta) \coloneqq \sum_{e \in \edgesEntering{t}}F^-_e(\theta) - \sum_{e \in \edgesLeaving{t}}F^+_e(\theta)$\\
	$q_e(\theta) \in \IR_{\geq 0}$ 			& queue length 
		& the length of the queue on edge $e$ at time $\theta$, defined as $q_e(\theta) \coloneqq F^+_e(\theta)-F^-_e(\theta+\tau_e)$ \\
	$c_e(\theta) \in \IR_{>0}$			& instantaneous travel time 
		& the current or instantaneous travel time at time $\theta$ over edge $e$, defined as $c_e(\theta) \coloneqq \tau_e + q_e(\theta)/\nu_e$ \\
	$\ell_{v}(\theta) \in \IR_{\geq 0}$	& node labels
		& the current distance from $v$ to $t$ at time $\theta$, i.e. the length of a shortest $v$-$t$ path w.r.t. $c_e(\theta)$. \\
	$E_{\theta} \subseteq E$ 			& active edges 
		& the set of all edges active at time $\theta$, i.e. \newline $E_{\theta} \coloneqq \set{vw \in E | \ell_{ v}(\theta) = \ell_{w}+c_{vw}(\theta)}$ \\ 
	$\theta_0 \in \IR_{\geq 0}$				& end of network inflow 
		& a time such that all network inflow rates are $0$ after $\theta_0$ \\
	$\termTime(f)$			& makespan 
		& the first time after $\theta_0$ with no flow volume inside the network (see \Cref{def:makespan}) \\
	$\sumOfTraveltimes(f)$ 	& total travel times 
		& the sum of all particles' travel times in $f$ (see \Cref{def:totalTravelTime})\\
	$\termTime[IDE](\network)$			& worst case IDE makespan 
		& the worst case makespan of an IDE in $\network$ i.e. $\termTime[IDE](\mathcal{N}) \coloneqq \sup\set{\termTime(f) | f \text{ an IDE flow in } \mathcal{N}}$ \\
	$\termTime[OPT](\network)$			& optimal makespan 
		& the optimal makespan of any feasible flow in $\network$ i.e. $\termTime[OPT](\mathcal{N}) \coloneqq \inf\set{\termTime(f) | f \text{ a feasible flow in } \mathcal{N}}$ \\
	$\sumOfTraveltimes[IDE](\network)$ 	& worst case IDE total travel times
		& the worst case total travel times of an IDE in $\network$, i.e. $\sumOfTraveltimes[IDE](\mathcal{N}) \coloneqq \sup\set{\sumOfTraveltimes(f) | f \text{ an IDE flow in } \mathcal{N}}$\\
	$\sumOfTraveltimes[OPT](\network)$ 	& optimal total travel times
		& the optimal total travel times of any feasible flow in $\network$, i.e. $\sumOfTraveltimes[OPT](\mathcal{N}) \coloneqq \inf\set{\sumOfTraveltimes(f) | f \text{ a feasible flow in } \mathcal{N}}$\\
	$\MPoA(U,\tau)$								& makespane price of anarchy (PoA)
		& the worst case ratio of the makespan of an IDE to the optimal makespan in a network with total flow volume $U$ and total edge length $\tau$ (see \Cref{def:PoA}) \\
	$\TPoA(U,\tau)$								& total travel times PoA
		& the worst case ratio of the total travel times of an IDE to the optimal total travel times in a network with total flow volume $U$ and total edge length $\tau$ (see \Cref{def:PoA}) \\
	$\bigO(g)$					& big O 
		& the set of all functions growing asymptotically at most as fast as $g$ \\
	$\bigOm(g)$					& big omega 
		& the set of all functions growing asymptotically at least as fast as $g$ \\
\end{longtable}

%!TEX root = ../article-PoA.tex
\section{Omitted Proofs}\label{appendix}

\begin{applemma}{\ref{lemma:SumOfTraveltimesAltDef}}
	For any terminating feasible flow we have
	\[\sumOfTraveltimes(f) = \int_{0}^{\termTime(f)}\theta\cdot Z'(\theta)d\theta - \sum_{v \in V\setminus\{t\}}\int_{0}^{\theta_0} \theta \cdot u_v(\theta)d\theta = \int_0^{\termTime(f)} F^\Delta(\theta)d\theta.\]
\end{applemma}

\begin{proof}
	We start with the first equality. Using a telescope sum argument and flow conservation this essentially reduces to the following equality for individual edges.
			
	\begin{claim}\label{claim:SumOfTravelTimesAltDefEdgeWise}
		For any edge $e$ we have
		\[\int_0^{\termTime(f)}c_e(\theta)f_e^+(\theta)d\theta = \int_0^{\termTime(f)}(\theta+\tau_e)f_e^-(\theta+\tau_e) + \theta f_e^+(\theta)d\theta.\]
	\end{claim}

	\begin{proofClaim}
		The proof of this equality is basically a calculation using integration by parts together with several elementary properties of feasible flows. These properties are
		\begin{enumerate}[label=(\alph*)]
			\item For any time $\theta$ we have $F^+_e(\theta) = F^-_e(\theta+c_e(\theta))$.\label{prop:EqualityOfCumInandOutflow}
			\item The functions $F_e^+$, $F_e^-$, $q_e$ and $c_e$ are differentiable almost everywhere.
			\item For almost all $\theta$ we have
			$q'_e(\theta) = \begin{cases}
				f^+_e(\theta)-\nu_e, &\text{ if } q_e(\theta) > 0 \\
				\max\Set{f^+_e(\theta)-\nu_e,0}, &\text{ if } q_e(\theta) = 0
			\end{cases}$.\label{prop:DerivativeOfQueueLength}
			\item Given two times $a < b$ with empty queue, we have
			$\int_a^b q_e(\theta)q_e'(\theta)d\theta = 0$.\label{prop:IntegralOfQueueTimesQueueDerIsZero}
		\end{enumerate}
		Proofs of the first three properties can be found in \cite[Lemma 3.1]{SeringThesis}.\footnote{Note, that in \cite{SeringThesis}, queues are placed at the end of the edge (instead of at the beginning like in our model) and their length is denoted by $z_e$. The notation $q_e(\theta)$ is used there to denote the waiting time for particles entering edge $e$ at time $\theta$ and, thus, corresponds to $\frac{q_e(\theta)}{\nu_e}$ in our notation.}
		Property \ref{prop:IntegralOfQueueTimesQueueDerIsZero} can be shown using integration by parts:
		\[0 = \left[q_e(\theta)^2\right]_a^b = \int_a^b q_e'(\theta) q_e(\theta) d\theta + \int_a^b q_e(\theta) q_e'(\theta) d\theta = 2\int_a^b q_e(\theta) q_e'(\theta) d\theta.\]
		Using these properties we can now show the desired equality:
		\begin{align*}
			&\int_0^{\termTime(f)}c_e(\theta)f_e^+(\theta)d\theta \\
			&\quad= \left[c_e(\theta)F^+_e(\theta)\right]_0^{\termTime(f)} - \int_0^{\termTime(f)}c_e'(\theta)F_e^+(\theta)d\theta  \\
			&\quad= \left[\tau_e F_e^+(\theta)\right]_0^{\termTime(f)} -  \int_0^{\termTime(f)}\frac{q_e'(\theta)}{\nu_e}F_e^+(\theta)d\theta \\
			&\quad = \left[(\theta+\tau_e) F_e^+(\theta) - \theta F_e^+(\theta)\right]_0^{\termTime(f)} - \left[\frac{q_e(\theta)}{\nu_e}F_e^+(\theta)\right]_0^{\termTime(f)} + \int_0^{\termTime(f)}\frac{q_e(\theta)}{\nu_e}f_e^+(\theta)d\theta \\
			&\quad \overset{\text{\ref{prop:EqualityOfCumInandOutflow}}}{=} \left[(\theta+\tau_e) F_e^-(\theta+\tau_e)\right]_0^{\termTime(f)} - \left[\theta F_e^+(\theta)\right]_0^{\termTime(f)} - 0 + \frac{1}{\nu_e}\int_0^{\termTime(f)}q_e(\theta)f_e^+(\theta)d\theta \\
			&\quad = \int_0^{\termTime(f)}(\theta+\tau_e)f_e^-(\theta+\tau_e)d\theta + \int_0^{\termTime(f)}F_e^-(\theta+\tau_e)d\theta - \int_0^{\termTime(f)}\theta f_e^+(\theta)d\theta \\
				&\quad\quad\quad- \int_0^{\termTime(f)} F_e^+(\theta)d\theta + \frac{1}{\nu_e}\int_0^{\termTime(f)}q_e(\theta)f_e^+(\theta)d\theta \\
			&\quad = \int_0^{\termTime(f)}(\theta+\tau_e)f_e^-(\theta+\tau_e) - \theta f_e^+(\theta)d\theta -  \int_0^{\termTime(f)}F_e^+(\theta)-F_e^-(\theta+\tau_e)d\theta \\
				&\quad\quad\quad+ \frac{1}{\nu_e}\int_0^{\termTime(f)}q_e(\theta)f_e^+(\theta)d\theta \\
			&\quad = \int_0^{\termTime(f)}(\theta+\tau_e)f_e^-(\theta+\tau_e) - \theta f_e^+(\theta)d\theta + \frac{1}{\nu_e}\int_0^{\termTime(f)}q_e(\theta)f_e^+(\theta) - \nu_e q_e(\theta)d\theta \\
			&\quad = \int_0^{\termTime(f)}(\theta+\tau_e)f_e^-(\theta+\tau_e) - \theta f_e^+(\theta)d\theta + \frac{1}{\nu_e}\int_0^{\termTime(f)}q_e(\theta)\left(f_e^+(\theta) - \nu_e\right)d\theta \\
			&\quad \overset{\text{\ref{prop:DerivativeOfQueueLength}}}{=} \int_0^{\termTime(f)}(\theta+\tau_e)f_e^-(\theta+\tau_e) - \theta f_e^+(\theta)d\theta + \frac{1}{\nu_e}\int_0^{\termTime(f)}q_e(\theta)q_e'(\theta)d\theta \\
			&\quad \overset{\text{\ref{prop:IntegralOfQueueTimesQueueDerIsZero}}}{=} \int_0^{\termTime(f)}(\theta+\tau_e)f_e^-(\theta+\tau_e) - \theta f_e^+(\theta)d\theta\qedhere
		\end{align*}
	\end{proofClaim}

	Using this claim and the fact that $f^+_e(\theta) = 0$ for all $\theta > \termTime(f)$, we can now deduce the first part of the lemma using flow conservation at the nodes:
	\begin{align*}
		\sumOfTraveltimes(f) &= \sum_{e \in E}\int_0^{\termTime(f)} c_e(\theta)f^+_e(\theta) d\theta \overset{\text{\Cref{claim:SumOfTravelTimesAltDefEdgeWise}}}{=} \sum_{e \in E}\int_0^{\termTime(f)} \theta f^-_e(\theta) d\theta - \int_0^{\termTime(f)} \theta f^+_e(\theta) d\theta \\
		&= \sum_{v \in V}\sum_{e \in \edgesEntering{v}}\int_0^{\termTime(f)} \theta f^-_e(\theta) d\theta - \sum_{v \in V}\sum_{e \in \edgesLeaving{v}}\int_0^{\termTime(f)} \theta f^+_e(\theta) d\theta \\
		&= \int_0^{\termTime(f)} \theta\left(\sum_{e \in \edgesEntering{t}} f^-_e(\theta) - \sum_{e \in \edgesLeaving{t}} f^+_e(\theta) \right)d\theta \\
			&\quad\quad+ \sum_{v \in V \setminus \set{t}}\int_0^{\termTime(f)} \theta\left(\sum_{e \in \edgesEntering{v}} f^-_e(\theta) - \sum_{e \in \edgesLeaving{v}} f^+_e(\theta) \right)d\theta \\
		&\overset{\mathclap{\text{\eqref{eq:FeasibleFlow-FlowConservation},\eqref{eq:DefZ}}}}{=}\,\, \sum_{e \in \edgesEntering{t}}\int_0^{\termTime(f)} \theta Z'(\theta) d\theta - \sum_{v \in V \setminus \set{t}}\int_0^{\termTime(f)} \theta u_v(\theta)d\theta \\
		&=\int_{0}^{\termTime(f)}\theta\cdot Z'(\theta)d\theta - \sum_{v \in V\setminus\{t\}}\int_{0}^{\theta_0} \theta \cdot u_v(\theta)d\theta
	\end{align*}
	where the last equality follows directly from the definition of $\theta_0$.
	
	For the second part of the lemma, we use the fact that $F^{\Delta}(\termTime(f)) = 0$ together with partial integration to obtain
	\begin{align*}
		0 = \left[F^\Delta(\theta)\theta\right]_0^{\termTime(f)} = \int_0^{\termTime(f)}{F^\Delta}'(\theta)\theta d\theta + \int_0^{\termTime(f)}{F^\Delta}(\theta)d\theta.
	\end{align*}
	From this we can show the second equality using \Cref{lem:GeIsTotalFlowInGraph}:
	\begin{align*}
		\int_0^{\termTime(f)}{F^\Delta}(\theta)d\theta &= -\int_0^{\termTime(f)}{F^\Delta}'(\theta)\theta d\theta = -\int_0^{\termTime(f)}\left(\sum_{v \in V\setminus{t}}U'_v(\theta) - Z'(\theta)\right)\theta d\theta \\
		&= \sum_{v \in V\setminus{t}}\int_0^{\termTime(f)} \theta \cdot u_v(\theta)d\theta - \int_0^{\termTime(f)} \theta \cdot Z'(\theta)d\theta.\qedhere
	\end{align*}
\end{proof}

\end{document}